\documentclass[sigconf,screen,anonymous=false]{acmart}
\usepackage{xspace,balance,tabularx,multirow}
\usepackage{flushend}
\usepackage{tikz}
\usepackage{pgfplots}
\pgfplotsset{compat=1.16}
\usetikzlibrary{patterns}
\usepackage{subfig}
\usepackage[ruled, vlined, linesnumbered]{algorithm2e}
\usepackage{xcolor}
\usepackage{colortbl}
\usepackage{bbold}
\SetKwComment{Comment}{$\triangleright$\ }{}
\usepackage{enumitem}
\usepackage{tablefootnote}
\usepackage{upgreek,textgreek}
\usepackage{pifont}%
\usepackage[noabbrev]{cleveref}
\usepackage{titlecaps}
\usepackage{lipsum}
\usepackage{makecell}
\usepackage{fancyhdr}

\pgfplotsset{every tick label/.append style={font=\tiny}}

\newlength{\starsize}
\newlength{\starspread}
\tikzset{starsize/.code={\setlength{\starsize}{#1}},
         starspread/.code={\setlength{\starspread}{#1}}}
\tikzset{starsize=1mm,
         starspread=3mm}
\pgfdeclarepatternformonly[\starspread,\starsize]%
  {my fivepointed stars}%
  {\pgfpointorigin}%
  {\pgfqpoint{\starspread}{\starspread}}%
  {\pgfqpoint{\starspread}{\starspread}}%
  {%
   \pgftransformshift{\pgfqpoint{\starsize}{\starsize}}
   \pgfpathmoveto{\pgfqpointpolar{18}{\starsize}}
   \pgfpathlineto{\pgfqpointpolar{162}{\starsize}}
   \pgfpathlineto{\pgfqpointpolar{306}{\starsize}}
   \pgfpathlineto{\pgfqpointpolar{90}{\starsize}}
   \pgfpathlineto{\pgfqpointpolar{234}{\starsize}}
   \pgfpathclose%
   \pgfusepath{fill}
  }

\makeatletter
\newcommand*\bigcdot{\mathpalette\bigcdot@{.5}}
\newcommand*\bigcdot@[2]{\mathbin{\vcenter{\hbox{\scalebox{#2}{$\m@th#1\bullet$}}}}}
\makeatother

\newcommand{\stitle}[1]{\vspace*{0.5em}\noindent{\bf \underline{#1.}\/}}

\newcommand{\eat}[1]{}

\SetKwBlock{MyBlock}{\empty}{\empty}

\newcommand{\U}{\mathcal{U}\xspace}

\newcommand{\I}{\mathcal{I}\xspace}
\newcommand{\G}{\mathcal{G}\xspace}

\newcommand{\EDG}{\mathcal{E}\xspace}

\newcommand{\B}{\mathcal{B}\xspace}

\newcommand{\algo}{\texttt{ASC}\xspace}
\newcommand{\kalgo}{$K$-\texttt{ASC}\xspace}
\newcommand{\kalgoc}{\texttt{C-KASC}\xspace}
\newcommand{\NMC}{\texttt{Na{\"{\i}}ve-MC}\xspace}

\newcommand{\update}[1]{{\color{black}{#1}}}
\newcommand{\updatefig}{\color{black}
\arrayrulecolor{black}}

\newenvironment{customlegend}[1][]{%
    \begingroup
    \csname pgfplots@init@cleared@structures\endcsname
    \pgfplotsset{#1}%
}{%
    \csname pgfplots@createlegend\endcsname
    \endgroup
}%

\def\addlegendimage{\csname pgfplots@addlegendimage\endcsname}

\makeatletter
\newcommand\footnoteref[1]{\protected@xdef\@thefnmark{\ref{#1}}\@footnotemark}
\makeatother

\let\oldnl\nl%
\newcommand{\nonl}{\renewcommand{\nl}{\let\nl\oldnl}}%

\SetKwComment{Comment}{/* }{ */}

\makeatletter 
\g@addto@macro{\@algocf@init}{\SetKwInOut{Parameter}{Parameters}} 
\makeatother

\definecolor{myred}{HTML}{fd7f6f}
\definecolor{myred_new}{HTML}{D8D8D8}
\definecolor{myred_new2}{HTML}{D7191C}
\definecolor{myblue}{HTML}{7eb0d5}
\definecolor{mygreen}{HTML}{b2e061}
\definecolor{mypurple}{HTML}{bd7ebe}
\definecolor{myorange}{HTML}{ffb55a}
\definecolor{myyellow}{HTML}{ffee65}
\definecolor{mypurple2}{HTML}{beb9db}
\definecolor{mypink}{HTML}{fdcce5}
\definecolor{mycyan2}{HTML}{a29bfe}
\definecolor{mycyan}{HTML}{B33771}

\definecolor{myblue2}{HTML}{115f9a}
\definecolor{myred2}{HTML}{c23728}

\definecolor{myblue}{HTML}{8e44ad}
\definecolor{my_blue}{HTML}{8e44ad}
\definecolor{my_cyan}{HTML}{d95319}
\definecolor{my_teal}{HTML}{77ac30}
\definecolor{my_purple}{HTML}{4dbeee} %
\definecolor{my_violet}{HTML}{4cd137}
\definecolor{my_green}{HTML}{008080}

\definecolor{my_red}{HTML}{CC79A7}
\definecolor{my_yellow}{HTML}{117733}
\definecolor{my_pp}{HTML}{785ef0}

\definecolor{NSCcol1}{HTML}{3867d6}
\definecolor{NSCcol2}{HTML}{aec7e8}
\definecolor{NSCcol3}{HTML}{ff7f0e}
\definecolor{NSCcol4}{HTML}{ffbb78}
\definecolor{NSCcol5}{HTML}{98df8a}

\definecolor{lighpurple}{HTML}{CBC3E3}

\definecolor{rubinred}{HTML}{D10056}

\definecolor{royalpurple}{HTML}{7851A9}

\AtBeginDocument{%
  }

\setcopyright{acmlicensed}
\copyrightyear{2018}
\acmYear{2027}
\acmDOI{XXXXXXX.XXXXXXX}
\acmConference[SIGMOD '27]
  {ACM SIGMOD/PODS International Conference on Management of Data}
  {June 13--19, 2027}
  {Huntington Beach, CA, USA}
\acmISBN{978-1-4503-XXXX-X/18/06}

\acmSubmissionID{642}

\title[Swing Algorithm]{Efficient Swing Computation for Retrieval in Large-Scale Recommender Systems}
\subtitle{Technical Report}

\begin{document}
\title{Efficient Swing Computation for Retrieval in Large-Scale Recommender Systems}

\author{Runhao Jiang}
\affiliation{%
  \institution{Hong Kong Baptist University}
  \country{Hong Kong SAR, China}
}
\email{csrhjiang@comp.hkbu.edu.hk}
\orcid{0009-0000-4841-9175}

\author{Renchi Yang}
\affiliation{%
  \institution{Hong Kong Baptist University}
  \country{Hong Kong SAR, China}
}
\email{renchi@hkbu.edu.hk}
\orcid{0000-0002-7284-3096}

\begin{abstract}
Given a user-item graph $\G$, a query item $v_q$ and \update{a target item $v_t$}, the Swing score $\textsf{sw}(v_q,v_t)$ of the item pair $(v_q,v_t)$ leverages the user-item-user interaction structure to evaluate their similarity. This measure is found to be highly effective in item-to-item (i2i) retrieval task and finds extensive applications in industrial-scale recommender systems.
However, existing solutions towards computing Swing scores are either prohibitively expensive due to their quadratic time complexity w.r.t. the item degree, or rely on truncation heuristics that yield unsatisfactory quality, rendering them impractical particularly on graphs with billions of interactions.

In this paper, we present \algo{} and \kalgo{}, two novel and efficient algorithms for approximate and top-$K$ Swing queries, to address the aforementioned limitations. Specifically, these algorithms provide rigorous theoretical guarantees in probabilistic relative and additive errors of Swing values. The basic idea of \algo{} is to combine two randomized algorithms, \texttt{GNS} and \texttt{USS}, in a simple yet non-trivial way to adaptively process high- and low-degree query items with minimal runtime cost. In particular, \kalgo{} offers practical efficiency and effectiveness for top-$K$ queries through a filter-refinement paradigm with carefully-designed heuristics. 
Extensive experiments over eight real datasets demonstrate that \algo{} and \kalgo{} can achieve orders of magnitude speed-up over competitors in terms of computational time while offering the same approximate and top-$K$ query result quality, and in particular, \kalgo{} is highly efficient on massive graphs including the billion-edge {\em Yambda} and {\em MAG} datasets.
\end{abstract}



\maketitle

\section{Introduction}

In practice, industrial-scale recommender systems typically adopt a multi-stage pipeline, fundamentally separated into a retrieval (or candidate generation) phase and a subsequent ranking phase~\cite{huang2025comprehensive}.
In the retrieval stage, {\em item-to-item} (i2i)~\cite{sarwar2001item} is arguably one of the most popular and widely deployed retrieval strategies, which seeks to produce a ranked list of items from a massive catalog given a query item.
The fundamental premise of i2i retrieval is that users who have interacted (e.g., clicked, liked, or purchased) with a specific item are likely to engage with other items sharing similar collaborative or semantic characteristics~\cite{linden2003amazon}.
\update{When user-item interactions are represented as a bipartite graph, this task naturally amounts to retrieving items according to a graph-derived proximity measure.}
In the past two decades, a variety of i2i retrieval approaches have been developed, including collaborative filtering~\cite{sarwar2001item,linden2003amazon,yang2020large}, embedding-based paradigms~\cite{barkan2016item2vec,wang2018billion,zhang2023contrastive}, and others~\cite{huang2025comprehensive,feng2025llm}.

Amid them, {\em Swing}~\cite{yang2020large} has gained considerable traction and seen extensive practical application across industrial platforms~\cite{zzh2022industrial,wan2023interval,yang2020large,luo2025trawl,ismailov2022two}.
\update{For instance, peer-reviewed studies have documented the use of Swing in Alibaba's cold-start~\cite{cao2022gift,bao2025grain,yang2022task} and multi-scenario~\cite{huan2023samd} recommendations, while Kuaishou integrates Swing into its live streaming gifting prediction~\cite{deng2024mmbee} and query suggestion~\cite{guo2026onesug}. In addition, Shopee constructs high-quality item graphs based on Swing for advertisement recommendation~\cite{nguyen2023lightsage}. Beyond these peer-reviewed studies, Swing is also used for query rewriting and personalized search~\cite{li2022query,wu2024hi}, short-video search~\cite{bao2024beyond}, advertising~\cite{xue2026generative}, and other services~\cite{xia2026qarm,feng2025llm,wang2026pi2i,xu2023multi}, as well as in building item graphs for GNN-based recommender systems~\cite{xue2025e2e}.}
Unlike traditional {\em item-based collaborative filtering} (itemCF)~\cite{linden2003amazon} that calculates local similarities between items based on merely user co-actions, 
{Swing} leverages the user-item-user interaction structure, referred to as ``swing'' structure, of the user-item bipartite graph $\G$, and evaluates all user pairs interacted with the given query item $v_q$ and target item $v_t\in \G$.
Formally, the Swing similarity between items $v_q$ and $v_t$ is defined as:
\begin{small}
\begin{equation*}
\textsf{sw}(v_q,v_t)=\sum_{u_i, u_j\in \U(v_q)\cap \U(v_t)}{\frac{1}{\alpha+|\I(u_i)\cap \I(u_j)|}},
\end{equation*}
\end{small}
where $\U(v_q)$ and $\I(u_i)$ denote the sets of neighbors of item $v_q$ and user $u_i$, respectively, and $\alpha$ acts as a smoothing coefficient.
The similarity contribution of the user pair $(u_i, u_j)$ to the item pair $(v_q, v_t)$ is inversely proportional to the total number of items that users $u_i$ and $u_j$ have co-interacted with.
Through the consideration of all swing structures and the foregoing weighting scheme, the retrieval algorithm based on Swing effectively filters out noisy casual user-item interactions and neutralizes active user bias in traditional itemCF, and thus, is able to identify deeply correlated item pairs as better results~\cite{yang2020large,huang2025comprehensive}.

\update{Despite its superior effectiveness, efficiently processing Swing queries over large graphs is challenging due to the nested formulation.}
As pinpointed in \cite{yang2020large}, given a query item $v_q$, the exact computation of the Swing scores $\textsf{sw}(v_q,v_t)$ across all target items $v_t\in \G$ needs an average time complexity of $O(d(v_q)^2\cdot \overline{d})$, where $d(v_q)$ is the degree of $v_q$ in $\G$ and $\overline{d}$ represents the average user degree.
In practical industrial e-commerce recommender systems, item catalogs often comprise billions of items, many of which can easily accumulate thousands of user interactions, while popular items have even hundreds of thousands of interacting users (i.e., a large $d(v_q)$) due to the power-law distribution of user behavior~\cite{yang2022hicf,kersbergen2021learnings}.
Evaluating a single query item thus requires iterating over millions or billions of user pairs due to the $d(v_q)^2$ factor in its time complexity, which can take up to hours and is prohibitively expensive.
As a workaround, tech giants like Alibaba rely heavily on extensive MapReduce parallelization frameworks and powerful distributed computing centers to process {Swing} computations~\cite{yang2020large}. 
An alternative solution is to artificially cap the maximum number of users considered for a query item (typically at 600)~\cite{aliyunswing}. However, this truncation heuristic might overlook critical user pairs, inevitably degrading the result quality, as evidenced by our empirical studies in \S\ref{sec:performance-eval}.

\stitle{Present Work}
To bridge this technical gap, this paper formulates the Swing computation task as the $(\epsilon,\lambda)$-approximate Swing query and top-$K$ Swing query, and presents two new algorithms \algo{} (\underline{A}daptive \underline{S}wing \underline{C}omputation) and \kalgo{} for answering these queries.
Particularly, the $(\epsilon,\lambda)$-approximate Swing query requires the approximate Swing $\widetilde{\textsf{\textnormal{sw}}}(v_q,v_t)$ of any target item $v_t$ to satisfy: (i) the relative error guarantee of $|\widetilde{\textsf{\textnormal{sw}}}(v_q,v_t)-{\textsf{\textnormal{sw}}}(v_q,v_t)|\le \epsilon\cdot {\textsf{\textnormal{sw}}}(v_q,v_t)$ when ${\textsf{\textnormal{sw}}}(v_q,v_t)\ge \lambda$, and (ii) the additive guarantee of $|\widetilde{\textsf{\textnormal{sw}}}(v_q,v_t)-{\textsf{\textnormal{sw}}}(v_q,v_t)|\le \epsilon\cdot \lambda$ if ${\textsf{\textnormal{sw}}}(v_q,v_t)$ is below the threshold $\lambda$. 
Our preliminary empirical studies of the na{\"{\i}}ve brute-force method and Monte-Carlo approach for Swing computation reveal that both methods suffer from expensive set intersection operations, and the demand to enumerate or sample substantial user pairs, rendering them rather inefficient, particularly for high-degree items.

Inspired by these observations, we first develop a series of non-trivial optimizations ameliorating \texttt{QFilter}~\cite{han2018speeding} for faster computations of the intersection set and its cardinality over user-item graphs.
Based thereon, we propose two randomized algorithms {\em Grouped Na{\"{\i}}ve Sampling} (\texttt{GNS}) and {\em User Subset Sampling} (\texttt{USS}), which judiciously reduce set intersections in na{\"{\i}}ve Monte-Carlo and sidestep costly set intersections of user pairs with large degrees in the brute-force method, while offering rigorous theoretical guarantees in terms of $(\epsilon,\lambda)$-approximation.
On top of that, we integrate \texttt{GNS} and \texttt{USS} into a unified framework \algo{} for answering $(\epsilon,\lambda)$-approximate Swing queries, enabling the adaptive processing of query items of varied neighborhoods with minimal runtime cost.
Furthermore, to overcome the inefficiency of \algo{} in dealing with top-$K$ Swing queries for high-degree query items over large graphs, we design a filter-refinement framework \kalgo{}. Under the hood, it first derives rough Swing estimations of possible target items using a moiety of the sample budget, whilst efficiently maintaining their lower and upper bounds via our carefully-crafted incremental algorithm, which facilitates the subsequent identification of top-$K$ candidates and borderline items.
Afterwards, we adaptively refine the approximate Swing scores of the pinpointed borderline items using the remaining sample budget for an accurate top-$K$ ranking. 
Our extensive experiments over 8 real datasets exhibit that \algo{} and \kalgo{} are orders of magnitude faster than baseline exact and Monte-Carlo approaches when attaining the same $(\epsilon,\lambda)$-approximation guarantees or comparable top-$K$ ranking precisions. In particular, on the largest graph {\em MAG}, \kalgo{} is able to achieve an average precision of more than 99.9\% for top-$100$ queries within 1.5 milliseconds, whereas the exact method requires 8.5 seconds.

To summarize, our contributions are as follows:
\begin{itemize}[nosep,leftmargin=*]
\item \update{First, we introduce the $(\epsilon,\lambda)$-approximate Swing formulation combining probabilistic relative and additive error guarantees, and pinpoint the limitations of baseline methods through in-depth empirical studies and analyses.}
\item \update{Second, building on standard randomized-estimation principles, we develop two complementary estimators, with \texttt{GNS} removing redundant interactions through grouping and \texttt{USS} avoiding materialization through cardinality-only estimation. We integrate them into \algo{} using an analytical cost model for query-adaptive estimator selection with approximation and runtime guarantees.}
\item \update{Third, we further upgrade \algo{} to accelerate top-$K$ Swing queries through \kalgo{}, building on filter-refinement concentrate approximation effort on borderline candidates, with empirical bounds and bitmap pruning for efficient refinement.}
\item Lastly, we conduct extensive experiments to demonstrate the superiority of \algo{} and \kalgo{} over the competing methods in terms of both query effectiveness and efficiency.
\end{itemize}

\section{Preliminaries}

\subsection{Notations and Problem Statement}\label{sec:prob-stat}
Let $\G=(\U\cup \I,\EDG)$ be a user-item interaction graph, where $\U$ (resp. $\I$) is a set of $|\U|$ (resp. $|\I|$) users (resp. items), and $\EDG$ is a set of $|\EDG|$ edges between $\U$ and $\I$. 
For each edge $e_{q,i}\in \EDG$, we say $v_q$ and $u_i$ are neighbors to each other, and we use $\U(v_q)$ (resp. $\I(u_i)$) to denote the set of neighbors of item $v_q$ (resp. user $u_i$), where the degree is $d(v_q)=|\U(v_q)|$ (resp. $d(u_i)=|\I(u_i)|$). 
For each item $v_q$, we denote by $D(v_q)$ the sum of its neighbors' degrees, i.e., $\sum_{u\in \U(v_q)}{d(u)}$.
Table~\ref{tbl:symbol} lists the notations frequently used throughout this paper.
\begin{figure}[!ht]
    \centering
    \includegraphics[width=0.8\columnwidth]{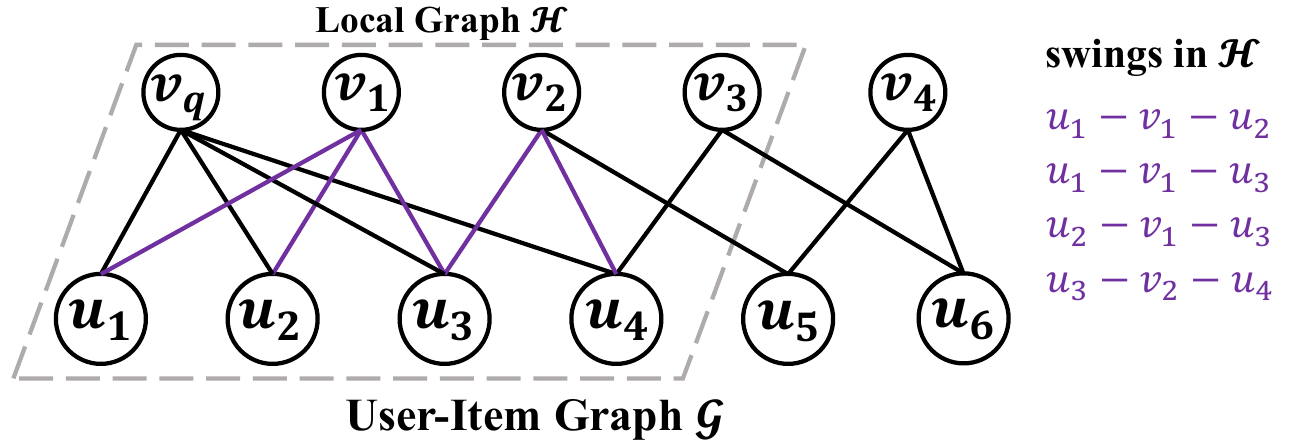}
    \vspace{-2ex}
    \caption{An Example for Swing}
    \label{fig:swing}
    \vspace{-3ex}
\end{figure}

\stitle{Definition of Swing~\cite{yang2020large}} 
As exemplified in Figure~\ref{fig:swing}, we consider the local graph $\mathcal{H}$ containing all users who interacted with query item $v_q$ and all items interacted with by these users, i.e., $u_1$-$u_4$ and $v_1$-$v_3$.
Each ``user-item-user'' structure in $\mathcal{H}$ is referred to as a {\em swing}, i.e., $\langle u_1,v_1,u_2 \rangle, \langle u_1,v_1,u_3 \rangle, \langle u_2,v_1,u_3 \rangle, \langle u_3,v_2,u_4 \rangle$.
Mathematically, considering all swing structures involving query item $v_q$ and target item $v_t$, the standard {\em Swing} score of item pair $(v_q,v_t)$ is defined as
\begin{small}
\begin{equation}\label{eq:swing}
\textsf{sw}(v_q,v_t)=\sum_{u_i, u_j\in \U(v_q)\cap \U(v_t)}{\frac{1}{\alpha+|\I(u_i)\cap \I(u_j)|}},
\end{equation}
\end{small}
where $\alpha$ stands for the smoothing coefficient, typically $\alpha=1$. 
Particularly, a higher Swing $\textsf{sw}(v_q,v_t)$ indicates a stronger correlation between items $v_q$ and $v_t$. 
Lemma~\ref{lem:swing-range}\footnote{All missing proofs appear in 
Appendix~\ref{sec:proof}
.} gives the range of $\textsf{sw}(v_q,v_t)$.
\begin{lemma}\label{lem:swing-range}
$\forall{v_t\in \I\setminus\{v_q\}},\ 0\le \textsf{\textnormal{sw}}(v_q,v_t)\le \frac{d(v_q)^2-d(v_q )}{\alpha+2}$.
\end{lemma}

\begin{table}[!t]
\centering
\renewcommand{\arraystretch}{1.1}
\begin{footnotesize}
\caption{Frequently used notations.}\vspace{-3mm} \label{tbl:symbol}
\resizebox{\columnwidth}{!}{%
\begin{tabular}{|p{0.51in}|p{2.5in}|}
\hline
{\bf Notation}  &  {\bf Description}\\
\hline
$\U, \I, \EDG$   & The user set, item set and edge set.\\ \hline
$\U(v_q), \I(u_i)$   & The sets of neighbors of item $v_q$ and user $u_i$, respectively.\\ \hline
$d(v_q), D(v_q)$   & The degree of $v_q$ and sum of degrees of its neighbors.\\ \hline
$\textsf{sw}(v_q,v_t), \alpha$ & The exact Swing score of $v_t$ w.r.t. $v_q$ and smoothing coefficient. \\ \hline
$\widetilde{\textsf{\textnormal{sw}}}(v_q,v_t)$ & The approximate Swing score of $v_t$ w.r.t. $v_q$. \\ \hline
$\epsilon, \lambda, \delta$ & The relative error threshold, Swing threshold, failure probability.\\ \hline
$n, n_f, n_r$ & The numbers of random samples. \\ \hline
$\rho$ & The empirical runtime ratio used in Eq.~\eqref{eq:select}. \\ \hline
\end{tabular}%
}
\end{footnotesize}
\vspace{-2ex}
\end{table}

\stitle{Problem Statement} 
\update{In this paper, we formalize i2i retrieval as a top-$K$ graph-proximity query, termed the {\em top-$K$ Swing query}.}
To be precise, given a user-item graph $\G$ and a query item $v_q$, the top-$K$ Swing query asks for the $K$ items $\mathcal{Y}$ (typically $|\mathcal{Y}|=K=20$~\cite{librecommed}) with the highest Swing scores w.r.t. $v_q$ from the item set $\I\setminus \{v_q\}$.

\begin{definition}[$(\epsilon,\lambda)$-Approximate Swing]\label{def:approx-swing}
Let $\widetilde{\textsf{\textnormal{sw}}}(v_q,v_t)$ be an approximation of ${\textsf{\textnormal{sw}}}(v_q,v_t)\ \forall{v_t\in \I}$. We say $\{\widetilde{\textsf{\textnormal{sw}}}(v_q,v_t)\}_{v_t\in \I}$ are $(\epsilon,\lambda)$-approximate Swing values if they satisfy the following conditions:
for any item $v_t\in \I$ with ${\textsf{\textnormal{sw}}}(v_q,v_t)\ge \lambda$,
\begin{equation}\label{eq:approx-bound}
|\widetilde{\textsf{\textnormal{sw}}}(v_q,v_t)-{\textsf{\textnormal{sw}}}(v_q,v_t)|\le \epsilon\cdot {\textsf{\textnormal{sw}}}(v_q,v_t)
\end{equation}
holds, and for any item $v_t\in \I$ s.t. ${\textsf{\textnormal{sw}}}(v_q,v_t)<\lambda$, we have
\begin{equation}\label{eq:approx-bound-2}
|\widetilde{\textsf{\textnormal{sw}}}(v_q,v_t)-{\textsf{\textnormal{sw}}}(v_q,v_t)|\le \epsilon\cdot \lambda.
\end{equation}
\end{definition}

A fundamental problem underlying the top-$K$ Swing query is thus the accurate computation or estimation of Swing scores for the given query item $v_q$. Instead of computing the exact Swing values, which could be highly costly, this paper studies the computation of {\em $(\epsilon,\lambda)$-approximate Swing} scores formulated in Definition~\ref{def:approx-swing}. 

The parameters $\epsilon$ and $\lambda$ stand for the relative error threshold and Swing threshold, respectively, whose typical values are $0.05$ and $\frac{d(v_q)(d(v_q)-1)}{100(\alpha+2)}$.
Particularly, this approximate definition offers a stronger relative error
guarantee when ${\textsf{\textnormal{sw}}}(v_q,v_t)\ge \lambda$, and a weaker additive error guarantee otherwise.
The reason is that 
ensuring relative error assurance for items with ${\textsf{\textnormal{sw}}}(v_q,v_t)<\lambda$ is expensive and unnecessary as they are less important.

\subsection{Set Intersection}\label{sec:set-intersection}
As a fundamental operation in Swing computation, $\I(u_i)\cap \I(u_j)$ aims to find the intersection of two sets $\I(u_i)$ and $\I(u_j)$. One simple and classic way for intersecting ordered sets with comparable sizes (i.e., $|\I(u_i)|\approx |\I(u_j)|$) is to employ a linear merge by scanning both sets in parallel, which requires $O(|\I(u_i)| + |\I(u_j)|)$ operations~\cite{ding2011fast} and is wasteful when set sizes differ significantly. For different set sizes ($|\I(u_i)|<|\I(u_j)|$), the complexity can be reduced to $O(|\I(u_i)|+\log_2{\binom{|\I(u_i)|+|\I(u_j)|}{|\I(u_i)|}})$ by leveraging the asymmetry~\cite{hwang1972simple}.
Amid significant subsequent works for set intersections~\cite{baeza2005experimental,sanders2007intersection,ding2011fast,zhang2024bigset,tsirogiannis2009improving,bille2007fast}, recent works~\cite{inoue2014faster,lemire2016simd,han2018speeding} capitalize on the two-pointer strategy and {\em Single Instruction Multiple Data} (SIMD) instructions to perform multiple comparisons in parallel and have achieved remarkable performance improvements.

\begin{figure}[!t]
    \centering
    \includegraphics[width=0.7\columnwidth]{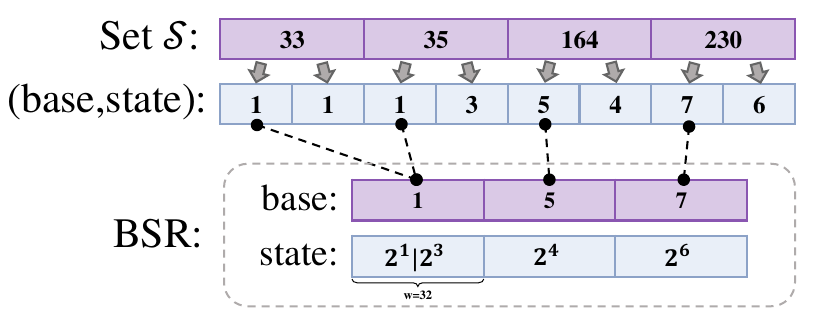}
    \vspace{-3ex}
    \caption{An Example for BSR ($w=32$)}
    \label{fig:BSR}
\end{figure}

In particular, \texttt{QFilter}~\cite{han2018speeding} 
encodes sets using the {\em Base and State Representation} (BSR), a compact layout, and filters out unnecessary comparisons in a single byte-level check with high probability.
As shown in Figure~\ref{fig:BSR}, consider an original set $\mathcal{S}$ consisting of index values. Each index is first decomposed into a base value and a state value, where the base is obtained by integer division by $w$, and the state is the remainder. We let $w=32$, e.g., $33 \rightarrow(\lfloor 33/32 \rfloor, 33 \bmod 32)=(1,1)$, $35\rightarrow(\lfloor 35/32 \rfloor,35\bmod32)=(1,3)$. Entries sharing the same base are then merged by aggregating state bits via bitwise OR, e.g., $(1,1)$ and $(1,3)$ are combined into $(1,2^1|2^3)$, yielding a $w\times$ compression ratio. The base and state arrays are stored separately, enabling faster SIMD loading of consecutive values.

We develop a series of optimizations that tailor \texttt{QFilter} to accelerate neighbor set intersections over user-item graphs, referred to as \texttt{QFilter++} \footnote{The algorithmic details of \texttt{QFilter++} are deferred to 
Appendix~\ref{sec:FSI}
.}.
Notably, using \texttt{QFilter++}, the intersection cardinality computation can be done more efficiently (up to 3$\times$ faster) through a bitwise AND and a hardware-supported popcount without explicitly materializing the intersection set.

\subsection{Baseline Methods for Swing Computation}\label{sec:baseline}

Below, we discuss baseline algorithms for computing Swing scores.

\stitle{Exact Method} In \cite{yang2020large}, the authors describe a brute-force method (Algorithm~\ref{alg:basic}, henceforth referred to as \texttt{Exact}) for computing the exact Swing scores of all items in $\I$ with a query item $v_q$. Specifically, \texttt{Exact} detects swing structures by enumerating all user pairs in $\U(v_q)$. That is, for each user pair $(u_i,u_j)\in \U(v_q)\times \U(v_q)$ with $u_i\neq u_j$, \texttt{Exact} computes the intersection of their neighbor sets as $\I_{i,j}\gets \I(u_i)\cap \I(u_j)$ (Line 3). For each target item $v_t$ in the set intersection $\I_{i,j}$, its Swing is increased by $\frac{1}{\alpha+|\I_{i,j}|}$ at Line 5, i.e., the weight for swing structure $\langle u_i,v_t,u_j\rangle$. 

\begin{lemma}\label{lem:exact-time}
Given query item $v_q$, the worst-case time complexity of Algorithm~\ref{alg:basic} is $O\left(\frac{d(v_q)-1}{2}\cdot D(v_q)\right)$.
\end{lemma}

According to Lemma~\ref{lem:exact-time}, the worst-case time complexity of \texttt{Exact} is bounded by $O\left(\frac{d(v_q)-1}{2}\cdot D(v_q)\right)$, which can be rewritten as $O\left(\frac{d(v_q)(d(v_q)-1)}{2}\cdot \frac{D(v_q)}{d(v_q)}\right)$, where $\frac{d(v_q)(d(v_q)-1)}{2}$ is the number of all distinct user-pairs, and $\frac{D(v_q)}{d(v_q)}$ is the average degree of $v_q$'s neighbors, which can be perceived as the average cost of the set intersection.
Due to the $\frac{d(v_q)(d(v_q)-1)}{2}$ factor, the time complexity of \texttt{Exact} is quadratic to $d(v_q)$, which is rather costly for query items with numerous neighbors. Practical solutions often resort to a {\em truncation} strategy~\cite{aliyunswing} that simply limits the number of $v_q$'s neighbors (e.g., at most $600$) for Swing computation. However, this workaround leads to an inaccurate query result.

\begin{algorithm}[!t]
\small
\caption{Exact Method~\cite{yang2020large}}\label{alg:basic}
\KwIn{$\G=(\U\cup\I,\EDG)$ and query item $v_q$}
\KwOut{$\textsf{sw}(v_q,v_t)\ \forall{v_t\in \I}$}
\For{$u_i\in \U(v_q)$}{
\For{$u_j\in \U(v_q)\setminus \{u_i\}$}{
$\I_{i,j}\gets \I(u_i)\cap \I(u_j)$\;
\For{$v_t\in \I_{i,j}$}{
Increase $\textsf{sw}(v_q,v_t)$ by $\frac{1}{\alpha+|\I_{i,j}|}$\;
}
}
}
\end{algorithm}

\begin{algorithm}[!t]
\small
\caption{\NMC Method}\label{alg:sample}
\KwIn{$\G=(\U\cup\I,\EDG)$, query item $v_q$, and sample size $n$}
\KwOut{$\widetilde{\textsf{sw}}(v_q,v_t)\ \forall{v_t\in \I}$}
\For{$r\gets 1$ to $n$}{
Sample two users $u_i, u_j$ ($u_i\neq u_j$) from $\U(v_q)$\;
$\I_{i,j}\gets \I(u_i) \cap \I(u_j)$\;
\For{$v_t\in \I_{i,j}$}{
Increase $\widetilde{\textsf{sw}}(v_q,v_t)$ by $\frac{d(v_q)\cdot (d(v_q)-1)}{n(\alpha+|\I_{i,j}|)}$\;
}
}
\end{algorithm}

\stitle{Na{\"{\i}}ve Monte-Carlo Method}
Since the Swing score is the summation of weights $\frac{1}{\alpha+|\I_{i,j}|}$ of all possible swing structures $\langle u_i,v_t,u_j\rangle$, a simple and straightforward idea is to utilize the Monte-Carlo method for computing approximate Swing scores.
Instead of enumerating all user pairs so as to detect all valid swing structures, the na{\"{\i}}ve Monte-Carlo approach (referred to as \NMC) is to sample a fixed number $n$ of user pairs from $\U(v_q)\times \U(v_q)$ for Swing estimation. As illustrated at Lines 2-5 in Algorithm~\ref{alg:sample}, for $r$-th trial, \NMC picks a user $u_i$ from $\U(v_q)$ uniformly at random, followed by another user $u_j$ from $\U(v_q)\setminus \{u_i\}$. Afterwards, akin to \texttt{Exact}, \NMC obtains the set intersection $\I_{i,j}$ and updates the approximate Swing score $\widetilde{\textsf{sw}}(v_q,v_t)$ for each target item in $\I_{i,j}$ by adding an increment $\omega_r=\frac{d(v_q)\cdot (d(v_q)-1)}{\alpha+|\I_{i,j}|}$ divided by $n$.

The probability of drawing a specific user pair is $\frac{1}{d(v_q)\cdot (d(v_q)-1)}$. 
Let $Z_r$ be the random variable representing the observed score for the $r$-th trial. By the definition, consider any target item $v_t$,
\begin{small}
\begin{align*}
\mathbb{E}[Z_r] &= \sum_{u_i,u_j\in \U(v_q), u_i\neq u_j}{\frac{1}{d(v_q)\cdot (d(v_q)-1)}\cdot \omega_r\cdot \mathbb{1}[v_t\in \I_{i,j}]}\\
& = \sum_{u_i,u_j\in \U(v_q), u_i\neq u_j}{\frac{1}{\alpha+|\I_{i,j}|}\cdot \mathbb{1}[v_t\in \I_{i,j}]} = \textsf{sw}(v_q,v_t),
\end{align*}
\end{small}
which indicates that $\frac{1}{n}\cdot\sum_{r=1}^{n}{Z_r}$ is an unbiased estimator of $\textsf{sw}(v_q,v_t)$.
Since the set intersection at Line 3 takes $O(\min(d(u_i), d(u_j)))$ time, the worst time complexity of \NMC is bounded by $O(n\cdot \max_{u\in \U(v_q)}{d(u)})$. On average, its time complexity is $O\left(n\cdot \frac{D(v_q)}{d(v_q)}\right)$.
Moreover, by setting the number $n$ of samples properly, the estimated Swing scores returned by Algorithm~\ref{alg:sample} can be guaranteed to be $(\epsilon,\lambda)$-approximate Swing with a high probability. We defer the detailed analysis to \S\ref{sec:sampling}.

\subsection{Empirical Study and Analysis}\label{sec:empirical-study}
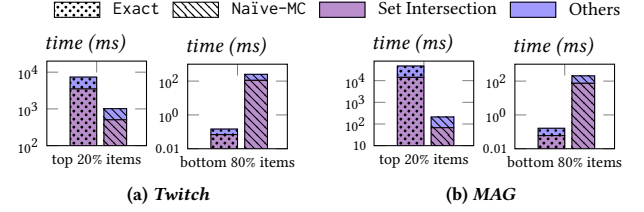
\begin{figure}[!t]
\centering
\begin{small}
\begin{tikzpicture}
   \hspace{3mm}\begin{customlegend}[
        legend entries={\texttt{Exact},\NMC, Set Intersection, Others},
        legend columns=4,
        area legend,
        legend style={at={(0.45,1.25)},anchor=north,draw=none,font=\footnotesize,column sep=0.1cm}]
        \addlegendimage{pattern color=black, pattern={crosshatch dots}} 
        \addlegendimage{pattern color=black, pattern={north west lines}}
        \addlegendimage{preaction={fill, my_blue!60}} 
        \addlegendimage{preaction={fill, mycyan2}}
    \end{customlegend}
\end{tikzpicture}
\\[-8pt]
\vspace{-1ex}
\subfloat[{\em Twitch}]{
\begin{tikzpicture}[scale=1]
\begin{axis}[
    height=2.7cm,
    width=2.9cm,
    xtick=\empty,
    ybar stacked,
    bar shift=7pt,
    bar width=0.3cm,
    enlarge x limits=0.4,
    ylabel={\em time (ms)},
    ymin=100,
    ymax=20000,
    ytick={100,1000,10000},
    yticklabels={$10^{2}$,$10^{3}$,$10^{4}$},
    ymode=log,
    log origin y=infty,
    log basis y={10},
    xticklabel style = {font=\scriptsize},
    yticklabel style = {font=\tiny},
    symbolic x coords={top 20\% items}, 
    xtick = data,
    tick align=inside,
    every axis y label/.style={at={(current axis.north west)},right=5mm,above=0mm},
    legend style={draw=none, at={(1.02,1.02)},anchor=north west,cells={anchor=west},font=\scriptsize},
    legend image code/.code={ \draw [#1] (0cm,-0.1cm) rectangle (0.3cm,0.15cm);},
    ]

\addplot [pattern color=black, preaction={fill, my_blue!60}, pattern={north west lines}] coordinates
    {(top 20\% items,506.800)};

\addplot [pattern color=black, preaction={fill, mycyan2}, pattern={north west lines}] coordinates
    {(top 20\% items,507.501)};

\end{axis}
\begin{axis}[height=2.7cm,
    width=2.9cm,
    xtick=\empty,
    ybar stacked,
    bar shift=-5pt,
    ymin=100,
    ymax=20000,
    ymode=log,
    log origin y=infty,
    log basis y={10},
    symbolic x coords={top 20\% items},
    axis line style={draw=none},
    tick style={draw=none},
    yticklabel=\empty,
]
\addplot [pattern color=black, preaction={fill, my_blue!60}, pattern={crosshatch dots}] coordinates
    {(top 20\% items,3513.010) };

\addplot [pattern color=black, preaction={fill, mycyan2}, pattern={crosshatch dots}] coordinates
    {(top 20\% items,3836.848) };
\end{axis}
\end{tikzpicture}%
\begin{tikzpicture}[scale=1]
\begin{axis}[
    height=2.7cm,
    width=2.9cm,
    xtick=\empty,
    ybar stacked,
    bar shift=7pt,
    bar width=0.3cm,
    enlarge x limits=0.4,
    ylabel={\em time (ms)},
    ymin=0.01,
    ymax=1000,
    ytick={0.01,1,100},
    yticklabels={$0.01$,$10^{0}$,$10^{2}$},
    ymode=log,
    log origin y=infty,
    log basis y={10},
    xticklabel style = {font=\scriptsize},
    yticklabel style = {font=\tiny},
    symbolic x coords={bottom 80\% items},
    xtick = data,
    tick align=inside,
    every axis y label/.style={at={(current axis.north west)},right=5mm,above=0mm},
    legend style={draw=none, at={(1.02,1.02)},anchor=north west,cells={anchor=west},font=\scriptsize},
    legend image code/.code={ \draw [#1] (0cm,-0.1cm) rectangle (0.3cm,0.15cm);},
    ]

\addplot [pattern color=black, preaction={fill, my_blue!60}, pattern={north west lines}] coordinates
    {(bottom 80\% items,111.053)};

\addplot [pattern color=black, preaction={fill, mycyan2}, pattern={north west lines}] coordinates
    {(bottom 80\% items,140.63225)};
\end{axis}
\begin{axis}[height=2.7cm,
    width=2.9cm,
    xtick=\empty,
    ybar stacked,
    bar shift=-5pt,
    ymin=0.01,
    ymax=1000,
    ymode=log,
    log origin y=infty,
    log basis y={10},
    symbolic x coords={bottom 80\% items},
    axis line style={draw=none},
    tick style={draw=none},
    yticklabel=\empty,
]
\addplot [pattern color=black, preaction={fill, my_blue!60}, pattern={crosshatch dots}] coordinates
    {(bottom 80\% items,0.0685) };

\addplot [pattern color=black, preaction={fill, mycyan2}, pattern={crosshatch dots}] coordinates
    {(bottom 80\% items,0.07725) };
\end{axis}
\end{tikzpicture}\hspace{4mm}%
}
\subfloat[{\em MAG}]{
\begin{tikzpicture}[scale=1]
\begin{axis}[
    height=2.7cm,
    width=2.9cm,
    xtick=\empty,
    ybar stacked,
    bar shift=7pt,
    bar width=0.3cm,
    enlarge x limits=0.4,
    ylabel={\em time (ms)},
    ymin=10,
    ymax=80000,
    ytick={10,100,1000,10000},
    yticklabels={$10$,$10^{2}$,$10^{3}$,$10^{4}$,$10^{5}$},
    ymode=log,
    log origin y=infty,
    log basis y={10},
    xticklabel style = {font=\scriptsize},
    yticklabel style = {font=\tiny},
    symbolic x coords={top 20\% items},
    xtick = data,
    tick align=inside,
    every axis y label/.style={at={(current axis.north west)},right=5mm,above=0mm},
    legend style={draw=none, at={(1.02,1.02)},anchor=north west,cells={anchor=west},font=\scriptsize},
    legend image code/.code={ \draw [#1] (0cm,-0.1cm) rectangle (0.3cm,0.15cm);},
    ]

\addplot [pattern color=black, preaction={fill, my_blue!60}, pattern={north west lines}] coordinates
    {(top 20\% items,67.492)};

\addplot [pattern color=black, preaction={fill, mycyan2}, pattern={north west lines}] coordinates
    {(top 20\% items,143.330)};
\end{axis}
\begin{axis}[height=2.7cm,
    width=2.9cm,
    xtick=\empty,
    ybar stacked,
    bar shift=-5pt,
    ymin=10,
    ymax=80000,
    ymode=log,
    log origin y=infty,
    log basis y={10},
    symbolic x coords={top 20\% items},
    axis line style={draw=none},
    tick style={draw=none},
    yticklabel=\empty,
]
\addplot [pattern color=black, preaction={fill, my_blue!60}, pattern={crosshatch dots}] coordinates
    {(top 20\% items,14229.659) };

\addplot [pattern color=black, preaction={fill, mycyan2}, pattern={crosshatch dots}] coordinates
    {(top 20\% items,33244.560) };
\end{axis}
\end{tikzpicture}%
\begin{tikzpicture}[scale=1]
\begin{axis}[
    height=2.7cm,
    width=2.9cm,
    xtick=\empty,
    ybar stacked,
    bar shift=7pt,
    bar width=0.3cm,
    enlarge x limits=0.4,
    ylabel={\em time (ms)},
    ymin=0.01,
    ymax=1000,
    ytick={0.01,1,100},
    yticklabels={$0.01$,$10^{0}$,$10^{2}$},
    ymode=log,
    log origin y=infty,
    log basis y={10},
    xticklabel style = {font=\scriptsize},
    yticklabel style = {font=\tiny},
    symbolic x coords={bottom 80\% items},
    xtick = data,
    tick align=inside,
    every axis y label/.style={at={(current axis.north west)},right=5mm,above=0mm},
    legend style={draw=none, at={(1.02,1.02)},anchor=north west,cells={anchor=west},font=\scriptsize},
    legend image code/.code={ \draw [#1] (0cm,-0.1cm) rectangle (0.3cm,0.15cm);},
    ]

\addplot [pattern color=black, preaction={fill, my_blue!60}, pattern={north west lines}] coordinates
    {(bottom 80\% items,74.1)};

\addplot [pattern color=black, preaction={fill, mycyan2}, pattern={north west lines}] coordinates
    {(bottom 80\% items,130.21)};
\end{axis}
\begin{axis}[height=2.7cm,
    width=2.9cm,
    xtick=\empty,
    ybar stacked,
    bar shift=-5pt,
    ymin=0.01,
    ymax=1000,
    ymode=log,
    log origin y=infty,
    log basis y={10},
    symbolic x coords={bottom 80\% items},
    axis line style={draw=none},
    tick style={draw=none},
    yticklabel=\empty,
]
\addplot [pattern color=black, preaction={fill, my_blue!60}, pattern={crosshatch dots}] coordinates
    {(bottom 80\% items,0.0595) };

\addplot [pattern color=black, preaction={fill, mycyan2}, pattern={crosshatch dots}] coordinates
    {(bottom 80\% items,0.103) };
\end{axis}
\end{tikzpicture}\hspace{0mm}%
}

\vspace{-1em}
\end{small}
\caption{Empirical efficiency of \texttt{Exact} and \NMC for different query items on {\em Twitch} and {\em MAG}.} \label{fig:Motive}
\end{figure}

To reveal the strengths and weakness of the deterministic method \texttt{Exact} and randomized algorithm \NMC in Swing computation, we conduct empirical studies to evaluate the efficiency of these two baselines on two large real datasets, {\em Twitch} and {\em MAG}.
Considering that \texttt{Exact} method outputs exact Swing values whilst \NMC is an approximate algorithm, for a fair comparison, we tune the number $n$ of random samples used in \NMC so that it achieves at least 99.9\% precision in terms of top-$K$ ($K=20$) ranking accuracy.
\update{Note that this empirical study uses the original set intersection implementation rather than \texttt{QFilter++}.
}
Figure~\ref{fig:Motive} reports the average query time (in milliseconds), along with its breakdowns into set intersections and other operations, for \texttt{Exact} and \NMC over two distinct query sets, i.e., the top 20\% high-degree items and the remaining 80\% of items.

From these results, we can make the following observations. 
First, for popular items (top 20\% high-degree items), both \texttt{Exact} and \NMC incur high query costs, due to the need for enumerating and sampling substantial user-pairs. For instance, on {\em Twitch} dataset, \texttt{Exact} and \NMC take more than 3.8 and 0.5 seconds, respectively, on average to answer a single query, which are far from satisfactory for practical applications. 
Second, for both methods, the primary computational bottleneck is set intersections, which account for over 80\% of the total query time on {\em Twitch} and around 50\% on {\em MAG}.
Last but not least, \texttt{Exact} is more efficient for low-degree items but incurs prohibitive cost on popular items, whereas \NMC exhibits more stable query performance across all query items. Although it is remarkably faster than \texttt{Exact} on popular ones, its query cost still remains large. In addition, for the remaining $80\%$ of items, the actual cost of \texttt{Exact} can be very cheap, i.e., requires less than $1$ ms on average, since these items admit a small or highly constrained search space.
In contrast, \NMC incurs non-negligible overhead and large sample size to stabilize high variances in estimations and achieve acceptable accuracy~\cite{motwani1996randomized}, whose value can easily outstrip the actual number of distinct user-pairs, i.e., $n> \frac{d(v_q)(d(v_q)-1)}{2}$, especially for low-degree items.

In summary, our above observations suggest that both \texttt{Exact} and \NMC suffer from severe deficiencies in Swing computation/approximation. Therefore, it is necessary to optimize both approaches and devise adaptive frameworks combining their merits  for efficient $(\epsilon,\lambda)$-approximate Swing computation as well as top-$K$ Swing processing.

\section{The \algo{} Algorithm}\label{sec:algo-ASC}
This section presents \algo{}, our solution for $(\epsilon,\lambda)$-approximate Swing computation. \algo{} is motivated by the aforementioned empirical observations of baseline algorithms, including expensive set intersections and their limitations in processing high- and low-degree query items.
To curtail the significant number of set intersections required in \NMC, \S\ref{sec:sampling} introduces an optimized version of \NMC, i.e., \texttt{GNS}, for computing $(\epsilon,\lambda)$-approximate Swing.
Subsequently, we develop \texttt{USS} in \S\ref{sec:subset_sample} to circumvent costly intersections of large neighbor sets in \texttt{Exact}, while ensuring approximation guarantees.
Lastly, we propose to integrate \texttt{GNS} and \texttt{USS} within a hybrid framework \algo{} that dynamically selects the algorithm for fast Swing approximation along with rigorous theoretical analyses in \S\ref{sec:combine}.
Figure~\ref{fig:ASC} provides a figurative illustration of \algo{}.

\begin{algorithm}[!t]
\small
\caption{Grouped Na{\"{\i}}ve Sampling (\texttt{GNS})}\label{alg:group_sample}
\KwIn{$\G=(\U\cup\I,\EDG)$, query item $v_q$, error threshold $\epsilon$.}
\KwOut{$\widetilde{\textsf{sw}}(v_q,v_t)\ \forall{v_t\in \I}$}
Calculate $n$ according to Eq.~\eqref{eq:comp-n}\;
Initialize an empty map $S$\;
\For{$r\gets 1$ to $n$}{
Sample $u_i, u_j\in \U(v_q)$ s.t. $u_i\neq u_j$ and $d(u_i),d(u_j)\ge 2$\;
Increase $S(u_i,u_j)$ by 1\;
}
\For{$(u_i,u_j)\in S$}{
$\I_{i,j}\gets \I(u_i) \cap \I(u_j)$\;
\For{$v_t\in \I_{i,j}$}{
Increase $\widetilde{\textsf{sw}}(v_q,v_t)$ by $\frac{S(u_i,u_j)\cdot d(v_q)\cdot (d(v_q)-1)}{n(\alpha+|\I_{i,j}|)}$\;
}
}
\end{algorithm}

\subsection{Grouped Na{\"{\i}}ve Sampling}\label{sec:sampling}
Since \NMC requires drawing numerous user-pair samples from $\U(v_q)$ randomly, many of them are invalid (e.g., $\I_{i,j}=\emptyset$) or duplicated, leading to significant redundant or repeated computations for set intersection.
\update{Following the sampling principle of \NMC, our {\em grouped na{\"{\i}}ve sampling} (\texttt{GNS}) method tackles this issue by grouping repeated user-pair samples so that set intersections are computed only once for each distinct pair, while preserving unbiased estimation.}
\update{Intuitively, \texttt{GNS} replaces repeated set intersections with a cheap counter update whenever a sampled user pair reappears.}

Algorithm~\ref{alg:group_sample} displays the pseudo-code of \texttt{GNS}. 
At Line 1, given relative error threshold $\epsilon$, Swing threshold $\lambda$, and failure probability $\delta$ (typically $\delta=10^{-4}$), the sample size $n$ is calculated by
\begin{small}
\begin{equation}\label{eq:comp-n}
n= \left\lceil\frac{2(d(v_q)^2-d(v_q))\cdot(\epsilon/3+1)}{(\alpha+2)\cdot\lambda\epsilon^2}\cdot\log\left(\frac{1}{\delta}\right)\right\rceil.
\end{equation}
\end{small}
Instead of computing the set intersection for each sampled user pair separately, \texttt{GNS} maintains an additional map $S$ (Line 2) to record the occurrences for distinct user pairs. More specifically, at Lines 3-5, Algorithm~\ref{alg:group_sample} first picks $n$ user-pairs from $\U(v_q)$ such that the users in each pair are different and have at least $2$ neighbors (i.e., interacting items).
This requirement is to avoid sampling invalid user pairs since each swing structure derived from a user pair $(u_i,u_j)$ should have at least two common neighboring items with one of them being $v_q$ according to Eq.~\eqref{eq:swing}.
Accordingly, the degrees of sampled users must be greater than or equal to $2$.
Notice that we can ensure the sampling of such users via a filtering of users in $\U(v_q)$ in the preprocessing.
Then, for each sampled pair $(u_i,u_j)$, we increase $S(u_i,u_j)$ by $1$.
Afterwards, for each distinct user pair $(u_i,u_j)$ in the map $S$, \texttt{GNS} finds the intersection $\I_{i,j}$ of $\I(u_i)$ and $\I(u_j)$, efficiently through our \texttt{QFilter++} algorithm (Line 7).
Next, for each common item $v_t\neq v_q$ in $\I_{i,j}$, we increase the approximate Swing $\widetilde{\textsf{sw}}(v_q,v_t)$ by $\frac{S(u_i,u_j)\cdot d(v_q)\cdot (d(v_q)-1)}{n(\alpha+|\I_{i,j}|)}$ as at Lines 8-9.

\begin{figure}[!t]
    \centering
    \includegraphics[width=\columnwidth]{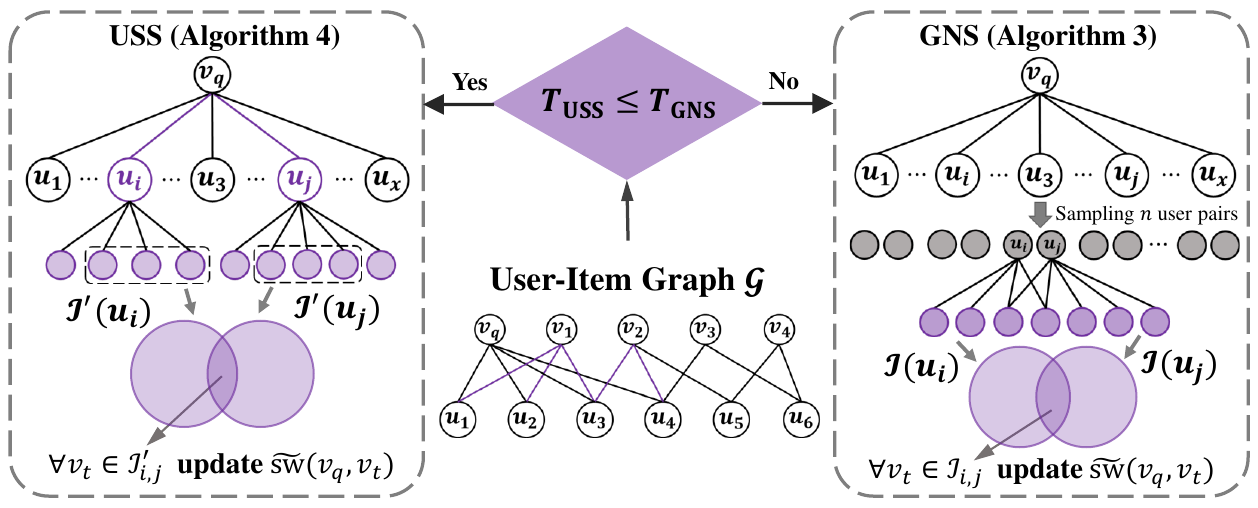}
    \vspace{-6ex}
    \caption{\update{Illustration of \algo{}}}\label{fig:ASC}
\end{figure}

\stitle{Complexity Analysis}
By counting repeated user-pairs with the map $S$, \texttt{GNS} ensures that the set intersections are only computed for distinct user pairs, whose count is at most $d(v_q)\cdot (d(v_q)-1)$ as in \texttt{Exact}. Accordingly, by Lemma~\ref{lem:exact-time}, the computational cost entailed by set intersections in \texttt{GNS} will be bounded by $O\left(\frac{d(v_q)-1}{2}\cdot D(v_q)\right)$.
Since the number of distinct user pairs in $S$ is at most $n$, the set intersection cost can also be estimated by $O\left(n\cdot \frac{D(v_q)}{d(v_q)}\right)$.
Plus the $O(n)$ cost for updating the map $S$ at Lines 2-4, the overall time complexity of \texttt{GNS} is $O\left(n+\min\left\{n,\frac{d(v_q)(d(v_q)-1)}{2}\right\}\cdot \frac{D(v_q)}{d(v_q)}\right)$.

\stitle{Correctness Analysis}
Since each user-pair is still sampled independently at Lines 3-5 in Algorithm~\ref{alg:group_sample}, the resulting $\widetilde{\textsf{\textnormal{sw}}}(v_q, v_t)$ is also an unbiased estimator of ${\textsf{\textnormal{sw}}}(v_q, v_t)$ as analyzed in \NMC.
By leveraging the Chernoff bound~\cite{chung2006concentration}, we can establish the following accuracy guarantee in Theorem~\ref{lem:naive-mc} for both \NMC and \texttt{GNS} when their sample sizes $n$ are greater than the value in Eq.~\eqref{eq:comp-n}.
\begin{theorem}\label{lem:naive-mc}
Let $\widetilde{\textsf{\textnormal{sw}}}(v_q,v_t)\ \forall{v_t\in \I}$ be the approximate Swing scores output by Algorithms~\ref{alg:sample} or~\ref{alg:group_sample}. For any item $v_t\in \I$, $\widetilde{\textsf{\textnormal{sw}}}(v_q,v_t)$ is the $(\epsilon,\lambda)$-approximate Swing with a probability of at least $1-\delta$.
\end{theorem}

\subsection{User Subset Sampling}\label{sec:subset_sample}
In what follows, we shift focus to expensive intersections for user pairs with high degrees in \texttt{Exact}, as it requires retrieving all interacted items and materializing the full intersection. 
However, computing only the intersection cardinality is conspicuously faster than set intersection (3$\times$ speedup when using our \texttt{QFilter++}) as pinpointed in \S\ref{sec:set-intersection}, especially when the overlap is large. 
\update{Motivated by this observation, our {\em user subset sampling} (\texttt{USS}) builds on randomized subset sampling to estimate Swing from small sampled subsets and intersection cardinality, thereby avoiding explicit intersection materialization.}
\update{Intuitively, \texttt{USS} targets a different bottleneck from \texttt{GNS} by reducing the processing cost of each pair, rather than the number of pairs processed, when pair enumeration is affordable.}

Algorithm~\ref{alg:subset_sample} presents the pseudo-code of \texttt{USS}. At Line 1, the sample ratio $\gamma$ is computed based on relative error threshold $\epsilon$, Swing threshold $\lambda$, and failure probability $\delta$ as follows
\begin{small}
\begin{equation}\label{eq:subset_ratio}
\gamma\le \min\left\{{\frac{2(\epsilon/3+1)}{(\alpha+2)\cdot\lambda\epsilon^2 }\cdot\log\left(\frac{1}{\delta}\right)}, 1\right\}.
\end{equation}
\end{small}
When enumerating all valid user-pairs where each has at least two interacting items, instead of computing the exact intersection of sets $\I(u_i)$ and $\I(u_j)$, \texttt{USS} constructs a sampled subset $\I'(u_i) $ from $\I(u_i)$ to enable efficient small-scale set intersection. 
More specifically, at Line 3, Algorithm~\ref{alg:subset_sample} first uniformly samples $\lceil \gamma\cdot d(u_i) \rceil$ items from $\I(u_i)$ where $d(u_i)\le d(u_j)$, ensuring that the sampling ratio is no less than $\gamma$ and preserving the correctness of the error bound. It then computes the intersection set, denoted as $\I'_{i,j}$, between $\I'(u_i)$ and $\I(u_j)$. Notice that $|\I'(u_i)|$ is typically small, we only need to check items in $\I'(u_i)$ against a preprocessed hash map of $\I(u_j)$, resulting in constant-time lookup per item and avoiding expensive set matching. 
\update{Afterwards, at Line 5, we filter out user pairs with no overlap and compute the intersection cardinality using our size-only variant of \texttt{QFilter++} algorithm in Appendix~\ref{sec:FSI}
to yield a mean speedup of $1.5\times$ and up to $3.6\times$ speedup compared to the set materialization version.} Next, we increase the approximate Swing $\widetilde{\textsf{sw}}(v_q,v_t)$ for each common item $v_t\neq v_q$ in $\I'_{i,j}$ by $\frac{1}{\frac{|\I^\prime(u_i)|}{d(u_i)}\cdot (\alpha+|\I_{i,j}|)}$ as at Lines 6-7.

\begin{algorithm}[!t]
\small
\caption{User Subset Sampling (\texttt{USS})}\label{alg:subset_sample}
\KwIn{$\G=(\U\cup\I,\EDG)$, query item $v_q$, error threshold $\epsilon$.}
\KwOut{$\widetilde{\textsf{sw}}(v_q,v_t)\ \forall{v_t\in \I}$}
Calculate $\gamma$ according to Eq.~\eqref{eq:subset_ratio}\;
\For{$u_i\neq u_j\in \U(v_q)$ such that $d(u_i),d(u_j)\ge 2$}{
$\I^\prime(u_i)\gets$ Sample $\lceil \gamma\cdot d(u_i)\rceil$ items from $\I(u_i)$\;
$\I^\prime_{i,j}\gets \I^\prime(u_i) \cap \I(u_j)$\;
\lIf{$|\I^\prime_{i,j}|>0$}{Compute $|\I_{i,j}|$}
\For{$v_t\in \I^\prime_{i,j}$}{
Increase $\widetilde{\textsf{sw}}(v_q,v_t)$ by $\frac{1}{\frac{|\I^\prime(u_i)|}{d(u_i)}\cdot (\alpha+|\I_{i,j}|)}$\;
}
}
\end{algorithm}

\stitle{Complexity Analysis}
Since Lines 3-4 process only a $\gamma$ fraction of items, their worst-case complexity is 
$O\left(\frac{(d(v_q)-1) D(v_q)}{2}\cdot \gamma \right)$ as per Lemma~\ref{lem:exact-time}. 
For Lines 5-7, the number of user pairs is the same as \texttt{Exact}. 
Let $\eta$ be the cost of computing an intersection cardinality. This stage takes 
$O\left(\frac{(d(v_q)-1) D(v_q)}{2}\cdot \eta \right)$ time. 
Therefore, the complexity of \texttt{USS} is 
$O\left(\frac{(d(v_q)-1) D(v_q)}{2}\cdot (\eta+\gamma) \right)$, 
which is asymptotically the same as \texttt{Exact} but with a smaller constant factor.

\stitle{Correctness Analysis}
\update{Consider the $r$-th update associated with a user pair $(u_i,u_j)$.
Since $\I^\prime(u_i)$ is drawn uniformly from $\I(u_i)$ at Line~4,
each target item $v_t\in\I_{i,j}$ is included in $\I^\prime_{i,j}$ with probability
$\frac{|\I^\prime(u_i)|}{d(u_i)}$.
Line~7 compensates for this sampling probability by scaling its contribution
by the reciprocal factor, such that the expected update remains
$\frac{1}{\alpha+|\I_{i,j}|}$.
By the linearity of expectation over all pairs sharing $v_t$, we have
$\mathbb{E}[\widetilde{\textsf{\textnormal{sw}}}(v_q,v_t)]
={\textsf{\textnormal{sw}}}(v_q,v_t)$, indicating that
$\widetilde{\textsf{\textnormal{sw}}}(v_q,v_t)$ is strictly unbiased.
The detailed derivation is deferred to Appendix~\ref{sec:proof}.
}
Next, we can harness the Chernoff bound~\cite{chung2006concentration} to obtain the following approximation guarantee for \texttt{USS}: 
\begin{theorem}\label{lem:USS}
$\widetilde{\textsf{\textnormal{sw}}}(v_q,v_t)\ \forall{v_t\in \I}$ returned by Algorithm~\ref{alg:subset_sample} are $(\epsilon,\lambda)$-approximate Swing scores with a probability of at least $1-\delta$.
\end{theorem}

\subsection{Complete Algorithm and Analysis}\label{sec:combine}
As revealed in \S\ref{sec:empirical-study}, \texttt{Exact} and \NMC excel at processing different query items. 
With the foregoing efficiency optimizations, \texttt{GNS} is able to obtain conspicuous speed-ups over \NMC without compromising result quality on high-degree items, but is still less than efficient on low-degree items, as validated in \S\ref{sec:approx-query}.
On the other hand, although \texttt{USS} improves upon \texttt{Exact} by averting the time-consuming intersections of large neighbor sets of users through sampling, it remains computationally expensive for high-degree items.
\update{This complementarity motivates \algo{} to perform query-adaptive estimator selection through an analytical cost model that selects the lower-cost estimator between \texttt{GNS} and \texttt{USS} for each query item.}
Algorithm~\ref{alg:ASC} illustrates the pseudo-code of \algo{}. Precisely, \algo{} invokes \texttt{USS} (Algorithm~\ref{alg:subset_sample}) when the following condition holds (Lines 2-3):
\begin{small}
\begin{equation}\label{eq:select}
\frac{d(v_q)\cdot (d(v_q)-1)}{2}\cdot \rho \le n\left(1+\frac{d(v_q)}{D(v_q)}\right),
\end{equation}
\end{small}
and runs \texttt{GNS} (Algorithm~\ref{alg:group_sample}) otherwise (Lines 4-5). The parameter $\rho$ stands for the empirical runtime ratio of primitive operations in \texttt{USS} and \texttt{GNS}.
As per our empirical analyses on real datasets, $\rho$ falls into a small interval of $[0.22,0.67]$ and setting $\rho=0.5$ always leads to an effective selection of \texttt{GNS} and \texttt{USS} in \algo{}.

\begin{algorithm}[!t]
\small
\caption{\texttt{ASC}}\label{alg:ASC}
\KwIn{$\G=(\U\cup\I,\EDG)$, query item $v_q$, error threshold $\epsilon$.}
\KwOut{$\widetilde{\textsf{sw}}(v_q,v_t)\ \forall{v_t\in \I}$}
Calculate $n$ according to Eq.~\eqref{eq:comp-n}\;
\If{Eq.~\eqref{eq:select} holds}{
Invoke Algorithm~\ref{alg:subset_sample} with $v_q$ and $\epsilon$\;
}\Else{
Invoke Algorithm~\ref{alg:group_sample} with $v_q$ and $\epsilon$\;
}
\end{algorithm}

To enable effective selection in \algo{}, the key is to accurately estimate the empirical runtime of both methods prior to execution, given that the asymptotic complexity is often inaccurate and abstracts away constant factors, memory access patterns, and implementation details that dominate actual performance in practice.
Specifically, given $v_q$, we quantify the runtime for \texttt{USS} by 
\begin{small}
\begin{equation}
T_{\texttt{USS}}=\rho_{\texttt{USS}} \cdot \left(\frac{d(v_q)(d(v_q)-1)}{2}\cdot\frac{D(v_q)}{d(v_q)}\right),
\end{equation}
\end{small}
where $\frac{(d(v_q)-1)D(v_q)}{2}$ is the count of primitive operations (i.e., Lines 3-7 in Algorithm~\ref{alg:subset_sample}) involved in \texttt{USS} and $\rho_{\texttt{USS}}$ denotes the constant factor calibrated such that $T_{\texttt{USS}}$ matches the observed execution time, i.e., averaged empirical time per primitive operation.
Similarly, we can formulate the empirical runtime for \texttt{GNS} as
\begin{small}
\begin{equation}
T_{\texttt{GNS}}=\rho_{\texttt{GNS}} \cdot \left(n+\min\left\{n,\frac{d(v_q)(d(v_q)-1)}{2}\right\}\cdot \frac{D(v_q)}{d(v_q)}\right),
\end{equation}
\end{small}
where $\rho_{\texttt{GNS}}$ stands for the averaged empirical time for drawing each user-pair and a single primitive operation as in \texttt{USS} (Lines 3-7 in Algorithm~\ref{alg:subset_sample} or Lines 7-9 in Algorithm~\ref{alg:group_sample}). Intuitively, the empirical runtime ratio $\rho = \frac{\rho_{\texttt{USS}}}{\rho_{\texttt{GNS}}}<1$.

\stitle{Complexity and Correctness Analysis}
Our rigorous analysis in Theorem~\ref{lem:cost-ASC} indicates that with the selection condition in Eq.~\eqref{eq:select}, \algo{} can achieve the minimal computational cost of \texttt{GNS} and \texttt{USS}.
\begin{theorem}\label{lem:cost-ASC}
Algorithm~\ref{alg:ASC} runs in $\min\{T_{\texttt{USS}}, T_{\texttt{GNS}}\}$ time.
\end{theorem}
Since \algo{} is a combination of \texttt{GNS} and \texttt{USS}, it offers the same accuracy assurances. 
According to Theorems~\ref{lem:naive-mc} and~\ref{lem:USS}, $\textsf{sw}(v_q,v_t)$ $\forall{v_t\in \I}$ returned by Algorithm~\ref{alg:ASC} are $(\epsilon,\lambda)$-approximate Swing values with a high probability.

\section{The \kalgo{} Algorithm}\label{sec:algo-KASC}
Although the efficiency techniques developed in \algo{} considerably enhance its practical efficiency for accurate Swing estimation, and thus, enable faster top-$K$ Swing query processing,
it is still rather costly for handling popular items over large-scale graphs as shown in \S\ref{sec:performance-eval}, due to the sheer amount of random samples needed in \texttt{GNS}.
\update{To overcome this limitation, this section presents \kalgo{} adapting the filter-refinement paradigm to top-$K$ Swing queries, concentrating the sampling budget on borderline candidates rather than uniformly refining all items.}
We first provide a synoptic overview of \kalgo{} in \S\ref{sec:sol-overview}, followed by elucidating the secondary algorithms for {\em greedy candidate generation} and {\em adaptive pairwise refinement} in \S\ref{sec:cand-gen} and \S\ref{sec:PR}, respectively. Subsequently, \S\ref{sec:ASC-theory} theoretically analyzes the correctness and computational complexities of the proposed algorithms.
\update{
Lastly, we further extend \kalgo{} to provide a theoretical guarantee on the correctness of top-$K$ results, albeit at a higher computational cost. The algorithmic details are deferred to Appendix~\ref{sec:KASC-cor} in the interest of space.
}

\begin{figure}[!t]
\centering
\begin{small}
\begin{tikzpicture}
    \begin{customlegend}
    [legend columns=2,
        legend entries={Recall@$K$, Recall@$(K+5)$},
        legend style={at={(0.45,1.35)},anchor=north,draw=none,font=\tiny,column sep=0.2cm}]
    \addlegendimage{only marks,line width=0.25mm,mark=o,color=my_blue}
    \addlegendimage{only marks,line width=0.25mm,mark=diamond,color=NSCcol1}
    \end{customlegend}
\end{tikzpicture}
\\[-\lineskip]
\vspace{-4mm}
\subfloat[{\em SteamGame}]{
\begin{tikzpicture}[scale=1,every mark/.append style={mark size=2pt}]
    \begin{axis}[
        height=\columnwidth/3.1,
        width=\columnwidth/2.85,
        ylabel={$K=20$},
        enlarge x limits=0.1,
        xmin=2, xmax=5,
        ymin=0.94,
        xtick={2.0,3.0,4.0,5.0},
        ytick={0.94,0.96,0.98,1},
        xticklabel style = {font=\tiny},
        yticklabel style = {font=\tiny},
        xticklabels={0.1,0.05,0.02,0.01},
        yticklabels={94\%,96\%,98\%,100\%},
        every axis y label/.style={font=\tiny,at={(current axis.north west)},right=6mm,above=0mm},
        legend style={fill=none,font=\small,at={(0.02,0.99)},anchor=north west,draw=none},
    ]
    \addplot[line width=0.25mm,mark=o,color=my_blue]  
        plot coordinates {
            (2,0.965)
            (3,0.9805)
            (4,0.99025)
            (5,0.99525)
        };
    \addplot[line width=0.25mm,mark=diamond,color=NSCcol1]  
        plot coordinates {
            (2,0.9985)
            (3,1)
            (4,1)
            (5,1)
        };

    \end{axis}
\end{tikzpicture}
\begin{tikzpicture}[scale=1,every mark/.append style={mark size=2pt}]
    \begin{axis}[
        height=\columnwidth/3.1,
        width=\columnwidth/2.85,
        ylabel={$K=100$},
        enlarge x limits=0.1,
        xmin=2, xmax=5,
        ymin=0.94,
        xtick={2.0,3.0,4.0,5.0},
        ytick={0.94,0.96,0.98,1},
        xticklabel style = {font=\tiny},
        yticklabel style = {font=\tiny},
        xticklabels={0.1,0.05,0.02,0.01},
        yticklabels={94\%,96\%,98\%,100\%},
        every axis y label/.style={font=\tiny,at={(current axis.north west)},right=6mm,above=0mm},
        legend style={fill=none,font=\small,at={(0.02,0.99)},anchor=north west,draw=none},
    ]
    \addplot[line width=0.25mm,mark=o,color=my_blue]  
        plot coordinates {
            (2,0.9549)
            (3,0.9777)
            (4,0.98965)
            (5,0.9948)
        };
    \addplot[line width=0.25mm,mark=diamond,color=NSCcol1]  
        plot coordinates {
            (2,0.97445)
            (3,0.9931)
            (4,0.9994)
            (5,1.0000)
        };

    \end{axis}
\end{tikzpicture}\hspace{0mm}
}
\subfloat[{\em Twitch}]{
\begin{tikzpicture}[scale=1,every mark/.append style={mark size=2pt}]
    \begin{axis}[
        height=\columnwidth/3.1,
        width=\columnwidth/2.85,
        ylabel={$K=20$},
        xmin=2, xmax=5,
        enlarge x limits=0.1,
        ymin=0.97,
        xtick={2.0,3.0,4.0,5.0},
        ytick={0.97,0.98,0.99,1},
        xticklabel style = {font=\tiny},
        yticklabel style = {font=\tiny},
        xticklabels={0.1,0.05,0.02,0.01},
        yticklabels={97\%,98\%,99\%,100\%},
        every axis y label/.style={font=\tiny,at={(current axis.north west)},right=6mm,above=0mm},
        legend style={fill=none,font=\small,at={(0.02,0.99)},anchor=north west,draw=none},
    ]
    \addplot[line width=0.25mm,mark=o,color=my_blue]  
        plot coordinates {
            (2,0.98525)
            (3,0.99175)
            (4,0.9955)
            (5,0.99825)
        };
    \addplot[line width=0.25mm,mark=diamond,color=NSCcol1]  
        plot coordinates {
            (2,1)
            (3,1)
            (4,1)
            (5,1)
        };
    
    \end{axis}
\end{tikzpicture}
\begin{tikzpicture}[scale=1,every mark/.append style={mark size=2pt}]
    \begin{axis}[
        height=\columnwidth/3.1,
        width=\columnwidth/2.85,
        ylabel={$K=100$},
        xmin=2, xmax=5,
        enlarge x limits=0.1,
        ymin=0.97,
        xtick={2.0,3.0,4.0,5.0},
        ytick={0.97,0.98,0.99,1},
        xticklabel style = {font=\tiny},
        yticklabel style = {font=\tiny},
        xticklabels={0.1,0.05,0.02,0.01},
        yticklabels={97\%,98\%,99\%,100\%},
        every axis y label/.style={font=\tiny,at={(current axis.north west)},right=6mm,above=0mm},
        legend style={fill=none,font=\small,at={(0.02,0.99)},anchor=north west,draw=none},
    ]
    \addplot[line width=0.25mm,mark=o,color=my_blue]  
        plot coordinates {
            (2,0.97675)
            (3,0.9888)
            (4,0.995)
            (5,0.9971)
        };
    \addplot[line width=0.25mm,mark=diamond,color=NSCcol1]  
        plot coordinates {
            (2,0.9915)
            (3,0.99905)
            (4,1.00)
            (5,1.00)
        };
    
    \end{axis}
\end{tikzpicture}\hspace{0mm}
}
\vspace{-3ex}
\end{small}
\caption{Top-$K$ recall by \texttt{GNS} when varying $\epsilon_r$} \label{fig:Fiter-Refine}
\vspace{-1ex}
\end{figure}
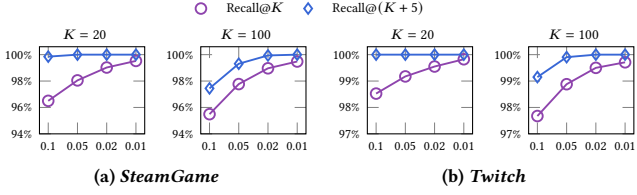

\vspace{-2ex}
\begin{figure}[!h]
    \centering
    \includegraphics[width=0.95\columnwidth]{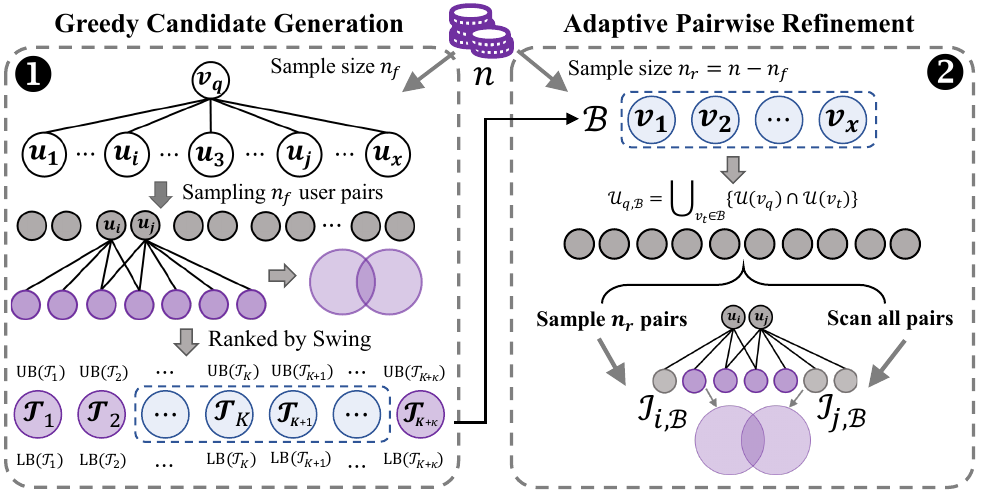}
    \vspace{-3ex}
    \caption{Illustration of \kalgo{}.}\label{fig:K-ASC}
    \vspace{-3ex}
\end{figure}

\subsection{Solution Overview}\label{sec:sol-overview}
Akin to \algo{}, \kalgo{} also adaptively selects \texttt{USS} or \texttt{GNS} for query processing. 
The distinction is that, instead of employing \texttt{GNS} for accurate estimation of Swing scores for all target items, \kalgo{} focuses on identifying a subset of items that are most likely to be the top-$K$ items directly.

The high-level idea is inspired by our empirical findings in Figure~\ref{fig:Fiter-Refine}, from which we can observe that, on both {\em SteamGame} and {\em Twitch} datasets, running \texttt{GNS} with a larger relative error threshold $\epsilon$ (and thus fewer random samples) yields a top-$(K+5)$ list with recall close to $100\%$. In contrast, achieving the same recall for the top-$K$ items returned directly by \texttt{GNS} requires a substantially smaller $\epsilon$, and therefore many more samples.
For instance, when $K=20$, attaining $99.8\%$ recall requires $\epsilon=0.1$ for the top-$25$ results, but $\epsilon=0.01$ for the top-$20$ results, indicating a $100$-fold increase in the number of samples in \texttt{GNS} as the sample size $n$ scales as $\frac{1}{\epsilon^2}$ (Theorem~\ref{lem:naive-mc}).
This suggests that most items in the exact top-$K$ set can be identified accurately at a low cost, whereas only a few ``borderline items'' dominate the sampling overhead.
Based on this observation, we propose to identify these borderline items and allocate additional samples specifically to them for Swing estimation, rather than expending unnecessary samples on easy items that have already been estimated with sufficient confidence.

As summarized in Algorithm~\ref{alg:KASC}, we invoke \texttt{USS} for Swing computation if Eq.~\eqref{eq:select} is satisfied as in \algo{} (Lines 1-3), and directly utilize all items with non-zero approximate Swing values as top-$K$ candidates $\mathcal{T}$ (Line 4).
However, when \texttt{GNS} is selected, \kalgo{} first greedily generates a set of candidates $\mathcal{T}$ ($|\mathcal{T}|=K+\kappa$ and $\kappa\ll K$) with a subset $\mathcal{B}\subset \mathcal{T}$ containing borderline candidates for being the top-$K$ using Algorithm~\ref{alg:candidate} by drawing $n_f=\lceil \beta\cdot n \rceil$ (typically $\beta=0.1$) random samples (Line 7). Subsequently, the remaining sampling budget $n_r=n-n_f$ is allocated to adaptively refine the Swing scores of borderline items in $\mathcal{B}$ by Algorithm~\ref{alg:pair} (Line 8).
Particularly, by enforcing $n_f+n_r=n$, the total number of sampling rounds in \kalgo{} remains bounded by $n$, ensuring its sampling cost under control.
Finally, the candidate items in $\mathcal{T}$ are sorted based on the latest Swing estimations for obtaining the final ranking and top-$K$ result $\mathcal{Y}$ (Line 9).
The basic idea is illustrated in Figure~\ref{fig:K-ASC}.

\SetKwIF{SameIf}{SameElseIf}{SameElse}
  {\update{\textnormal{Lines 1--3 are the same as Algorithm~\ref{alg:ASC}}}}
  {}{else if}{else}{}

\begin{algorithm}[!t]
\small
\caption{\kalgo{}}\label{alg:KASC}
\KwIn{$\G=(\U\cup\I,\EDG)$, query item $v_q$, error threshold $\epsilon$, and integers $K, \kappa$}
\KwOut{Top-$K$ items $\mathcal{Y}$}
{\nonl
\SameIf{}{%
    \setcounter{AlgoLine}{3}
    \update{$\mathcal{T}\gets
    \{v_t\in \I\setminus\{v_q\}\mid
    \widetilde{\textsf{sw}}(v_q,v_t)>0\}$\;}
}
\SameElse{%
    $n_f\gets \lceil \beta\cdot n \rceil;\ 
      n_r\gets n-n_f$\;

    $\mathcal{T},\mathcal{B}\gets$
    Algorithm~\ref{alg:candidate} with $v_q,n_f$\;

    Invoke Algorithm~\ref{alg:pair} with $\mathcal{B},n_r$\;
}}

$\mathcal{Y} \gets {\operatorname*{argtop}K}_{v_t\in \mathcal{T}}{\ \widetilde{\textsf{sw}}(v_q,v_t)}$\;
\end{algorithm}

\subsection{Greedy Candidate Generation}\label{sec:cand-gen}
The first task is to find the $K+\kappa$ candidates $\mathcal{T}$ for the final top-$K$ items and {\em borderline} items $\mathcal{B}$ therein, where $\kappa\ll K$ specifies the small number of the extra items that are potentially in the exact top-$K$. 
As aforementioned, we can easily generate $\mathcal{T}$ with roughly approximate Swing scores from a small number of samples.
To identify $\mathcal{B}$ from $\mathcal{T}$, our idea is to leverage the {\em empirical Bernstein inequality}~\cite{audibert2007tuning} in Theorem~\ref{lem:bernstein} to 
\update{construct tight item-specific confidence bounds for the estimated Swing scores using the empirical variances.}
With these bounds at hand, we are able to determine the borderline items that might fall into or outside the top-$K$ list.

\begin{algorithm}[!t]
\small
\caption{Greedy Candidate Generation}\label{alg:candidate}
\KwIn{$\G=(\U\cup\I,\EDG)$, query item $v_q$, and sample size $n_f$}
\KwOut{$\mathcal{T}$, $\mathcal{B}$}
{\nonl Lines 1-4 are the same as Lines 2-5 in Algorithm~\ref{alg:group_sample}\;}
\setcounter{AlgoLine}{4}
Initialize a counter $r(v_t)\ \forall{v_t\in \I}$\;
\For{$(u_i,u_j)\in S$}{
$\I_{i,j}\gets \I(u_i) \cap \I(u_j)$\;
\For{$v_t\in \I_{i,j}$}{
$r(v_t) \gets r(v_t) + S(u_i,u_j)$\;
Update $\widetilde{\textsf{sw}}(v_q,v_t)$ according to Eq.~\eqref{eq:update-mu}\;
Update $\sigma(v_t)$ according to Eq.~\eqref{eq:update-sigma}\;
}}
\For{$v_t \in \I$ s.t. $r(v_t)\neq 0$}{
$\widetilde{\textsf{sw}}(v_q,v_t) \gets \frac{r(v_t)}{n_f}\cdot \widetilde{\textsf{sw}}(v_q,v_t)$\;
$\sigma(v_t)\gets \sigma(v_t)+\frac{(n_f-r(v_t))(\frac{r(v_t)}{n_f}\cdot\widetilde{\textsf{sw}}(v_q,v_t)^2-\sigma(v_t))}{n_f-1}$\;
}
$\mathcal{T} \gets {\operatorname*{argtop-}(K+\kappa)}_{v_t\in \I\setminus\{v_q\}}{\ \widetilde{\textsf{sw}}(v_q,v_t)}$\;
Compute $\textsf{UB}(v)$ and $\textsf{LB}(v)\ \forall{v\in \mathcal{T}}$\;
\For{$i\gets K$ \KwTo $1$}{
\lIf{$\textnormal{\textsf{UB}}(\mathcal{T}_{K+1})\ge \textnormal{\textsf{LB}}(\mathcal{T}_{i})$}{
Add $\mathcal{T}_{i}$ to $\mathcal{B}$
}
}
\For{$i\gets 1$ \KwTo $\kappa$}{
\lIf{$\textnormal{\textsf{UB}}(\mathcal{T}_{K+i})\ge \textnormal{\textsf{LB}}(\mathcal{T}_{K})$}{
Add $\mathcal{T}_{K+i}$ to $\mathcal{B}$
}
}
\end{algorithm}

\begin{theorem}[\cite{audibert2007tuning}]\label{lem:bernstein}
Let $Z_1,\dotsc,Z_{n}$ be real-valued i.i.d. random variables, such that $0\le Z_i \le \zeta$. We denote by $Z=\frac{1}{n}\sum_{i=1}^{n}{Z_i}$ the empirical mean of these variables and $\overline{\sigma}=\frac{1}{n}\sum_{i=1}^{n}{(Z_i-Z)^2}$ their empirical variance. Then, we have
\begin{small}
\begin{equation}\label{eq:fx}
\begin{gathered}
\mathbb{P}\left[|Z-\mathbb{E}[Z]|\ge \Phi(n,\overline{\sigma},\zeta,\delta)\right]\le \delta,\ \textnormal{where}\\
\Phi(n,\overline{\sigma},\zeta,\delta)=\sqrt{\frac{2{\overline{\sigma}}\log{(3/\delta)}}{n}}+\frac{3\zeta\log{(3/\delta)}}{n}.
\end{gathered}
\end{equation}
\end{small}
\end{theorem}

\stitle{Algorithm} Algorithm~\ref{alg:candidate} displays the pseudo-code for finding the top-$K$ candidates $\mathcal{T}$ based upon \texttt{GNS}. Algorithm~\ref{alg:candidate} begins by sampling $n_f$ user-pairs and forming the map $S$ as in \texttt{GNS} (Lines 1-4). For each distinct user-pair $(u_i,u_j)$ in $S$, distinct from \texttt{GNS}, we iteratively and dynamically update the empirical mean $\widetilde{\textsf{sw}}(v_q,v_t)$ and variance $\sigma(v_t)$ for each target item $v_t$ in the way analogous to Welford's algorithm~\cite{welford1962note} that facilitates the efficient update (Lines 6-11). To be specific, Algorithm~\ref{alg:candidate} records and incrementally maintains the number of occurrences $r(v_t)$ of target item $v_t$ observed from user-pair samples at Line 9. Based thereon, we update $\widetilde{\textsf{sw}}(v_q,v_t)$ as follows:
\begin{small}
\begin{equation}\label{eq:update-mu}
\begin{gathered}
\widetilde{\textsf{sw}}(v_q,v_t) \gets \widetilde{\textsf{sw}}(v_q,v_t) + \frac{S(u_i,u_j)\cdot \Delta}{r(v_t)},\\
\Delta \gets \frac{ d(v_q)\cdot (d(v_q)-1)}{\alpha+|\I_{i,j}|} - \widetilde{\textsf{sw}}(v_q,v_t).
\end{gathered}
\end{equation}
\end{small}
Accordingly, 
the empirical variance $\sigma(v_t)$ is updated to
\begin{small}
\begin{equation}\label{eq:update-sigma}
\sigma(v_t) \gets \sigma(v_t) + \frac{S(u_i,u_j)\cdot \left(\frac{r(v_t)-S(u_i,u_j)}{r(v_t)}\Delta^2-\sigma(v_t)\right)}{r(v_t)-1}.
\end{equation}
\end{small}
After iterating over $S$, Algorithm~\ref{alg:candidate} applies a final correction to the empirical mean $\widetilde{\textsf{sw}}(v_q,v_t)$ and variance $\sigma(v_t)$ for each sampled target item $v_t$, i.e., item $v_t\in \I$ such that $r(v_t)\neq 0$ at Lines 12-14. 

Next, Algorithm~\ref{alg:candidate} proceeds to identify the confident and borderline candidates based on the above estimations and bounds. Firstly, we find $K+\kappa$ items $\mathcal{T}$ with the highest approximate Swing scores at Line 15, followed by calculating the upper and lower bounds of each item $v\in \mathcal{T}$ in Line 16 as follows: 
\begin{equation}\label{eq:up-low-bound}
\begin{gathered}
\textsf{UB}(v) \gets \widetilde{\textsf{sw}}(v_q,v_t) + \Phi(n_f,\sigma(v_t),\zeta,\delta),\\
\textsf{LB}(v) \gets \widetilde{\textsf{sw}}(v_q,v_t) - \Phi(n_f,\sigma(v_t),\zeta,\delta).
\end{gathered} 
\end{equation}
Then, we examine the sorted items in $\mathcal{T}$ using these bounds. More concretely, 
among the current top-$K$ items in $\mathcal{T}$, we mark an item $\mathcal{T}_i$ as borderline if $\textnormal{\textsf{LB}}(\mathcal{T}_{i}) \le \textnormal{\textsf{UB}}(\mathcal{T}_{K+1})$ (Lines 17-18).
This condition implies that these items might fall outside the exact top-$K$ list, and hence, should be added to $\mathcal{B}$ as borderline candidates.
On the other hand, at Lines 19-20, Algorithm~\ref{alg:candidate} further examines the remaining $\kappa$ items in $\mathcal{T}$, among which those with upper bound satisfying $\textnormal{\textsf{UB}}(\mathcal{T}_{K+i})\ge \textnormal{\textsf{LB}}(\mathcal{T}_{K})$ are also added to $\mathcal{B}$ as they are likely to be in the exact top-$K$ items.

\stitle{Complexity Analysis} 
Lines 1-11 extend \texttt{GNS} by additionally maintaining the occurrence count $r(v_t)$ and estimator variance $\sigma(v_t)$ for each target item $v_t$. 
Since each update in Lines 9--11 takes $O(1)$ time, this stage has complexity 
$O\left(n_f+\min\left\{n_f,\frac{d(v_q)(d(v_q)-1)}{2}\right\}\cdot \frac{D(v_q)}{d(v_q)}\right)$, 
following \S~\ref{sec:sampling} with $n$ replaced by $n_f$. 
Because the set of valid target items in Line 12 is collected from Line 8 and is bounded by the same term, both the correction step (Lines 12-14) and the top-$(K+\kappa)$ selection procedure (Line 15), which runs in linear time in the number of valid target items, have the same complexity. Since the candidate set $\mathcal{T}$ is restricted to the top-$(K+\kappa)$ items, the borderline candidate detection step (Lines 16-20) requires $O(K)$ time. 
Therefore, the overall complexity is 
$O\left(K+n_f+\min\left\{n_f,\frac{d(v_q)(d(v_q)-1)}{2}\right\}\cdot \frac{D(v_q)}{d(v_q)}\right)$.

\begin{figure}[!t]
    \centering
    \includegraphics[width=0.95\columnwidth]{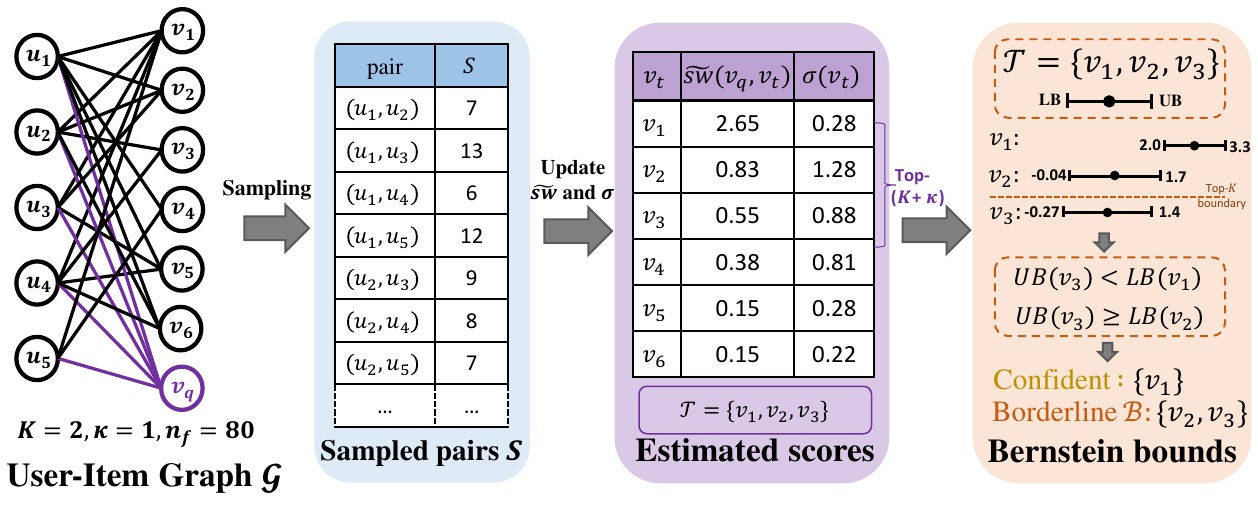}
    \vspace{-3ex}
    \caption{\update{A running example for Algorithm~\ref{alg:candidate}.}}\label{fig:candidate-example}
\end{figure}

\begin{example}
\update{
Figure~\ref{fig:candidate-example} presents a running example of
Greedy Candidate Generation with $K=2$, $\kappa=1$, and $n_f=80$.
Suppose that sampling the user pairs yields the frequency map $S$ shown in the figure.
By updating the approximate Swing scores and empirical variances
according to Eqs.~\eqref{eq:update-mu} and~\eqref{eq:update-sigma},
we obtain estimated Swing scores
$[2.65,0.83,0.55,0.38,0.15,0.15]$ for
$v_1$--$v_6$, respectively, leading to the top-$(K+\kappa)$ candidate
set $\mathcal{T}=\{v_1,v_2,v_3\}$.
Based on their empirical variances, the resulting Bernstein bounds of
$v_1$, $v_2$, and $v_3$ are
$[2.0,3.3]$, $[-0.04,1.7]$, and $[-0.27,1.4]$, respectively.
Since $\textsf{UB}(v_3)<\textsf{LB}(v_1)$ but
$\textsf{UB}(v_3)\ge\textsf{LB}(v_2)$,
$v_1$ is confidently retained, while $v_2$ and $v_3$ remain
uncertain around the top-$K$ boundary.
Thus, the borderline candidates
$\mathcal{B}=\{v_2,v_3\}$ require further refinement.
}
\end{example}

\subsection{Adaptive Pairwise Refinement}\label{sec:PR}

\begin{algorithm}[!t]
\small
\caption{Adaptive Pairwise Refinement}\label{alg:pair}
\KwIn{$\G=(\U\cup\I,\EDG)$, borderline items $\mathcal{B}$, sample size $n_r$}
\KwOut{$\widetilde{\textsf{sw}}(v_q,v_t)\ \forall{v_t\in \B}$.}
\If{$n_r \ge |\U(v_q)|/2$}{
    $\I_{i,\mathcal{B}} \gets \I(u_i) \cap \mathcal{B}\ \forall{u_i \in \U(v_q)}$\;
    $\U_{q,\mathcal{B}} \gets \{u_i\in \U(v_q) \mid |\I_{i,\mathcal{B}}|>0\}$\;
\If{$\frac{|\U_{q,\mathcal{B}}|\cdot(|\U_{q,\mathcal{B}}|-1)}{2} \le n_r$}{
    \For{$u_i\neq u_j\in \U_{q,\mathcal{B}}$}{
        \For{$v_t \in \I_{i,\mathcal{B}} \cap \I_{j,\mathcal{B}}$}{
            Increase $\widetilde{\textsf{sw}}(v_q,v_t)$ by $\frac{1}{\alpha+|\I_{i,j}|}$\;
        }
    }
}
\Else{
    \For{$r \gets 1$ to $n_r$}{
        Sample $u_i\neq u_j\in \U_{q,\mathcal{B}}$\;
        \For{$v_t \in \I_{i,\mathcal{B}} \cap \I_{j,\mathcal{B}}$}{
            Increase $\widetilde{\textsf{sw}}(v_q,v_t)$ by $\frac{|\U_{q,\mathcal{B}}|\cdot(|\U_{q,\mathcal{B}}|-1)}{n_r(\alpha+|\I_{i,j}|)}$\;
        }
    }
}
}
\Else{
    \For{$r \gets 1$ to $n_r$}{
        Sample $u_i\neq u_j\in \U(v_q)$ s.t. $d(u_i),d(u_j)\ge 2$\;
        \For{$v_t\in \mathcal{B}$ and $v_t\in  \I(u_i)$ and $v_t \in \I(u_j)$}{
                Increase $\widetilde{\textsf{sw}}(v_q,v_t)$ by $\frac{|\U(v_q)|\cdot(|\U(v_q)|-1)}{n_r(\alpha+|\I_{i,j}|)}$\;
        }
    }
}
\end{algorithm}

After obtaining the borderline items $\mathcal{B}$ via Algorithm~\ref{alg:candidate}, our job is to refine the approximate Swing of each item pair $(v_q,v_t)\ \forall{v_t\in \mathcal{B}}$. 
\update{We refine only these candidates by restricting computation to their relevant user pairs, substantially reducing the refinement search space, and
Algorithm~\ref{alg:pair} describes the algorithmic steps.}
By taking as input $\mathcal{B}$ and the remaining sample budget $n_r$, we adaptively select the Swing estimation approach by first checking the condition of $n_r \ge |\U(v_q)|/2$ at Line 1.
When it fails, i.e., $d(v_q)$ exceeds the current sample budget $n_r$, extracting the valid user set from $\U(v_q)$ for Swing computation is costly. Instead, we resort to sampling user pairs like \texttt{GNS}, followed by efficient membership checks within $\mathcal{B}$, $\I(u_i)$, and $\I(u_j)$ to avert expensive set intersections (Lines 14-17).

Conversely, Algorithm~\ref{alg:pair} will execute Lines 2-12. 
Different from Algorithms \ref{alg:basic}-\ref{alg:ASC} where the target items are unknown and probing $O(|\U(v_q)|\cdot(|\U(v_q)|-1))$ user pairs is required, 
computing the Swing for $v_q$ and items in $\mathcal{B}$ merely involves user pairs in $\U_{q,\mathcal{B}}=\bigcup_{v_t\in\mathcal{B}}\{\U(v_q)\cap \U(v_t)\}$, 
which can be constructed by Lines 2-3. The search space is significantly reduced to $O({|\U_{q,\mathcal{B}}|\cdot (|\U_{q,\mathcal{B}}|-1)})$ as $|\U_{q,\mathcal{B}}|$ is often dwarfed by $|\U(v_q)|$.
Notice that $\mathcal{B}$ is typically small (e.g., fewer than 32), the neighbor list can be compressed into a bitmap represented by a 32-bit integer, expediting intersection computation. 
After obtaining $\U_{q,\mathcal{B}}$, Algorithm~\ref{alg:pair} either computes the Swing exactly using the brute-force method (Lines 5-7) or estimates the score in a randomized fashion (Lines 9-12).

More specifically, if the estimated runtime of the brute-force method is less than that of the randomized approach (Line 4), Algorithm~\ref{alg:pair} enumerates all distinct user pairs in the common neighbor set $\U_{q,\mathcal{B}}$ of $v_q$ and items in $\mathcal{B}$, and increases the Swing $\widetilde{\textsf{sw}}(v_q,v_t)$ by $\frac{1}{\alpha+|\I_{i,j}|}$. Note that the set intersection cardinality $|\I_{i,j}|$ can be directly and efficiently computed via \texttt{QFilter++}.
Otherwise, the randomized method is a better choice, Algorithm~\ref{alg:pair} estimates $\textsf{sw}(v_q,v_t)$ by randomly picking $n_r$ user pairs from $\U_{q,\mathcal{B}}$, each of which is sampled uniformly at random with a probability of $\frac{1}{|\U_{q,\mathcal{B}}|\cdot (|\U_{q,\mathcal{B}}|-1)}$. 
For each sampled user pair $(u_i,u_j)$, 
we increase $\widetilde{\textsf{sw}}(v_q,v_t)$ by $ \frac{|\U_{q,t}|\cdot (|\U_{q,t}|-1)}{n_r(\alpha+|\I_{i,j}|)}$.

\stitle{Complexity Analysis} 
When $n_r \ge |\mathcal{U}(v_q)|/2$, the algorithm runs in 
$O(d(v_q)\cdot|\mathcal{B}|)$ time for filtering (Lines 2-3), and the adaptive part (Lines 4-12) takes $O\!\left(\min\!\left\{n_r,\frac{d(v_q)(d(v_q)-1)}{2}\right\}\cdot\frac{D(v_q)}{d(v_q)}\right)$ time. 
When $n_r < |\mathcal{U}(v_q)|/2$, each iteration takes 
$O\!\left(|\mathcal{B}| + \frac{D(v_q)}{d(v_q)}\right)$ time, leading to a total of 
$O\!\left(n_r|\mathcal{B}| + n_r\frac{D(v_q)}{d(v_q)}\right)$ for Lines 14-17. 
By combining both cases and using  the fact $|\mathcal{B}| \le K$, the overall complexity is
$O\!\left(K\cdot \min\{n_r,d(v_q)\} + \min\!\left\{n_r,\frac{d(v_q)(d(v_q)-1)}{2}\right\}\cdot\frac{D(v_q)}{d(v_q)}\right)$.

\begin{figure}[!t]
    \centering
    \includegraphics[width=0.95\columnwidth]{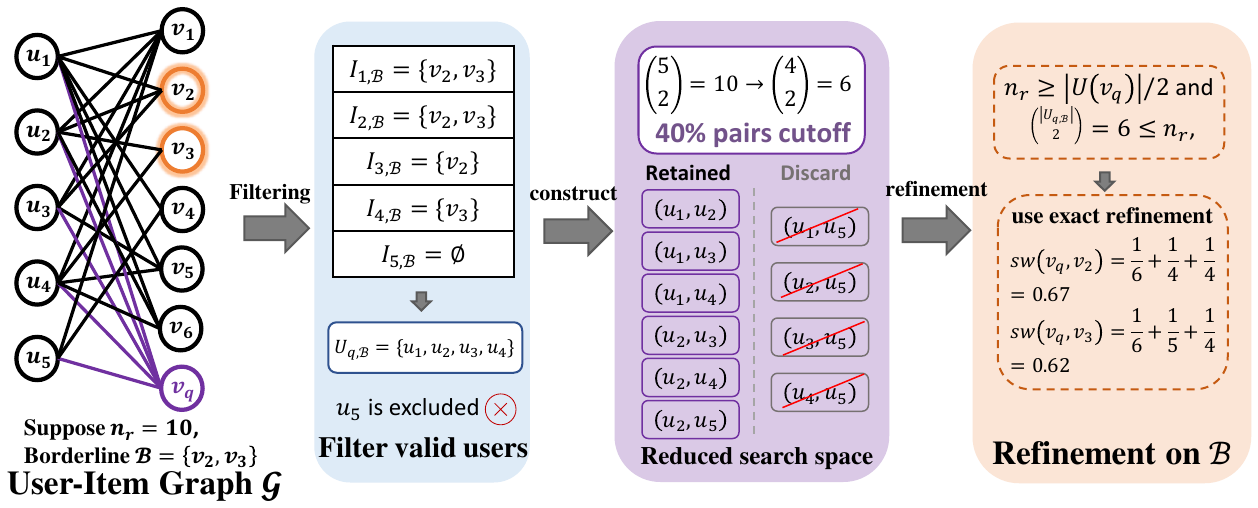}
    \vspace{-3ex}
    \caption{\update{A running example for Algorithm~\ref{alg:pair}.}}\label{fig:refinement-example}
    \vspace{-1ex}
\end{figure}

\begin{example}
\update{
Figure~\ref{fig:refinement-example} presents a running example of
Adaptive Pairwise Refinement with $n_r=10$, $\alpha=1$, and
$\mathcal{B}=\{v_2,v_3\}$.
Since $|\U(v_q)|=5$ and
$n_r\ge |\U(v_q)|/2$, Algorithm~\ref{alg:pair} first filters the
users according to their connections to the borderline items.
This yields
$\I_{1,\mathcal{B}}=\I_{2,\mathcal{B}}=\{v_2,v_3\}$,
$\I_{3,\mathcal{B}}=\{v_2\}$,
$\I_{4,\mathcal{B}}=\{v_3\}$, and
$\I_{5,\mathcal{B}}=\emptyset$, and hence
$\U_{q,\mathcal{B}}=\{u_1,u_2,u_3,u_4\}$.
As a result, the user-pair search space shrinks from
$\binom{5}{2}=10$ to $\binom{4}{2}=6$, i.e., by $40\%$.
Since $\binom{|\U_{q,\mathcal{B}}|}{2}=6\le n_r$,
Algorithm~\ref{alg:pair} (Lines 4--7) performs exact refinement over the reduced
search space.
For $v_2$, the contributing user pairs are
$(u_1,u_2)$, $(u_1,u_3)$, and $(u_2,u_3)$, yielding
$\textsf{sw}(v_q,v_2)=1/6+1/4+1/4=0.67$.
Similarly,
$\textsf{sw}(v_q,v_3)=1/6+1/5+1/4=0.62$.
Thus, all borderline-item Swing scores are exactly refined, excluding irrelevant pairs involving $u_5$.
}
\end{example}

\subsection{Theoretical Analyses}\label{sec:ASC-theory}

\stitle{Correctness Analysis}
First, if Inequality~\eqref{eq:select} holds ($d(v_q)$ is relatively low), Algorithm~\ref{alg:KASC} then takes the branch in Lines 1-4 and Algorithm~\ref{alg:subset_sample} is invoked with $\epsilon$, the resulting Swing estimates satisfy the $(\epsilon,\lambda)$-approximation guarantee according to Theorem~\ref{lem:USS}.
Otherwise, \kalgo{} handles high-degree query items using Algorithms~\ref{alg:candidate} and~\ref{alg:pair}.
Lemma~\ref{lem:mean-var} ensures the correctness of sample mean and variance maintained in Algorithm~\ref{alg:candidate}, which in turn justifies the lower and upper bounds in Eq.~\eqref{eq:up-low-bound} by harnessing Theorem~\ref{lem:bernstein}.

\begin{lemma}\label{lem:mean-var}
$\forall{v_t}\in \I$, 
$\widetilde{\textsf{sw}}(v_q,v_t)$ and $\sigma(v_t)$ in Algorithm~\ref{alg:candidate} are the sample mean and variance over the $n_f$ sampled observations.
\end{lemma}

Since Algorithm~\ref{alg:candidate} runs \texttt{GNS} with a sample size of $n_f=\beta\cdot n$, by Theorem~\ref{lem:naive-mc}, $\widetilde{\textsf{\textnormal{sw}}}(v_q,v_t)\ \forall{v_t\in \I}$ computed by Algorithm~\ref{alg:candidate} are $(\epsilon,\frac{\lambda}{\beta})$-approximate Swing w.p. at least $1-\delta$.
Similarly, Algorithm~\ref{alg:pair} outputs $(\epsilon,\frac{\lambda}{1-\beta})$-approximate Swing $\widetilde{\textsf{\textnormal{sw}}}(v_q,v_t)$ $\forall{v_t\in \mathcal{B}}$ w.p. of at least $1-\delta$, using $n_r=n-n_f$ samples. 
The theoretical approximation assurance of \kalgo{} is established in Theorem~\ref{lem:kASC-approx}.
\begin{theorem}\label{lem:kASC-approx}
Let $\widetilde{\textsf{\textnormal{sw}}}(v_q,v_t)\ \forall{v_t\in \I}$ be the Swing estimates by \kalgo{}. For any item $v_t\in \I \setminus \mathcal{B}$, $\widetilde{\textsf{\textnormal{sw}}}(v_q,v_t)$ is the $(\epsilon,\lambda/\beta)$-approximate Swing w.p. of at least $1-\delta$, and for any item $v_t\in \mathcal{B}$, $\widetilde{\textsf{\textnormal{sw}}}(v_q,v_t)$ is the $(\epsilon,\lambda/(1-\beta))$-approximate Swing w.p. of at least $1-\delta$.
\end{theorem}

\stitle{Complexity Analysis}
When Lines 1-3 are executed (i.e., \texttt{USS} is invoked), the worst-case time complexity of \kalgo{} is the same as \texttt{Exact}.
On the contrary, \kalgo{} invokes Algorithms~\ref{alg:candidate} and~\ref{alg:pair} with a total sample size $n=n_f + n_r$.
The overall complexity of \kalgo{} is 
$O\!\left(K\cdot \min\{n,d(v_q)\}+\min\!\left\{n,\frac{d(v_q)(d(v_q)-1)}{2}\right\}\cdot \frac{D(v_q)}{d(v_q)}+n\right)$ as per the complexity analyses in \S\ref{sec:cand-gen} and \S\ref{sec:PR}.

\section{Related Work}

\subsection{Similarity Search on Bipartite Graphs}

Early studies~\cite{katz1953new,haveliwala2002topic,jeh2003scaling} primarily adapted similarity measures originally designed for general graphs or set-based data, such as the Jaccard coefficient, the Katz index~\cite{katz1953new}, and Personalized PageRank (PPR)~\cite{haveliwala2002topic, jeh2003scaling}. To better model bipartite structures, bipartite SimRank~\cite{jeh2002simrank} was proposed, and its optimized variants~\cite{antonellis2008simrank++, dey2020p,yang2024efficient} have been extended to various applications. Another line of research focuses on extending classical PPR to bipartite graphs, leading to the Hidden Personalized PageRank (HPP)~\cite{deng2009generalized} and its symmetric extension, Bidirectional HPP~\cite{yang2022efficient}. Despite their effectiveness, these methods suffer from high computational complexity and scalability issues on large graphs. Approximation frameworks such as \texttt{SimPush}~\cite{shi2020realtime}, \texttt{Approx-BHPP}~\cite{yang2022efficient}, and \texttt{BIRD}~\cite{liu2024bird} integrate push-based methods with iterative refinement to achieve near-linear time complexity while maintaining accuracy guarantees. Recently,~\citet{yang2020large} exploit the user–item–user interaction structure in bipartite graphs, which provides stronger and more reliable signals for capturing product relationships. 
In addition, \citet{Liu2015Stability} study similarity measurement stability under some information loss. \citet{song2024efficient} proposed a temporal similarity measure on temporal-based bipartite graph to capture the activeness of nodes over time.

\subsection{i2i Retrieval in Recommender Systems}

The i2i retrieval has emerged as a fundamental paradigm in recommender systems~\cite{linden2003amazon,li2022query,wu2024hi,bao2025grain,cao2022gift,yang2022task,bao2024beyond,deng2024mmbee,feng2025llm}, where personalized candidates are generated based on item similarity from users' historical interactions. 
Compared with user-based methods, i2i approaches are more scalable due to the relative stability of the item–item similarity matrix, enabling efficient offline computation. They also require limited user history, mitigating cold start for long-tail users~\cite{linden2003amazon,feng2025llm,huang2025comprehensive}.
Early item-based collaborative filtering methods rely on direct user–item interactions to derive item similarity from the rating matrix~\cite{linden2003amazon, sarwar2001item}. Beyond using attribute or side information~\cite{shi2014collaborative,bao2009stacking,liang2010connecting}, the Swing algorithm~\cite{yang2020large} extends this paradigm by leveraging shared interactions to capture more robust item relationships.
Recently, \texttt{LLM-i2i}~\cite{feng2025llm} augments Swing with an LLM-based data generator to mitigate data sparsity and noise. Meanwhile, \texttt{CFG}~\cite{lv2023circle} incorporates Swing's local structural patterns into a Transformer to enhance structural awareness.
With deep learning, retrieval methods have evolved from graph-based similarity modeling to more expressive neural architectures~\cite{wu2021self,xu2022mixture,pei2024rimirec,huang2025comprehensive}. For example, \texttt{SGL}~\cite{wu2021self} introduces self-supervised contrastive learning, \texttt{MVKE}~\cite{xu2022mixture} explores multiple interest-aware embeddings, and \texttt{RimiRec}~\cite{pei2024rimirec} models hierarchical interests.

\section{Experiments}\label{sec:experiment}
This section experimentally evaluates the effectiveness and efficiency of \algo{} and \kalgo{}
for $(\epsilon,\lambda)$-approximate and top-$K$ Swing queries over 8 real datasets.
For the interest of space, we refer interested readers to 
Appendix~\ref{sec:add-exp}
for more experimental details and results.
All algorithms are implemented in C++ and compiled by g++ 8.5 with -O3 optimization, and all experiments are conducted on a Linux machine with an AMD EPYC 7742@2.25GHz CPU and 1TB RAM. 
\update{We further discuss how single-node query processing maps to sharded deployments on large graphs in Appendix~\ref{sec:discuss-shard}.}
For reproducibility, the source code and datasets are available at 
\url{https://github.com/HKBU-LAGAS/K-ASC}.
\subsection{Experimental Setup}\label{sec:experiment-setup}

\begin{table}[!t]
\scriptsize
\renewcommand{\arraystretch}{0.9}
\centering
\caption{Statistics of datasets used in experiments.}
\label{tab:datasets-retrieval}
\vspace{-3ex}
\begin{tabular}{c|rrrc}
\hline
{\bf Dataset} & {\bf $|\U|$} & {\bf $|\I|$} & {\bf $|\EDG|$} & {\bf Item Type} \\ \hline
MovieLens &	6,040 & 3,706 & 1,000,209 & movie \\
Gowalla & 29,858 & 40,981 & 1,027,370 & location \\
AmazonBook & 52,643 & 91,599 & 2,984,108 & book \\
SteamGame & 2,567,538 & 15,474 & 6,889,728 & video game\\
MIND & 876,956	& 97,509	& 18,149,915 & news \\
Twitch & 15,524,309 & 790,100 & 234,422,289 & streamer \\
Yambda & 921,023 & 8,746,689 & 1,339,219,563 & music \\
MAG	& 10,541,560 &	2,784,240 &	1,095,315,106 & paper \\
\hline
\end{tabular}
\end{table}

\stitle{Datasets}
Table~\ref{tab:datasets-retrieval} lists eight real user-item graph datasets for experimentation, which are widely used in recommender system literature~\cite{yang2022efficient,yang2022scalable,deng2009generalized,huang2025comprehensive,liu2024bird,wan2018item,rappaz2021recommendation}.
{\em Gowalla}~\cite{cho2011friendship} models user check-ins at locations. {\em AmazonBook}~\cite{he2016ups} and {\em SteamGame}~\cite{wan2018item} capture user interactions (e.g., reviews and purchases) with books and games, respectively. {\em MovieLens}~\cite{harper2015movielens}, {\em MIND}~\cite{yang2022scalable}, {\em Twitch}~\cite{rappaz2021recommendation}, and {\em Yambda}~\cite{ploshkin2025yambda} represent user interactions with movies, news, streaming content, and music, respectively. {\em MAG}~\cite{yang2022scalable} is an academic bipartite graph.

\stitle{Baselines}
\update{We compare our \algo{} and \kalgo{} for $(\epsilon,\lambda)$-approximate and top-$K$ Swing queries against 7 baselines: 5 Swing-specific algorithms \texttt{Exact}, \texttt{Truncated}, \NMC, \texttt{GNS}, and \texttt{USS}, and 2 general-purpose approximate similarity search algorithms \texttt{All-Pairs}\cite{bayardo2007scaling} and \texttt{WHIMP}\cite{sharma2017hashes}}.
\texttt{Truncated} is a variant of \texttt{Exact} that adopts the heuristic capping the maximum number of neighbors at 600.
\update{As \texttt{All-Pairs} and \texttt{WHIMP} were originally designed for cosine similarity, we adapt them as candidate generation methods for Swing retrieval. Further implementation details are provided in Appendix~\ref{app:baseline_adaptation}.}
For a fair comparison, we fix the failure probability $\delta = 10^{-4}$ and vary the relative error threshold $\epsilon \in \{0.02, 0.05, 0.1, 0.2\}$ to obtain the best trade-off between effectiveness and efficiency in all randomized approaches, i.e., \NMC, \texttt{GNS}, \texttt{USS}, \algo{} and \kalgo{}.
The set intersection and cardinality calculation are implemented by \texttt{QFilter++} in all methods for best performance.

\stitle{Evaluation Protocol}
We construct two query sets for evaluation. The random query set consists of 1,000 items sampled uniformly at random from all items, while the popular query set consists of 200 items sampled from popular items (top 20\% by degree). These two sets are used to evaluate performance on global and high-degree queries, respectively. Ground truth is obtained using \texttt{Exact}, and the top-$K$ results are derived by sorting the exact Swing scores.

We evaluate each method using two metrics, precision and relative error, for $K\in\{20,50,100\}$. Given a query item $v_q$, let $V_K$ denote the ground-truth top-$K$ result set and $\widehat{V}_K$ denote the approximate top-$K$ result set. Precision is defined as
$\mathrm{Precision}@K=\frac{|V_K\cap \widehat{V}_K|}{|V_K|}$,
which measures the fraction of approximate top-$K$ items that also appear in the ground-truth top-$K$ results.
To evaluate score accuracy, we further compute the average relative error over the approximate top-$K$ results. Let $\textsf{\textnormal{sw}}(v_q,v_t)$ and $\widetilde{\textsf{\textnormal{sw}}}(v_q,v_t)$ denote the exact and approximate Swing scores, respectively. The average relative error is defined as
$
Error@K = \frac{1}{|\widehat{V}_K|}\sum_{v_t\in \widehat{V}_K}
\left|
\frac{\widetilde{\textsf{\textnormal{sw}}}(v_q,v_t)-\textsf{\textnormal{sw}}(v_q,v_t)}
{\textsf{\textnormal{sw}}(v_q,v_t)}
\right|
$.

\input{figs/approximate_query}

\subsection{Approximation Query Performance}\label{sec:approx-query}
The results are presented in Figures~\ref{fig:time_pop} and~\ref{fig:time_rand}, which report the average execution time on both popular and random query sets as $\epsilon$ varies over $\{0.02,0.05,0.1,0.2\}$. 
\update{
Note that \texttt{All-Pairs} and \texttt{WHIMP} rely on candidate generation rather than $\epsilon$-controlled Swing approximation, and thus do not provide the $(\epsilon,\lambda)$-approximation guarantee evaluated here.
}
Based on the results, \algo{} consistently achieves the best efficiency across all eight datasets and under all error settings. 
\texttt{Na{\"{\i}}ve-MC} is competitive only when the required error tolerance is loose, but its running time increases rapidly as $\epsilon$ decreases, making it ineffective under strict accuracy requirements. 
\texttt{GNS} is more efficient than \texttt{Na{\"{\i}}ve-MC} in most cases, showing that the grouped strategy improves over naive estimation. 
\update{However, \texttt{GNS} still suffers from the same inefficiency trend as $\epsilon$ decreases, and its advantage diminishes on low-degree queries. In contrast, \texttt{USS} exhibits relatively stable running time, often achieving $2\times$ speedup compared to \texttt{Exact} across different $\epsilon$ and datasets, which suggests that subset sampling achieves a stable improvement.}
More importantly, \algo{} typically achieves about $3\times$ speedup on popular queries and about $5\times$ speedup on random queries compared with \texttt{GNS}. 
The advantage becomes much more pronounced on large graphs, i.e., {\em Yambda} and {\em MAG}, where \algo{} reaches up to $100\times$ speedup over \texttt{Exact} when $\epsilon=0.02$.
Overall, these results show that \algo{} effectively adapts to different error regimes, converges to \texttt{Exact} on small graphs, and delivers substantial acceleration on large datasets where efficient approximate retrieval is most needed.

\input{figs/re-avg_time_precision}

\input{figs/avg_err}

\input{figs/i2i_performance}

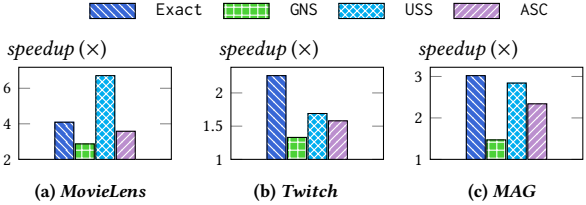
\begin{figure}[!t]
\centering
\begin{small}
\begin{tikzpicture}
   \hspace{3mm}\begin{customlegend}[
        legend entries={\texttt{Exact}, \texttt{GNS},\texttt{USS},\algo{}},
        legend columns=6,
        area legend,
        legend style={at={(0.45,1.25)},anchor=north,draw=none,font=\footnotesize,column sep=0.2cm}]
        \addlegendimage{pattern color=white, preaction={fill, NSCcol1}, pattern={north west lines}} 
        \addlegendimage{pattern color=white, preaction={fill, my_violet}, pattern={grid}} 
        \addlegendimage{pattern color=white, preaction={fill, cyan}, pattern={crosshatch}}
        \addlegendimage{pattern color=white, preaction={fill, my_blue!60},pattern={north east lines}} 
    \end{customlegend}
\end{tikzpicture}
\\[-4pt]
\vspace{-2mm}
\subfloat[{\em MovieLens}]{
\begin{tikzpicture}[scale=1]
\begin{axis}[
    height=2.8cm,
    width=3.6cm,
    xtick=\empty,
    ybar=0.6pt,
    bar width=0.25cm,
    enlarge x limits=0.1,
    ylabel={\em speedup} ($\times$),
    xticklabel=\empty,
    ymin=2,
    ytick={2,4,6},
    xticklabel style = {font=\scriptsize},
    yticklabel style = {font=\scriptsize},
    every node near coord/.append style={font=\scriptsize},
    every axis y label/.style={at={(current axis.north west)},right=5mm,above=0mm},
    legend style={draw=none, at={(1.02,1.02)},anchor=north west,cells={anchor=west},font=\tiny},
    legend image code/.code={ \draw [#1] (0cm,-0.1cm) rectangle (0.3cm,0.15cm);},
    ]

\addplot [pattern color=white, preaction={fill, NSCcol1}, pattern={north west lines}] coordinates {(1,4.09)}; 
\addplot [pattern color=white, preaction={fill, my_violet}, pattern={grid}] coordinates {(1,2.87)}; 
\addplot [pattern color=white, preaction={fill, cyan}, pattern={crosshatch}] coordinates {(1,6.71)}; 
\addplot [pattern color=white, preaction={fill, my_blue!60},pattern={north east lines}] coordinates {(1,3.58)}; 

\end{axis}
\end{tikzpicture}
\hspace{1mm}
}%
\subfloat[{\em Twitch}]{
\begin{tikzpicture}[scale=1]
\begin{axis}[
    height=2.8cm,
    width=3.6cm,
    xtick=\empty,
    ybar=0.6pt,
    bar width=0.25cm,
    enlarge x limits=0.1,
    ylabel={\em speedup} ($\times$),
    xticklabel=\empty,
    ymin=1,
    xticklabel style = {font=\scriptsize},
    yticklabel style = {font=\scriptsize},
    every node near coord/.append style={font=\scriptsize},
    every axis y label/.style={at={(current axis.north west)},right=5mm,above=0mm},
    legend style={draw=none, at={(1.02,1.02)},anchor=north west,cells={anchor=west},font=\tiny},
    legend image code/.code={ \draw [#1] (0cm,-0.1cm) rectangle (0.3cm,0.15cm);},
    ]

\addplot [pattern color=white, preaction={fill, NSCcol1}, pattern={north west lines}] coordinates {(1,2.26)}; 
\addplot [pattern color=white, preaction={fill, my_violet}, pattern={grid}] coordinates {(1,1.33)}; 
\addplot [pattern color=white, preaction={fill, cyan}, pattern={crosshatch}] coordinates {(1,1.69)}; 
\addplot [pattern color=white, preaction={fill, my_blue!60},pattern={north east lines}] coordinates {(1,1.58)}; 

\end{axis}
\end{tikzpicture}
\hspace{1mm}
}%
\subfloat[{\em MAG}]{
\begin{tikzpicture}[scale=1]
\begin{axis}[
    height=2.8cm,
    width=3.6cm,
    xtick=\empty,
    ybar=0.6pt,
    bar width=0.25cm,
    enlarge x limits=0.1,
    ylabel={\em speedup} ($\times$),
    xticklabel=\empty,
    ymin=1,
    ytick={1,2,3,4},
    xticklabel style = {font=\scriptsize},
    yticklabel style = {font=\scriptsize},
    every node near coord/.append style={font=\scriptsize},
    every axis y label/.style={at={(current axis.north west)},right=5mm,above=0mm},
    legend style={draw=none, at={(1.02,1.02)},anchor=north west,cells={anchor=west},font=\tiny},
    legend image code/.code={ \draw [#1] (0cm,-0.1cm) rectangle (0.3cm,0.15cm);},
    ]

\addplot [pattern color=white, preaction={fill, NSCcol1}, pattern={north west lines}] coordinates {(1,3.02)}; 
\addplot [pattern color=white, preaction={fill, my_violet}, pattern={grid}] coordinates {(1,1.47)}; 
\addplot [pattern color=white, preaction={fill, cyan}, pattern={crosshatch}] coordinates {(1,2.84)}; 
\addplot [pattern color=white, preaction={fill, my_blue!60},pattern={north east lines}] coordinates {(1,2.34)}; 

\end{axis}
\end{tikzpicture}
\hspace{1mm}
}%
\vspace{-3mm}
\end{small}
\caption{Effectiveness of \texttt{QFilter++} in Swing Computation.} \label{fig:ab-FSI}
\vspace{-2ex}
\end{figure}

\subsection{Top-$K$ Query Performance}\label{sec:performance-eval}
Figures~\ref{fig:precision-pop} and~\ref{fig:precision-rand} report the query performance, i.e., the trade-off between query time and precision, on the popular and random query sets, respectively, for \algo{}, \kalgo{}, and \update{seven} baselines across eight user-item graphs, with $K=20,50,100$. 
Overall, both \algo{} and \kalgo{} consistently outperform the baselines, achieving near-perfect precision with significantly lower query time, while \kalgo{} further improves efficiency via top-$K$ heuristic optimizations. 
More specifically, on {\em Yambda}, \kalgo{} attains above $99.97\%$ precision for all $K$ values, while reducing query time by over two orders of magnitude compared to \texttt{Exact}. 
Compared with \algo{}, \kalgo{} further reduces query time by about $5\times$, demonstrating the effectiveness of the top-$K$ optimization.
On {\em MAG}, the advantage becomes even more pronounced, where \kalgo{} yields up to three orders of magnitude speedup over \texttt{Exact} with nearly $100\%$ precision, and remains about $2\times$ faster than \algo{}.
\update{
Against the literature-based baselines \texttt{All-Pairs} and \texttt{WHIMP}, \kalgo{} is at least tens of times faster on {\em Yambda} and over two orders of magnitude faster on {\em MAG} at nearly perfect precision, with the speedup exceeding two orders of magnitude on both graphs for random queries.
}

On relatively small graphs such as {\em MIND} and {\em Twitch}, \algo{} already achieves substantial acceleration over \texttt{Exact}, while \kalgo{} further improves upon \algo{} by about $5\times$ on {\em MIND} and $8\times$ on {\em Twitch}, yielding overall speedups of up to two orders of magnitude.
\update{
Meanwhile, compared with \texttt{All-Pairs} and \texttt{WHIMP}, \kalgo{} achieves over $10\times$ speedup for popular queries and up to two orders of magnitude for random queries.
}
When query cost is already small, i.e., {\em Gowalla} and {\em AmazonBook}, both methods naturally converge to \texttt{Exact}, as expected on small graphs. 
Compared with popular queries, random queries exhibit a similar overall trend but are generally more favorable to \algo{} and \kalgo{}. Since random queries are typically less dominated by high-degree items, both methods achieve lower query time while still maintaining near-perfect precision.

Figure~\ref{fig:avgerr} reports the average top-$K$ relative error for different $K$ on the random query set, corresponding to the results in Figure~\ref{fig:time_rand}.
Overall, \algo{} achieves the best accuracy, while \kalgo{} offers a better efficiency and accuracy trade-off. 
This advantage is particularly evident on large graphs such as {\em Yambda} and {\em MAG}. 
Although \kalgo{} has slightly higher relative error due to its top-$K$ optimization, the increase is marginal given the substantial reduction in query time.

\update{
\subsection{Item-to-item Retrieval Performance}\label{sec:i2i-performance}
This section evaluates the empirical effectiveness of Swing for i2i retrieval in comparison with seven representative neighborhood-based, high-order proximity, and embedding-based methods.
Figure~\ref{fig:RecPerformance} illustrates the trade-off between top-$K$ recall and running time on {\em AmazonBook}, {\em SteamGame}, {\em Twitch}, and {\em MAG}, with $K=20,50,100$.
Detailed baseline descriptions and experimental settings are provided in Appendix~\ref{sec:i2i-setting}.
Overall, exact Swing consistently achieves competitive or superior recall across all four datasets and different values of $K$, confirming its effectiveness for i2i retrieval.
Meanwhile, \algo{} preserves nearly identical retrieval quality to exact Swing while substantially reducing the running time.
By exploiting the top-$K$ retrieval objective, \kalgo{} further accelerates retrieval with negligible recall loss, achieving a speedup of over $4{,}000\times$ over exact Swing on {\em MAG} while retaining about $99\%$ of its recall.
}

\subsection{Component Analysis}

\stitle{Effectiveness of \texttt{QFilter++}}
Figure~\ref{fig:ab-FSI} shows the average speedup over baseline implementations without \texttt{QFilter++} across all queries.
To ensure a fair comparison, we implement the set intersection in the baselines using Ankerl's hashmap~\cite{ankerlhashmap}, which is reported to be approximately $2\times$ faster than the standard hashmap~\cite{ankerlhashmap-benchmark}.
Compared to this highly optimized hashmap implementation, \texttt{QFilter++} still consistently achieves superior performance. The gains are more pronounced on {\em MAG}, reaching up to $2.3\times$.
Note that for queries with high-degree nodes, the speedup is even more significant, as \texttt{QFilter++} can prune more unnecessary candidates.

\stitle{Effectiveness of Adaptive Pairwise Refinement}
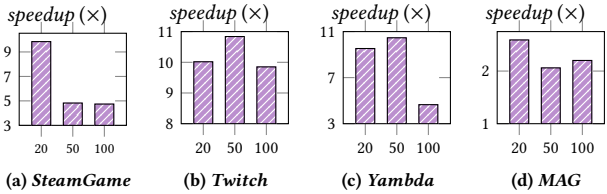
\begin{figure}[!t]
\centering
\begin{small}
\subfloat[{\em SteamGame}]{
\begin{tikzpicture}[scale=1]
\begin{axis}[
    height=2.8cm,
    width=3cm,
    ybar=0.6pt,
    bar width=0.25cm,
    enlarge x limits=0.35,
    ylabel={\em speedup} ($\times$),
    ymin=3,
    ytick={3,5,7,9},
    xtick={1,2,3},
    xticklabels={$20$,$50$,$100$},
    yticklabel style = {font=\scriptsize},
    every node near coord/.append style={font=\scriptsize},
    every axis y label/.style={at={(current axis.north west)},right=5mm,above=0mm},
    legend style={draw=none, at={(1.02,1.02)},anchor=north west,cells={anchor=west},font=\tiny},
    legend image code/.code={ \draw [#1] (0cm,-0.1cm) rectangle (0.3cm,0.15cm);},
    ]

\addplot [pattern color=white, preaction={fill, my_blue!60},pattern={north east lines}] coordinates {(1,9.844237194)
(2,4.82210555)
(3,4.739060511)
};

\end{axis}
\end{tikzpicture}
\hspace{1mm}
}%
\subfloat[{\em Twitch}]{
\begin{tikzpicture}[scale=1]
\begin{axis}[
    height=2.8cm,
    width=3cm,
    xtick=\empty,
    ybar=0.6pt,
    bar width=0.25cm,
    enlarge x limits=0.35,
    ylabel={\em speedup} ($\times$),
    ymin=8,
    ymax=11,
    xtick={1,2,3},
    xticklabels={$20$,$50$,$100$},
    xticklabel style = {font=\scriptsize},
    yticklabel style = {font=\scriptsize},
    every node near coord/.append style={font=\scriptsize},
    every axis y label/.style={at={(current axis.north west)},right=5mm,above=0mm},
    legend style={draw=none, at={(1.02,1.02)},anchor=north west,cells={anchor=west},font=\tiny},
    legend image code/.code={ \draw [#1] (0cm,-0.1cm) rectangle (0.3cm,0.15cm);},
    ]
\addplot [pattern color=white, preaction={fill, my_blue!60},pattern={north east lines}] coordinates {(1,10.01668352)
(2,10.83862144)
(3,9.849863286)
};

\end{axis}
\end{tikzpicture}
\hspace{1mm}
}%
\subfloat[{\em Yambda}]{
\begin{tikzpicture}[scale=1]
\begin{axis}[
    height=2.8cm,
    width=3cm,
    xtick=\empty,
    ybar=0.6pt,
    bar width=0.25cm,
    enlarge x limits=0.35,
    ylabel={\em speedup} ($\times$),
    ymin=3,
    ymax=11,
    ytick={3,7,11},
    xtick={1,2,3},
    xticklabels={$20$,$50$,$100$},
    xticklabel style = {font=\scriptsize},
    yticklabel style = {font=\scriptsize},
    every node near coord/.append style={font=\scriptsize},
    every axis y label/.style={at={(current axis.north west)},right=5mm,above=0mm},
    legend style={draw=none, at={(1.02,1.02)},anchor=north west,cells={anchor=west},font=\tiny},
    legend image code/.code={ \draw [#1] (0cm,-0.1cm) rectangle (0.3cm,0.15cm);},
    ]

\addplot [pattern color=white, preaction={fill, my_blue!60},pattern={north east lines}] coordinates {(1,9.529487268)
(2,10.4606626)
(3,4.652175386)
};

\end{axis}
\end{tikzpicture}
\hspace{1mm}
}%
\subfloat[{\em MAG}]{
\begin{tikzpicture}[scale=1]
\begin{axis}[
    height=2.8cm,
    width=3cm,
    xtick=\empty,
    ybar=0.6pt,
    bar width=0.25cm,
    enlarge x limits=0.35,
    ylabel={\em speedup} ($\times$),
    ymin=1,
    ytick={1,2,3,4},
    xtick={1,2,3},
    xticklabels={$20$,$50$,$100$},
    xticklabel style = {font=\scriptsize},
    yticklabel style = {font=\scriptsize},
    every node near coord/.append style={font=\scriptsize},
    every axis y label/.style={at={(current axis.north west)},right=5mm,above=0mm},
    legend style={draw=none, at={(1.02,1.02)},anchor=north west,cells={anchor=west},font=\tiny},
    legend image code/.code={ \draw [#1] (0cm,-0.1cm) rectangle (0.3cm,0.15cm);},
    ]
\addplot [pattern color=white, preaction={fill, my_blue!60},pattern={north east lines}] coordinates {(1,2.590946502)
(2,2.060209424)
(3,2.201398601)
}; 
\end{axis}
\end{tikzpicture}
\hspace{1mm}
}%
\vspace{-3mm}
\end{small}
\caption{Effectiveness of Algo.~\ref{alg:pair} in \kalgo{} for $K=20,50,100$
} \label{fig:ab-PR}
\end{figure}
Since \kalgo{} enhances \algo{} with a refinement step (Algorithm~\ref{alg:pair}) that re-evaluates borderline items identified by empirical bounds, we study its effectiveness empirically.
Figure~\ref{fig:ab-PR} reports the runtime speedup of \kalgo{} over \algo{} for random queries under comparable top-$K$ precision (both above $99.9\%$) for $K \in \{20,50,100\}$.
The results show that the refinement consistently improves efficiency under the same near-exact accuracy target. 
Specifically, the speedup reaches around one order of magnitude on {\em Twitch} and $4.7$-$10.5\times$ on the billion-scale graph {\em Yambda}.
The speedup slightly decreases as $K$ increases, due to the cost of refining more candidates. 

\subsection{Parameter Analysis}\label{sec:para-analysis}
\begin{figure}[!t]
\centering
\begin{small}
\begin{tikzpicture}
    \begin{customlegend}
    [legend columns=5,
        legend entries={$K=20$,$K=50$,$K=100$},
        legend style={at={(0.45,1.35)},anchor=north,draw=none,font=\footnotesize,column sep=0.2cm}]
    \addlegendimage{line width=0.7mm,mark size=4pt,color=NSCcol1}
    \addlegendimage{line width=0.7mm,mark size=4pt,color=my_violet}
    \addlegendimage{line width=0.7mm,mark size=4pt,color=cyan}
    \end{customlegend}
\end{tikzpicture}
\\[-\lineskip]
\vspace{-4mm}
\subfloat[{\em vary $\rho$}]{\label{fig:vary-time-rho}
\begin{tikzpicture}[scale=1,every mark/.append style={mark size=3pt}]
    \begin{axis}[
        height=2.7cm,
        width=3.6cm,
        ylabel={\it query time(ms)},
        xmin=1, xmax=6,
        ymin=1.2, 
        xtick={1,2,3,4,5,6},
        xticklabel style = {font=\scriptsize},
        yticklabel style = {font=\footnotesize},
        xticklabels={0.3,0.4,0.5,0.6,0.7,0.8},
        every axis y label/.style={font=\footnotesize,at={(current axis.north west)},right=5mm,above=0mm},
        legend style={fill=none,font=\small,at={(0.02,0.99)},anchor=north west,draw=none},
    ]
    \addplot[line width=0.7mm, smooth, color=NSCcol1]  %
    plot coordinates {
        (1, 1.452)
        (2, 1.363)
        (3, 1.398)
        (4, 1.384)
        (5, 1.594)
        (6, 1.401)
    };

    \addplot[line width=0.7mm, smooth, color=my_violet]  %
        plot coordinates {
            (1, 1.5)
            (2, 1.295)
            (3, 1.304)
            (4, 1.245)
            (5, 1.299)
            (6, 1.24)
        };
    
    \addplot[line width=0.7mm, smooth, color=cyan]  %
        plot coordinates {
            (1, 1.597)
            (2, 1.617)
            (3, 1.708)
            (4, 1.445)
            (5, 1.387)
            (6, 1.89)
        };

    \end{axis}
\end{tikzpicture}\hspace{0mm}
}
\subfloat[{\em vary $\beta$}]{\label{fig:vary-time-beta}
\begin{tikzpicture}[scale=1,every mark/.append style={mark size=3pt}]
    \begin{axis}[
        height=2.7cm,
        width=3.6cm,
        ylabel={\it query time(ms)},
        xmin=1, xmax=6,
        ymin=1.2, 
        xtick={1,2,3,4,5,6},
        ytick={1.2,1.6,2.0,2.4},
        xticklabel style = {font=\scriptsize},
        yticklabel style = {font=\footnotesize},
        xticklabels={0.05,0.1,0.15,0.2,0.25,0.3},
        yticklabels={1.2,1.6,2.0,2.4},
        every axis y label/.style={font=\footnotesize,at={(current axis.north west)},right=5mm,above=0mm},
        legend style={fill=none,font=\small,at={(0.02,0.99)},anchor=north west,draw=none},
    ]
    \addplot[line width=0.7mm, smooth, color=NSCcol1]  %
        plot coordinates {
            (1, 1.401)
            (2, 1.651)
            (3, 1.989)
            (4, 2.115)
            (5, 2.103)
            (6, 1.922)

        };

    \addplot[line width=0.7mm, smooth, color=my_violet]  %
        plot coordinates {
            (1, 1.208)
            (2, 1.263)
            (3, 1.282)
            (4, 1.356)
            (5, 1.391)
            (6, 1.475)
        };
    
    \addplot[line width=0.7mm, smooth, color=cyan]  %
        plot coordinates {
            (1, 2.403)
            (2, 2.278)
            (3, 1.549)
            (4, 1.529)
            (5, 1.597)
            (6, 1.541)

        };

    \end{axis}
\end{tikzpicture}\hspace{0mm}
}
\subfloat[{\em vary $\kappa$}]{\label{fig:vary-time-kap}
\begin{tikzpicture}[scale=1,every mark/.append style={mark size=3pt}]
    \begin{axis}[
        height=2.7cm,
        width=3.6cm,
        ylabel={\it query time(ms)},
        xmin=1, xmax=6,
        ymin=1.1, 
        xtick={1,2,3,4,5,6},
        ytick={1.1,1.3,1.5,1.7},
        xticklabel style = {font=\scriptsize},
        yticklabel style = {font=\footnotesize},
        xticklabels={1,3,5,7,10,15},
        every axis y label/.style={font=\footnotesize,at={(current axis.north west)},right=5mm,above=0mm},
        legend style={fill=none,font=\small,at={(0.02,0.99)},anchor=north west,draw=none},
    ]
    \addplot[line width=0.7mm, smooth, color=NSCcol1]  %
        plot coordinates {
            (1, 1.149)
            (2, 1.309)
            (3, 1.215)
            (4, 1.221)
            (5, 1.208)
            (6, 1.238)
    
        };

    \addplot[line width=0.7mm, smooth, color=my_violet]  %
        plot coordinates {
            (1, 1.198)
            (2, 1.256)
            (3, 1.263)
            (4, 1.267)
            (5, 1.253)
            (6, 1.34)

        };
    
    \addplot[line width=0.7mm, smooth, color=cyan]  %
        plot coordinates {
            (1, 1.504)
            (2, 1.583)
            (3, 1.441)
            (4, 1.433)
            (5, 1.421)
            (6, 1.649)

        };

    \end{axis}
\end{tikzpicture}\hspace{0mm}
}
\vspace{-3mm}
\subfloat[{\em vary $\rho$}]{\label{fig:vary-acc-rho}
\begin{tikzpicture}[scale=1,every mark/.append style={mark size=3pt}]
    \begin{axis}[
        height=2.7cm,
        width=3.6cm,
        ylabel={\it precision},
        xmin=1, xmax=6,
        ymin=0.999, ymax=1,
        xtick={1,2,3,4,5,6},
        ytick={0.999,1},
        xticklabel style = {font=\scriptsize},
        yticklabel style = {font=\footnotesize},
        xticklabels={0.3,0.4,0.5,0.6,0.7,0.8},
        yticklabels={99.9,100},
        every axis y label/.style={font=\footnotesize,at={(current axis.north west)},right=3mm,above=0mm},
        legend style={fill=none,font=\small,at={(0.02,0.99)},anchor=north west,draw=none},
    ]
    \addplot[line width=0.7mm, smooth, color=NSCcol1]  %
        plot coordinates {
            (1, 0.999695)
            (2, 0.999797)
            (3, 0.999746)
            (4, 0.999543)
            (5, 0.999644)
            (6, 0.999593)
        };

    \addplot[line width=0.7mm, smooth, color=my_violet]  %
        plot coordinates {
            (1, 0.999553)
            (2, 0.999512)
            (3, 0.999512)
            (4, 0.999533)
            (5, 0.999573)
            (6, 0.999573)
        };
    
    \addplot[line width=0.7mm, smooth, color=cyan]  %
        plot coordinates {
            (1, 0.999207)
            (2, 0.999329)
            (3, 0.999329)
            (4, 0.999258)
            (5, 0.999228)
            (6, 0.999228)
        };

    \end{axis}
\end{tikzpicture}\hspace{0mm}
}
\subfloat[{\em vary $\beta$}]{\label{fig:vary-acc-beta}
\begin{tikzpicture}[scale=1,every mark/.append style={mark size=3pt}]
    \begin{axis}[
        height=2.7cm,
        width=3.6cm,
        ylabel={\it precision},
        xmin=1, xmax=6,
        ymin=0.999, ymax=1,
        xtick={1,2,3,4,5,6},
        ytick={0.999,1},
        xticklabel style = {font=\scriptsize},
        yticklabel style = {font=\footnotesize},
        xticklabels={0.05,0.1,0.15,0.2,0.25,0.3},
        yticklabels={99.9,100},
        every axis y label/.style={font=\footnotesize,at={(current axis.north west)},right=3mm,above=0mm},
        legend style={fill=none,font=\small,at={(0.02,0.99)},anchor=north west,draw=none},
    ]
    \addplot[line width=0.7mm, smooth, color=NSCcol1]  %
        plot coordinates {
            (1, 0.999593)
            (2, 0.999644)
            (3, 0.999535)
            (4, 0.999492)
            (5, 0.999472)
            (6, 0.99939)

        };

    \addplot[line width=0.7mm, smooth, color=my_violet]  %
        plot coordinates {
            (1, 0.999248)
            (2, 0.999573)
            (3, 0.999654)
            (4, 0.999593)
            (5, 0.999492)
            (6, 0.999512)
        };
    
    \addplot[line width=0.7mm, smooth, color=cyan]  %
        plot coordinates {
            (1, 0.999093)
            (2, 0.999012)
            (3, 0.999126)
            (4, 0.999065)
            (5, 0.999147)
            (6, 0.999228)

        };

    \end{axis}
\end{tikzpicture}\hspace{0mm}
}
\subfloat[{\em vary $\kappa$}]{\label{fig:vary-acc-kap}
\begin{tikzpicture}[scale=1,every mark/.append style={mark size=3pt}]
    \begin{axis}[
        height=2.7cm,
        width=3.6cm,
        ylabel={\it precision},
        xmin=1, xmax=6,
        ymin=0.999, ymax=1,
        xtick={1,2,3,4,5,6},
        ytick={0.999,1},
        xticklabel style = {font=\scriptsize},
        yticklabel style = {font=\footnotesize},
        xticklabels={1,3,5,7,10,15},
        yticklabels={99.9,100},
        every axis y label/.style={font=\footnotesize,at={(current axis.north west)},right=3mm,above=0mm},
        legend style={fill=none,font=\small,at={(0.02,0.99)},anchor=north west,draw=none},
    ]
    \addplot[line width=0.7mm, smooth, color=NSCcol1]  %
        plot coordinates {
            (1, 0.999187)
            (2, 0.999593)
            (3, 0.999593)
            (4, 0.999593)
            (5, 0.999441)
            (6, 0.999593)

        };

    \addplot[line width=0.7mm, smooth, color=my_violet]  %
        plot coordinates {
            (1, 0.999248)
            (2, 0.999593)
            (3, 0.999573)
            (4, 0.999695)
            (5, 0.999634)
            (6, 0.999675)

        };
    
    \addplot[line width=0.7mm, smooth, color=cyan]  %
        plot coordinates {
            (1, 0.999136)
            (2, 0.999207)
            (3, 0.999228)
            (4, 0.999339)
            (5, 0.999248)
            (6, 0.999268)

        };

    \end{axis}
\end{tikzpicture}\hspace{0mm}
}
\end{small}
 \vspace{-3mm}
\caption{Performance when varying parameters on {\em MAG}.} \label{fig:vary}
\end{figure}
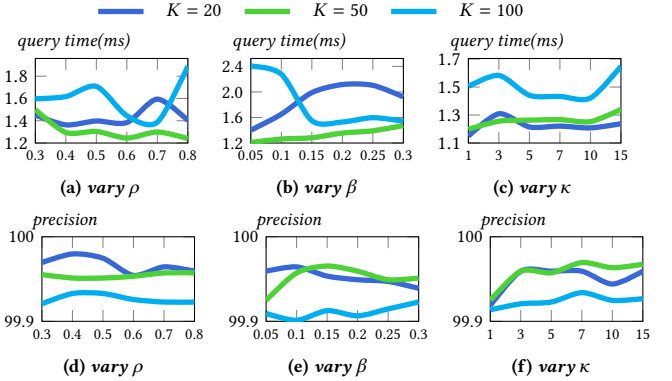
This section investigates the impact of $\rho$, $\beta$, and $\kappa$ in \kalgo{} on the billion-scale graph {\em MAG}, varying each while fixing the others and tuning $\epsilon$ to maintain top-$K$ precision above $99.9\%$, as shown in Figure~\ref{fig:vary}.
In practice, a larger $K$ requires a smaller $\epsilon$ to maintain precision. 
Accordingly, for each $K$, we recommend choosing $\epsilon$ to ensure high recall (with extra $\kappa$), and reducing it only when the target precision cannot be achieved by tuning other parameters.
\update{Detailed parameter settings and cross-dataset sensitivity and transferability analyses are provided in Appendices~\ref{sec:para-setting} and~\ref{sec:complete-parameter-analysis}, respectively.
}

\stitle{Varying $\rho$}
Figure~\ref{fig:vary-time-rho} shows that running time is largely insensitive to $\rho$, remaining stable around 0.2 ms across different values. 
\update{Given the cross-dataset results in Appendix~\ref{sec:complete-parameter-analysis}, we use $\rho=0.8$ as a transferable default for \kalgo{} without dataset-specific tuning.}

\stitle{Varying $\beta$}
Figure~\ref{fig:vary-time-beta} shows that $\beta$ has a significant impact, as it controls the budget for post-refinement. 
A smaller $\beta$ generally leads to lower running time for $K \in \{20,50\}$. 
However, when $K=100$, a slightly larger $\beta$ is needed to ensure accurate identification of borderline items, unless a smaller $\epsilon$ is used.

\stitle{Varying $\kappa$}
Figure~\ref{fig:vary-time-kap} indicates that the performance is relatively insensitive to $\kappa$, since increasing it yields little benefit once adequate recall is achieved.
Moreover, the additional refinement cost remains limited, resulting in only marginal changes in running time.

\begin{table}[!t]
\renewcommand{\arraystretch}{0.9}
\centering
\updatefig{}
\caption{\update{Runtime (ms) for popular query items at top-$20$.}}
\label{tab:extend-KASC}
\vspace{-2ex}
\begin{small}
\begin{tabular}{c|ccc}
\hline
\bf Dataset
& \bf \algo{}
& \bf \kalgo{}
& \bf \kalgoc{} \\
\hline
Movielens
& 257.406
& \textbf{202.361}
& 1327.461 \\

Gowalla
& 1.500
& \textbf{1.055}
& 3.185 \\

AmazonBook
& 3.434
& \textbf{2.626}
& 9.207 \\

SteamGame
& 354.812
& \textbf{79.181}
& 6971.122 \\

MIND
& 236.905
& \textbf{30.067}
& 1807.847 \\

Twitch
& 200.629
& \textbf{21.243}
& 1011.657 \\

Yambda
& 1377.023
& \textbf{319.663}
& 85208.520 \\

MAG
& 39.637
& \textbf{13.584}
& 384.344 \\
\hline
\end{tabular}
\end{small}
\arrayrulecolor{black}
\end{table}

\update{
\subsection{Correctness-Efficiency Trade-off of \kalgo{}}\label{sec:evalue-correctness}
This section investigates the trade-off between correctness and efficiency. In addition to the heuristic \kalgo{}, we investigate a correctness-guaranteed variant, \kalgoc{} certifying the top-$K$ results through confidence-bound separation. The algorithmic details and correctness analysis are deferred to Appendix~\ref{sec:KASC-cor}.
Table~\ref{tab:extend-KASC} reports the average query time of \algo{}, \kalgo{}, and \kalgoc{} for top-$20$ Swing queries with over $99.9\%$ empirical precision. Overall, \kalgo{} is consistently the most efficient, achieving up to $9.44\times$ speedup over \algo{}, whereas \kalgoc{} incurs substantial overhead for correctness certification. 
Although both \kalgo{} and \kalgoc{} rely on Eq.~\eqref{eq:up-low-bound}, \kalgoc{} requires substantially more samples to sufficiently tighten the bounds to separate the $K$-th and $(K+1)$-th items than \kalgo{} does for range-based identification. Thus, we adopt the heuristic \kalgo{} as our primary solution for the accuracy-efficiency trade-off.
}

\section{Conclusion}
This paper presents \algo{} and \kalgo{} for answering approximate and top-$K$ Swing queries. The main ideas include (i) adaptively combining \texttt{GNS} and \texttt{USS} to efficiently process items of varying degrees, and (ii) utilizing a filter-refinement paradigm with heuristic optimizations to drastically minimize sampling overhead. Compared to competing methods, our solutions significantly reduce computational cost and offer rigorous theoretical guarantees on both relative and additive errors. 
Extensive experiments demonstrate that \algo{} and \kalgo{} outperform exact and approximate baselines by orders of magnitude in query efficiency, enabling fast and accurate i2i retrieval on billion-scale user-item graphs.

\begin{acks}
This work was supported by the NSFC (No. 62302414), the Hong Kong RGC ECS grant (No. 22202623), GRF grant (No. 12201826), and YCRG (No. C2003-23Y), and Guangdong and Hong Kong Universities ``1+1+1'' Joint Research Collaboration Scheme, project No.: 2025A0505000002.
\end{acks}

\balance
\bibliographystyle{ACM-Reference-Format}
\bibliography{sample-base}


\begin{thebibliography}{81}


\ifx \showCODEN    \undefined \def \showCODEN     #1{\unskip}     \fi
\ifx \showDOI      \undefined \def \showDOI       #1{#1}\fi
\ifx \showISBNx    \undefined \def \showISBNx     #1{\unskip}     \fi
\ifx \showISBNxiii \undefined \def \showISBNxiii  #1{\unskip}     \fi
\ifx \showISSN     \undefined \def \showISSN      #1{\unskip}     \fi
\ifx \showLCCN     \undefined \def \showLCCN      #1{\unskip}     \fi
\ifx \shownote     \undefined \def \shownote      #1{#1}          \fi
\ifx \showarticletitle \undefined \def \showarticletitle #1{#1}   \fi
\ifx \showURL      \undefined \def \showURL       {\relax}        \fi
\providecommand\bibfield[2]{#2}
\providecommand\bibinfo[2]{#2}
\providecommand\natexlab[1]{#1}
\providecommand\showeprint[2][]{arXiv:#2}

\bibitem[{Aliyun.com}(2025)]%
        {aliyunswing}
\bibfield{author}{\bibinfo{person}{{Aliyun.com}}.} \bibinfo{year}{2025}\natexlab{}.
\newblock \bibinfo{title}{{Official Documentation}}.
\newblock \bibinfo{howpublished}{\url{https://help.aliyun.com/zh/airec/what-is-pai-rec/user-guide/swing-algorithm-tools}}.
\newblock
\newblock
\shownote{Online; accessed 17 June 2026}.


\bibitem[Antonellis et~al\mbox{.}(2008)]%
        {antonellis2008simrank++}
\bibfield{author}{\bibinfo{person}{Ioannis Antonellis}, \bibinfo{person}{Hector Garcia-Molina}, {and} \bibinfo{person}{Chi-Chao Chang}.} \bibinfo{year}{2008}\natexlab{}.
\newblock \showarticletitle{Simrank++ query rewriting through link analysis of the clickgraph (poster)}. In \bibinfo{booktitle}{\emph{Proceedings of the 17th international conference on World Wide Web}}. \bibinfo{pages}{1177--1178}.
\newblock


\bibitem[Audibert et~al\mbox{.}(2007)]%
        {audibert2007tuning}
\bibfield{author}{\bibinfo{person}{Jean-Yves Audibert}, \bibinfo{person}{R{\'e}mi Munos}, {and} \bibinfo{person}{Csaba Szepesv{\'a}ri}.} \bibinfo{year}{2007}\natexlab{}.
\newblock \showarticletitle{Tuning bandit algorithms in stochastic environments}. In \bibinfo{booktitle}{\emph{ALT}}. \bibinfo{pages}{150--165}.
\newblock


\bibitem[Baeza-Yates and Salinger(2005)]%
        {baeza2005experimental}
\bibfield{author}{\bibinfo{person}{Ricardo Baeza-Yates} {and} \bibinfo{person}{Alejandro Salinger}.} \bibinfo{year}{2005}\natexlab{}.
\newblock \showarticletitle{Experimental analysis of a fast intersection algorithm for sorted sequences}. In \bibinfo{booktitle}{\emph{International Symposium on String Processing and Information Retrieval}}. Springer, \bibinfo{pages}{13--24}.
\newblock


\bibitem[Bao et~al\mbox{.}(2025)]%
        {bao2025grain}
\bibfield{author}{\bibinfo{person}{Wei Bao}, \bibinfo{person}{Hao Chen}, \bibinfo{person}{Bang Lin}, \bibinfo{person}{Tao Zhang}, {and} \bibinfo{person}{Chengfu Huo}.} \bibinfo{year}{2025}\natexlab{}.
\newblock \showarticletitle{GRAIN: Group-Reinforced Adaptive Interaction Network for Cold-Start CTR Prediction in E-commerce Search}. In \bibinfo{booktitle}{\emph{Proceedings of the 48th International ACM SIGIR Conference on Research and Development in Information Retrieval}}. \bibinfo{pages}{4275--4279}.
\newblock


\bibitem[Bao et~al\mbox{.}(2024)]%
        {bao2024beyond}
\bibfield{author}{\bibinfo{person}{Wentian Bao}, \bibinfo{person}{Hu Liu}, \bibinfo{person}{Kai Zheng}, \bibinfo{person}{Chao Zhang}, \bibinfo{person}{Shunyu Zhang}, \bibinfo{person}{Enyun Yu}, \bibinfo{person}{Wenwu Ou}, {and} \bibinfo{person}{Yang Song}.} \bibinfo{year}{2024}\natexlab{}.
\newblock \showarticletitle{Beyond relevance: Improving user engagement by personalization for short-video search}.
\newblock \bibinfo{journal}{\emph{arXiv preprint arXiv:2409.11281}} (\bibinfo{year}{2024}).
\newblock


\bibitem[Bao et~al\mbox{.}(2009)]%
        {bao2009stacking}
\bibfield{author}{\bibinfo{person}{Xinlong Bao}, \bibinfo{person}{Lawrence Bergman}, {and} \bibinfo{person}{Rich Thompson}.} \bibinfo{year}{2009}\natexlab{}.
\newblock \showarticletitle{Stacking recommendation engines with additional meta-features}. In \bibinfo{booktitle}{\emph{Proceedings of the third ACM conference on Recommender systems}}. \bibinfo{pages}{109--116}.
\newblock


\bibitem[Bille et~al\mbox{.}(2007)]%
        {bille2007fast}
\bibfield{author}{\bibinfo{person}{Philip Bille}, \bibinfo{person}{Anna Pagh}, {and} \bibinfo{person}{Rasmus Pagh}.} \bibinfo{year}{2007}\natexlab{}.
\newblock \showarticletitle{Fast evaluation of union-intersection expressions}. In \bibinfo{booktitle}{\emph{International Symposium on Algorithms and Computation}}. Springer, \bibinfo{pages}{739--750}.
\newblock


\bibitem[Cao et~al\mbox{.}(2022)]%
        {cao2022gift}
\bibfield{author}{\bibinfo{person}{Yi Cao}, \bibinfo{person}{Sihao Hu}, \bibinfo{person}{Yu Gong}, \bibinfo{person}{Zhao Li}, \bibinfo{person}{Yazheng Yang}, \bibinfo{person}{Qingwen Liu}, {and} \bibinfo{person}{Shouling Ji}.} \bibinfo{year}{2022}\natexlab{}.
\newblock \showarticletitle{Gift: Graph-guided feature transfer for cold-start video click-through rate prediction}. In \bibinfo{booktitle}{\emph{Proceedings of the 31st ACM International Conference on Information \& Knowledge Management}}. \bibinfo{pages}{2964--2973}.
\newblock


\bibitem[Chan et~al\mbox{.}(1982)]%
        {chan1982updating}
\bibfield{author}{\bibinfo{person}{Tony~F Chan}, \bibinfo{person}{Gene~H Golub}, {and} \bibinfo{person}{Randall~J LeVeque}.} \bibinfo{year}{1982}\natexlab{}.
\newblock \showarticletitle{Updating formulae and a pairwise algorithm for computing sample variances}. In \bibinfo{booktitle}{\emph{COMPSTAT 1982 5th symposium held at Toulouse 1982: Part I: Proceedings in Computational Statistics}}. Springer, \bibinfo{pages}{30--41}.
\newblock


\bibitem[Cho et~al\mbox{.}(2011)]%
        {cho2011friendship}
\bibfield{author}{\bibinfo{person}{Eunjoon Cho}, \bibinfo{person}{Seth~A Myers}, {and} \bibinfo{person}{Jure Leskovec}.} \bibinfo{year}{2011}\natexlab{}.
\newblock \showarticletitle{Friendship and mobility: user movement in location-based social networks}. In \bibinfo{booktitle}{\emph{SIGKDD}}. \bibinfo{pages}{1082--1090}.
\newblock


\bibitem[Chung and Lu(2006)]%
        {chung2006concentration}
\bibfield{author}{\bibinfo{person}{Fan Chung} {and} \bibinfo{person}{Linyuan Lu}.} \bibinfo{year}{2006}\natexlab{}.
\newblock \showarticletitle{Concentration inequalities and martingale inequalities: a survey}.
\newblock \bibinfo{journal}{\emph{Internet mathematics}} \bibinfo{volume}{3}, \bibinfo{number}{1} (\bibinfo{year}{2006}), \bibinfo{pages}{79--127}.
\newblock


\bibitem[Deng et~al\mbox{.}(2009)]%
        {deng2009generalized}
\bibfield{author}{\bibinfo{person}{Hongbo Deng}, \bibinfo{person}{Michael~R Lyu}, {and} \bibinfo{person}{Irwin King}.} \bibinfo{year}{2009}\natexlab{}.
\newblock \showarticletitle{A generalized co-hits algorithm and its application to bipartite graphs}. In \bibinfo{booktitle}{\emph{Proceedings of the 15th ACM SIGKDD international conference on Knowledge discovery and data mining}}. \bibinfo{pages}{239--248}.
\newblock


\bibitem[Dey et~al\mbox{.}(2020)]%
        {dey2020p}
\bibfield{author}{\bibinfo{person}{Prasenjit Dey}, \bibinfo{person}{Kunal Goel}, {and} \bibinfo{person}{Rahul Agrawal}.} \bibinfo{year}{2020}\natexlab{}.
\newblock \showarticletitle{P-simrank: Extending simrank to scale-free bipartite networks}. In \bibinfo{booktitle}{\emph{WWW}}. \bibinfo{pages}{3084--3090}.
\newblock


\bibitem[Ding and K{\"o}nig(2011)]%
        {ding2011fast}
\bibfield{author}{\bibinfo{person}{Bolin Ding} {and} \bibinfo{person}{Arnd~Christian K{\"o}nig}.} \bibinfo{year}{2011}\natexlab{}.
\newblock \showarticletitle{Fast set intersection in memory}.
\newblock \bibinfo{journal}{\emph{Proceedings of the VLDB Endowment}} \bibinfo{volume}{4}, \bibinfo{number}{4} (\bibinfo{year}{2011}), \bibinfo{pages}{255--266}.
\newblock


\bibitem[et~al.(2003)]%
        {adamic2003friends}
\bibfield{author}{\bibinfo{person}{Adamic et al.}} \bibinfo{year}{2003}\natexlab{}.
\newblock \showarticletitle{Friends and neighbors on the web}.
\newblock \bibinfo{journal}{\emph{Social networks}} \bibinfo{volume}{25}, \bibinfo{number}{3} (\bibinfo{year}{2003}), \bibinfo{pages}{211--230}.
\newblock


\bibitem[et~al.(2007)]%
        {bayardo2007scaling}
\bibfield{author}{\bibinfo{person}{Bayardo et al.}} \bibinfo{year}{2007}\natexlab{}.
\newblock \showarticletitle{Scaling up all pairs similarity search}. In \bibinfo{booktitle}{\emph{WWW}}. \bibinfo{pages}{131--140}.
\newblock


\bibitem[et~al.(2016)]%
        {barkan2016item2vec}
\bibfield{author}{\bibinfo{person}{Barkan et al.}} \bibinfo{year}{2016}\natexlab{}.
\newblock \showarticletitle{Item2vec: neural item embedding for collaborative filtering}. In \bibinfo{booktitle}{\emph{2016 IEEE 26th international workshop on machine learning for signal processing (MLSP)}}. IEEE, \bibinfo{pages}{1--6}.
\newblock


\bibitem[et~al.(2024)]%
        {deng2024mmbee}
\bibfield{author}{\bibinfo{person}{Deng et al.}} \bibinfo{year}{2024}\natexlab{}.
\newblock \showarticletitle{MMBee: Live Streaming Gift-Sending Recommendations via Multi-Modal Fusion and Behaviour Expansion}. In \bibinfo{booktitle}{\emph{KDD}}. \bibinfo{pages}{4896--4905}.
\newblock


\bibitem[et~al.(2026)]%
        {guo2026onesug}
\bibfield{author}{\bibinfo{person}{Guo et al.}} \bibinfo{year}{2026}\natexlab{}.
\newblock \showarticletitle{Onesug: The unified end-to-end generative framework for e-commerce query suggestion}. In \bibinfo{booktitle}{\emph{AAAI}}, Vol.~\bibinfo{volume}{40}. \bibinfo{pages}{14774--14782}.
\newblock


\bibitem[et~al.(2023a)]%
        {huan2023samd}
\bibfield{author}{\bibinfo{person}{Huan et al.}} \bibinfo{year}{2023}\natexlab{a}.
\newblock \showarticletitle{SAMD: An Industrial Framework for Heterogeneous Multi-Scenario Recommendation}. In \bibinfo{booktitle}{\emph{KDD}}. \bibinfo{pages}{4175--4184}.
\newblock


\bibitem[et~al.(2014)]%
        {kusumoto2014scalable}
\bibfield{author}{\bibinfo{person}{Kusumoto et al.}} \bibinfo{year}{2014}\natexlab{}.
\newblock \showarticletitle{Scalable similarity search for SimRank}. In \bibinfo{booktitle}{\emph{SIGMOD}}. \bibinfo{pages}{325--336}.
\newblock


\bibitem[et~al.(2022)]%
        {li2022query}
\bibfield{author}{\bibinfo{person}{Li et al.}} \bibinfo{year}{2022}\natexlab{}.
\newblock \showarticletitle{Query Rewriting in TaoBao Search}. In \bibinfo{booktitle}{\emph{CIKM}}. \bibinfo{pages}{3262--3271}.
\newblock


\bibitem[et~al.(2023b)]%
        {nguyen2023lightsage}
\bibfield{author}{\bibinfo{person}{Nguyen et al.}} \bibinfo{year}{2023}\natexlab{b}.
\newblock \showarticletitle{LightSAGE: Graph Neural Networks for Large Scale Item Retrieval in Shopee’s Advertisement Recommendation}. In \bibinfo{booktitle}{\emph{RecSys}}. \bibinfo{pages}{334--337}.
\newblock


\bibitem[et~al.(2001)]%
        {sarwar2001item}
\bibfield{author}{\bibinfo{person}{Sarwar et al.}} \bibinfo{year}{2001}\natexlab{}.
\newblock \showarticletitle{Item-based collaborative filtering recommendation algorithms}. In \bibinfo{booktitle}{\emph{WWW}}. \bibinfo{pages}{285--295}.
\newblock


\bibitem[et~al.(2017)]%
        {sharma2017hashes}
\bibfield{author}{\bibinfo{person}{Sharma et al.}} \bibinfo{year}{2017}\natexlab{}.
\newblock \showarticletitle{When hashes met wedges: A distributed algorithm for finding high similarity vectors}. In \bibinfo{booktitle}{\emph{WWW}}. \bibinfo{pages}{431--440}.
\newblock


\bibitem[et~al.(2018)]%
        {wei2018topppr}
\bibfield{author}{\bibinfo{person}{Wei et al.}} \bibinfo{year}{2018}\natexlab{}.
\newblock \showarticletitle{Topppr: top-k personalized pagerank queries with precision guarantees on large graphs}. In \bibinfo{booktitle}{\emph{SIGMOD}}. \bibinfo{pages}{441--456}.
\newblock


\bibitem[Feng et~al\mbox{.}(2025)]%
        {feng2025llm}
\bibfield{author}{\bibinfo{person}{Yinfu Feng}, \bibinfo{person}{Yanjing Wu}, \bibinfo{person}{Rong Xiao}, {and} \bibinfo{person}{Xiaoyi Zen}.} \bibinfo{year}{2025}\natexlab{}.
\newblock \showarticletitle{LLM-I2I: Boost Your Small Item2Item Recommendation Model with Large Language Model}.
\newblock \bibinfo{journal}{\emph{arXiv preprint arXiv:2512.21595}} (\bibinfo{year}{2025}).
\newblock


\bibitem[Han et~al\mbox{.}(2018)]%
        {han2018speeding}
\bibfield{author}{\bibinfo{person}{Shuo Han}, \bibinfo{person}{Lei Zou}, {and} \bibinfo{person}{Jeffrey~Xu Yu}.} \bibinfo{year}{2018}\natexlab{}.
\newblock \showarticletitle{Speeding up set intersections in graph algorithms using simd instructions}. In \bibinfo{booktitle}{\emph{Proceedings of the 2018 International Conference on Management of Data}}. \bibinfo{pages}{1587--1602}.
\newblock


\bibitem[Harper and Konstan(2015)]%
        {harper2015movielens}
\bibfield{author}{\bibinfo{person}{F~Maxwell Harper} {and} \bibinfo{person}{Joseph~A Konstan}.} \bibinfo{year}{2015}\natexlab{}.
\newblock \showarticletitle{The movielens datasets: History and context}.
\newblock \bibinfo{journal}{\emph{Acm transactions on interactive intelligent systems (tiis)}} \bibinfo{volume}{5}, \bibinfo{number}{4} (\bibinfo{year}{2015}), \bibinfo{pages}{1--19}.
\newblock


\bibitem[Haveliwala(2002)]%
        {haveliwala2002topic}
\bibfield{author}{\bibinfo{person}{Taher~H Haveliwala}.} \bibinfo{year}{2002}\natexlab{}.
\newblock \showarticletitle{Topic-sensitive pagerank}. In \bibinfo{booktitle}{\emph{Proceedings of the 11th international conference on World Wide Web}}. \bibinfo{pages}{517--526}.
\newblock


\bibitem[Hayashi(2022)]%
        {hayashi2022rethinking}
\bibfield{author}{\bibinfo{person}{Katsuhiko Hayashi}.} \bibinfo{year}{2022}\natexlab{}.
\newblock \showarticletitle{Rethinking correlation-based item-item similarities for recommender systems}. In \bibinfo{booktitle}{\emph{SIGIR}}. \bibinfo{pages}{2287--2291}.
\newblock


\bibitem[He and McAuley(2016)]%
        {he2016ups}
\bibfield{author}{\bibinfo{person}{Ruining He} {and} \bibinfo{person}{Julian McAuley}.} \bibinfo{year}{2016}\natexlab{}.
\newblock \showarticletitle{Ups and downs: Modeling the visual evolution of fashion trends with one-class collaborative filtering}. In \bibinfo{booktitle}{\emph{TheWebConf}}. \bibinfo{pages}{507--517}.
\newblock


\bibitem[Huang et~al\mbox{.}(2025)]%
        {huang2025comprehensive}
\bibfield{author}{\bibinfo{person}{Junjie Huang}, \bibinfo{person}{Jizheng Chen}, \bibinfo{person}{Jianghao Lin}, \bibinfo{person}{Jiarui Qin}, \bibinfo{person}{Ziming Feng}, \bibinfo{person}{Weinan Zhang}, {and} \bibinfo{person}{Yong Yu}.} \bibinfo{year}{2025}\natexlab{}.
\newblock \showarticletitle{A comprehensive survey on retrieval methods in recommender systems}.
\newblock \bibinfo{journal}{\emph{ACM Transactions on Information Systems}} \bibinfo{volume}{44}, \bibinfo{number}{1} (\bibinfo{year}{2025}), \bibinfo{pages}{1--43}.
\newblock


\bibitem[Hwang and Lin(1972)]%
        {hwang1972simple}
\bibfield{author}{\bibinfo{person}{Frank~K. Hwang} {and} \bibinfo{person}{Shen Lin}.} \bibinfo{year}{1972}\natexlab{}.
\newblock \showarticletitle{A simple algorithm for merging two disjoint linearly ordered sets}.
\newblock \bibinfo{journal}{\emph{SIAM J. Comput.}} \bibinfo{volume}{1}, \bibinfo{number}{1} (\bibinfo{year}{1972}), \bibinfo{pages}{31--39}.
\newblock


\bibitem[Inoue et~al\mbox{.}(2014)]%
        {inoue2014faster}
\bibfield{author}{\bibinfo{person}{Hiroshi Inoue}, \bibinfo{person}{Moriyoshi Ohara}, {and} \bibinfo{person}{Kenjiro Taura}.} \bibinfo{year}{2014}\natexlab{}.
\newblock \showarticletitle{Faster set intersection with SIMD instructions by reducing branch mispredictions}.
\newblock \bibinfo{journal}{\emph{Proceedings of the VLDB Endowment}} \bibinfo{volume}{8}, \bibinfo{number}{3} (\bibinfo{year}{2014}), \bibinfo{pages}{293--304}.
\newblock


\bibitem[Ismailov(2022)]%
        {ismailov2022two}
\bibfield{author}{\bibinfo{person}{Adilzhan Ismailov}.} \bibinfo{year}{2022}\natexlab{}.
\newblock \showarticletitle{Two-Stage Approach for Predicting Fashion Compatibility}.
\newblock  (\bibinfo{year}{2022}).
\newblock


\bibitem[Jeh and Widom(2002)]%
        {jeh2002simrank}
\bibfield{author}{\bibinfo{person}{Glen Jeh} {and} \bibinfo{person}{Jennifer Widom}.} \bibinfo{year}{2002}\natexlab{}.
\newblock \showarticletitle{Simrank: a measure of structural-context similarity}. In \bibinfo{booktitle}{\emph{Proceedings of the eighth ACM SIGKDD international conference on Knowledge discovery and data mining}}. \bibinfo{pages}{538--543}.
\newblock


\bibitem[Jeh and Widom(2003)]%
        {jeh2003scaling}
\bibfield{author}{\bibinfo{person}{Glen Jeh} {and} \bibinfo{person}{Jennifer Widom}.} \bibinfo{year}{2003}\natexlab{}.
\newblock \showarticletitle{Scaling personalized web search}. In \bibinfo{booktitle}{\emph{Proceedings of the 12th international conference on World Wide Web}}. \bibinfo{pages}{271--279}.
\newblock


\bibitem[Katz(1953)]%
        {katz1953new}
\bibfield{author}{\bibinfo{person}{Leo Katz}.} \bibinfo{year}{1953}\natexlab{}.
\newblock \showarticletitle{A new status index derived from sociometric analysis}.
\newblock \bibinfo{journal}{\emph{Psychometrika}} \bibinfo{volume}{18}, \bibinfo{number}{1} (\bibinfo{year}{1953}), \bibinfo{pages}{39--43}.
\newblock


\bibitem[Kersbergen and Schelter(2021)]%
        {kersbergen2021learnings}
\bibfield{author}{\bibinfo{person}{Barrie Kersbergen} {and} \bibinfo{person}{Sebastian Schelter}.} \bibinfo{year}{2021}\natexlab{}.
\newblock \showarticletitle{Learnings from a retail recommendation system on billions of interactions at bol. com}. In \bibinfo{booktitle}{\emph{ICDE}}. IEEE, \bibinfo{pages}{2447--2452}.
\newblock


\bibitem[Lemire et~al\mbox{.}(2016)]%
        {lemire2016simd}
\bibfield{author}{\bibinfo{person}{Daniel Lemire}, \bibinfo{person}{Leonid Boytsov}, {and} \bibinfo{person}{Nathan Kurz}.} \bibinfo{year}{2016}\natexlab{}.
\newblock \showarticletitle{SIMD compression and the intersection of sorted integers}.
\newblock \bibinfo{journal}{\emph{Software: Practice and Experience}} \bibinfo{volume}{46}, \bibinfo{number}{6} (\bibinfo{year}{2016}), \bibinfo{pages}{723--749}.
\newblock


\bibitem[Liang et~al\mbox{.}(2010)]%
        {liang2010connecting}
\bibfield{author}{\bibinfo{person}{Huizhi Liang}, \bibinfo{person}{Yue Xu}, \bibinfo{person}{Yuefeng Li}, \bibinfo{person}{Richi Nayak}, {and} \bibinfo{person}{Xiaohui Tao}.} \bibinfo{year}{2010}\natexlab{}.
\newblock \showarticletitle{Connecting users and items with weighted tags for personalized item recommendations}. In \bibinfo{booktitle}{\emph{Proceedings of the 21st ACM conference on Hypertext and hypermedia}}. \bibinfo{pages}{51--60}.
\newblock


\bibitem[{LibRecommender}(2020)]%
        {librecommed}
\bibfield{author}{\bibinfo{person}{{LibRecommender}}.} \bibinfo{year}{2020}\natexlab{}.
\newblock \bibinfo{title}{{LibRecommender}}.
\newblock \bibinfo{howpublished}{\url{https://librecommender.readthedocs.io/en/latest/api/algorithms/swing.html}}.
\newblock
\newblock
\shownote{Online; accessed 17 June 2026}.


\bibitem[Linden et~al\mbox{.}(2003)]%
        {linden2003amazon}
\bibfield{author}{\bibinfo{person}{Greg Linden}, \bibinfo{person}{Brent Smith}, {and} \bibinfo{person}{Jeremy York}.} \bibinfo{year}{2003}\natexlab{}.
\newblock \showarticletitle{Amazon. com recommendations: Item-to-item collaborative filtering}.
\newblock \bibinfo{journal}{\emph{IEEE Internet computing}} \bibinfo{volume}{7}, \bibinfo{number}{1} (\bibinfo{year}{2003}), \bibinfo{pages}{76--80}.
\newblock


\bibitem[Liu and Luo(2024)]%
        {liu2024bird}
\bibfield{author}{\bibinfo{person}{Haoyu Liu} {and} \bibinfo{person}{Siqiang Luo}.} \bibinfo{year}{2024}\natexlab{}.
\newblock \showarticletitle{Bird: Efficient approximation of bidirectional hidden personalized pagerank}.
\newblock \bibinfo{journal}{\emph{PVLDB}} \bibinfo{volume}{17}, \bibinfo{number}{9} (\bibinfo{year}{2024}), \bibinfo{pages}{2255--2268}.
\newblock


\bibitem[Liu et~al\mbox{.}(2015)]%
        {Liu2015Stability}
\bibfield{author}{\bibinfo{person}{Jianguo Liu}, \bibinfo{person}{Lei Hou}, \bibinfo{person}{Xue Pan}, \bibinfo{person}{Qiang Guo}, {and} \bibinfo{person}{Tao Zhou}.} \bibinfo{year}{2015}\natexlab{}.
\newblock \showarticletitle{Stability of similarity measurements for bipartite networks}.
\newblock \bibinfo{journal}{\emph{Scientific Reports}}  \bibinfo{volume}{6} (\bibinfo{year}{2015}).
\newblock
\urldef\tempurl%
\url{https://doi.org/10.1038/srep18653}
\showDOI{\tempurl}


\bibitem[Luo et~al\mbox{.}(2025)]%
        {luo2025trawl}
\bibfield{author}{\bibinfo{person}{Weiqing Luo}, \bibinfo{person}{Chonggang Song}, \bibinfo{person}{Lingling Yi}, {and} \bibinfo{person}{Gong Cheng}.} \bibinfo{year}{2025}\natexlab{}.
\newblock \showarticletitle{Trawl: External knowledge-enhanced recommendation with llm assistance}. In \bibinfo{booktitle}{\emph{Proceedings of the 34th ACM International Conference on Information and Knowledge Management}}. \bibinfo{pages}{5907--5914}.
\newblock


\bibitem[Lv et~al\mbox{.}(2023)]%
        {lv2023circle}
\bibfield{author}{\bibinfo{person}{Jingsong Lv}, \bibinfo{person}{Hongyang Chen}, \bibinfo{person}{Yao Qi}, {and} \bibinfo{person}{Lei Yu}.} \bibinfo{year}{2023}\natexlab{}.
\newblock \showarticletitle{Circle Feature Graphormer: Can Circle Features Stimulate Graph Transformer?}
\newblock \bibinfo{journal}{\emph{arXiv preprint arXiv:2309.06574}} (\bibinfo{year}{2023}).
\newblock


\bibitem[{Martin Leitner-Ankerl}(2025a)]%
        {ankerlhashmap}
\bibfield{author}{\bibinfo{person}{{Martin Leitner-Ankerl}}.} \bibinfo{year}{2025}\natexlab{a}.
\newblock \bibinfo{title}{{Ankerl's Hashmap}}.
\newblock \bibinfo{howpublished}{\url{https://github.com/martinus/unordered_dense}}.
\newblock
\newblock
\shownote{Online; accessed 17 April 2026}.


\bibitem[{Martin Leitner-Ankerl}(2025b)]%
        {ankerlhashmap-benchmark}
\bibfield{author}{\bibinfo{person}{{Martin Leitner-Ankerl}}.} \bibinfo{year}{2025}\natexlab{b}.
\newblock \bibinfo{title}{{Comprehensive C++ Hashmap Benchmarks 2022}}.
\newblock \bibinfo{howpublished}{\url{https://martin.ankerl.com/2022/08/27/hashmap-bench-01/}}.
\newblock
\newblock
\shownote{Online; accessed 17 April 2026}.


\bibitem[Motwani and Raghavan(1996)]%
        {motwani1996randomized}
\bibfield{author}{\bibinfo{person}{Rajeev Motwani} {and} \bibinfo{person}{Prabhakar Raghavan}.} \bibinfo{year}{1996}\natexlab{}.
\newblock \showarticletitle{Randomized algorithms}.
\newblock \bibinfo{journal}{\emph{ACM Computing Surveys (CSUR)}} \bibinfo{volume}{28}, \bibinfo{number}{1} (\bibinfo{year}{1996}), \bibinfo{pages}{33--37}.
\newblock


\bibitem[Pei et~al\mbox{.}(2024)]%
        {pei2024rimirec}
\bibfield{author}{\bibinfo{person}{Haolei Pei}, \bibinfo{person}{Yuanyuan Xu}, \bibinfo{person}{Yangping Zhu}, {and} \bibinfo{person}{Yuan Nie}.} \bibinfo{year}{2024}\natexlab{}.
\newblock \showarticletitle{RimiRec: Modeling Refined Multi-interest in Hierarchical Structure for Recommendation}. In \bibinfo{booktitle}{\emph{Companion Proceedings of the ACM Web Conference 2024}}. \bibinfo{pages}{746--749}.
\newblock


\bibitem[Ploshkin et~al\mbox{.}(2025)]%
        {ploshkin2025yambda}
\bibfield{author}{\bibinfo{person}{Alexander Ploshkin}, \bibinfo{person}{Vladislav Tytskiy}, \bibinfo{person}{Alexey Pismenny}, \bibinfo{person}{Vladimir Baikalov}, \bibinfo{person}{Evgeny Taychinov}, \bibinfo{person}{Artem Permiakov}, \bibinfo{person}{Daniil Burlakov}, {and} \bibinfo{person}{Eugene Krofto}.} \bibinfo{year}{2025}\natexlab{}.
\newblock \showarticletitle{Yambda-5B—A Large-Scale Multi-Modal Dataset for Ranking and Retrieval}. In \bibinfo{booktitle}{\emph{Proceedings of the Nineteenth ACM Conference on Recommender Systems}}. \bibinfo{pages}{894--901}.
\newblock


\bibitem[Rappaz et~al\mbox{.}(2021)]%
        {rappaz2021recommendation}
\bibfield{author}{\bibinfo{person}{J{\'e}r{\'e}mie Rappaz}, \bibinfo{person}{Julian McAuley}, {and} \bibinfo{person}{Karl Aberer}.} \bibinfo{year}{2021}\natexlab{}.
\newblock \showarticletitle{Recommendation on live-streaming platforms: Dynamic availability and repeat consumption}. In \bibinfo{booktitle}{\emph{Proceedings of the 15th ACM conference on recommender systems}}. \bibinfo{pages}{390--399}.
\newblock


\bibitem[Sanders and Transier(2007)]%
        {sanders2007intersection}
\bibfield{author}{\bibinfo{person}{Peter Sanders} {and} \bibinfo{person}{Frederik Transier}.} \bibinfo{year}{2007}\natexlab{}.
\newblock \showarticletitle{Intersection in integer inverted indices}. In \bibinfo{booktitle}{\emph{2007 Proceedings of the Ninth Workshop on Algorithm Engineering and Experiments (ALENEX)}}. SIAM, \bibinfo{pages}{71--83}.
\newblock


\bibitem[Shi et~al\mbox{.}(2020)]%
        {shi2020realtime}
\bibfield{author}{\bibinfo{person}{Jieming Shi}, \bibinfo{person}{Tianyuan Jin}, \bibinfo{person}{Renchi Yang}, \bibinfo{person}{Xiaokui Xiao}, {and} \bibinfo{person}{Yin Yang}.} \bibinfo{year}{2020}\natexlab{}.
\newblock \showarticletitle{Realtime index-free single source SimRank processing on web-scale graphs}.
\newblock \bibinfo{journal}{\emph{PVLDB}} \bibinfo{volume}{13}, \bibinfo{number}{7} (\bibinfo{year}{2020}), \bibinfo{pages}{966--980}.
\newblock


\bibitem[Shi et~al\mbox{.}(2014)]%
        {shi2014collaborative}
\bibfield{author}{\bibinfo{person}{Yue Shi}, \bibinfo{person}{Martha Larson}, {and} \bibinfo{person}{Alan Hanjalic}.} \bibinfo{year}{2014}\natexlab{}.
\newblock \showarticletitle{Collaborative filtering beyond the user-item matrix: A survey of the state of the art and future challenges}.
\newblock \bibinfo{journal}{\emph{ACM Computing Surveys (CSUR)}} \bibinfo{volume}{47}, \bibinfo{number}{1} (\bibinfo{year}{2014}), \bibinfo{pages}{1--45}.
\newblock


\bibitem[Song et~al\mbox{.}(2024)]%
        {song2024efficient}
\bibfield{author}{\bibinfo{person}{Yifan Song}, \bibinfo{person}{Xiaolong Chen}, \bibinfo{person}{Wenqing Lin}, \bibinfo{person}{Jia Li}, \bibinfo{person}{Chen Zhang}, \bibinfo{person}{Yan Zhou}, \bibinfo{person}{Lei Chen}, {and} \bibinfo{person}{Jing Tang}.} \bibinfo{year}{2024}\natexlab{}.
\newblock \showarticletitle{Efficient Graph Embedding Generation and Update for Large-Scale Temporal Graph}.
\newblock \bibinfo{journal}{\emph{Proceedings of the VLDB Endowment}} \bibinfo{volume}{18}, \bibinfo{number}{4} (\bibinfo{year}{2024}), \bibinfo{pages}{929--942}.
\newblock


\bibitem[Tsirogiannis et~al\mbox{.}(2009)]%
        {tsirogiannis2009improving}
\bibfield{author}{\bibinfo{person}{Dimitris Tsirogiannis}, \bibinfo{person}{Sudipto Guha}, {and} \bibinfo{person}{Nick Koudas}.} \bibinfo{year}{2009}\natexlab{}.
\newblock \showarticletitle{Improving the performance of list intersection}.
\newblock \bibinfo{journal}{\emph{Proceedings of the VLDB Endowment}} \bibinfo{volume}{2}, \bibinfo{number}{1} (\bibinfo{year}{2009}), \bibinfo{pages}{838--849}.
\newblock


\bibitem[Wan et~al\mbox{.}(2023)]%
        {wan2023interval}
\bibfield{author}{\bibinfo{person}{B Wan}, \bibinfo{person}{Z Wu}, \bibinfo{person}{M Han}, {and} \bibinfo{person}{M Wan}.} \bibinfo{year}{2023}\natexlab{}.
\newblock \showarticletitle{Interval-valued q-rung orthopair fuzzy Weber operator and its group decision-making application}.
\newblock  (\bibinfo{year}{2023}).
\newblock


\bibitem[Wan and McAuley(2018)]%
        {wan2018item}
\bibfield{author}{\bibinfo{person}{Mengting Wan} {and} \bibinfo{person}{Julian McAuley}.} \bibinfo{year}{2018}\natexlab{}.
\newblock \showarticletitle{Item recommendation on monotonic behavior chains}. In \bibinfo{booktitle}{\emph{Proceedings of the 12th ACM conference on recommender systems}}. \bibinfo{pages}{86--94}.
\newblock


\bibitem[Wang et~al\mbox{.}(2018)]%
        {wang2018billion}
\bibfield{author}{\bibinfo{person}{Jizhe Wang}, \bibinfo{person}{Pipei Huang}, \bibinfo{person}{Huan Zhao}, \bibinfo{person}{Zhibo Zhang}, \bibinfo{person}{Binqiang Zhao}, {and} \bibinfo{person}{Dik~Lun Lee}.} \bibinfo{year}{2018}\natexlab{}.
\newblock \showarticletitle{Billion-scale commodity embedding for e-commerce recommendation in alibaba}. In \bibinfo{booktitle}{\emph{Proceedings of the 24th ACM SIGKDD international conference on knowledge discovery \& data mining}}. \bibinfo{pages}{839--848}.
\newblock


\bibitem[Wang et~al\mbox{.}(2026)]%
        {wang2026pi2i}
\bibfield{author}{\bibinfo{person}{Shaoqing Wang}, \bibinfo{person}{Yingcai Ma}, \bibinfo{person}{Kairui Fu}, \bibinfo{person}{Ziyang Wang}, \bibinfo{person}{Dunxian Huang}, \bibinfo{person}{Yuliang Yan}, {and} \bibinfo{person}{Jian Wu}.} \bibinfo{year}{2026}\natexlab{}.
\newblock \showarticletitle{PI2I: A Personalized Item-Based Collaborative Filtering Retrieval Framework}.
\newblock \bibinfo{journal}{\emph{arXiv preprint arXiv:2601.16815}} (\bibinfo{year}{2026}).
\newblock


\bibitem[Welford(1962)]%
        {welford1962note}
\bibfield{author}{\bibinfo{person}{Barry~Payne Welford}.} \bibinfo{year}{1962}\natexlab{}.
\newblock \showarticletitle{Note on a method for calculating corrected sums of squares and products}.
\newblock \bibinfo{journal}{\emph{Technometrics}} \bibinfo{volume}{4}, \bibinfo{number}{3} (\bibinfo{year}{1962}), \bibinfo{pages}{419--420}.
\newblock


\bibitem[Wu et~al\mbox{.}(2021)]%
        {wu2021self}
\bibfield{author}{\bibinfo{person}{Jiancan Wu}, \bibinfo{person}{Xiang Wang}, \bibinfo{person}{Fuli Feng}, \bibinfo{person}{Xiangnan He}, \bibinfo{person}{Liang Chen}, \bibinfo{person}{Jianxun Lian}, {and} \bibinfo{person}{Xing Xie}.} \bibinfo{year}{2021}\natexlab{}.
\newblock \showarticletitle{Self-supervised graph learning for recommendation}. In \bibinfo{booktitle}{\emph{Proceedings of the 44th international ACM SIGIR conference on research and development in information retrieval}}. \bibinfo{pages}{726--735}.
\newblock


\bibitem[Wu et~al\mbox{.}(2024)]%
        {wu2024hi}
\bibfield{author}{\bibinfo{person}{Yanjing Wu}, \bibinfo{person}{Yinfu Feng}, \bibinfo{person}{Jian Wang}, \bibinfo{person}{Wenji Zhou}, \bibinfo{person}{Yunan Ye}, \bibinfo{person}{Rong Xiao}, {and} \bibinfo{person}{Jun Xiao}.} \bibinfo{year}{2024}\natexlab{}.
\newblock \showarticletitle{Hi-gen: Generative retrieval for large-scale personalized e-commerce search}.
\newblock \bibinfo{journal}{\emph{arXiv preprint arXiv:2404.15675}} (\bibinfo{year}{2024}).
\newblock


\bibitem[Xia et~al\mbox{.}(2026)]%
        {xia2026qarm}
\bibfield{author}{\bibinfo{person}{Tian Xia}, \bibinfo{person}{Jiaqi Zhang}, \bibinfo{person}{Yueyang Liu}, \bibinfo{person}{Hongjian Dou}, \bibinfo{person}{Tingya Yin}, \bibinfo{person}{Jiangxia Cao}, \bibinfo{person}{Xulei Liang}, \bibinfo{person}{Tianlu Xie}, \bibinfo{person}{Lihao Liu}, \bibinfo{person}{Xiang Chen}, {et~al\mbox{.}}} \bibinfo{year}{2026}\natexlab{}.
\newblock \showarticletitle{QARM V2: Quantitative Alignment Multi-Modal Recommendation for Reasoning User Sequence Modeling}.
\newblock \bibinfo{journal}{\emph{arXiv preprint arXiv:2602.08559}} (\bibinfo{year}{2026}).
\newblock


\bibitem[Xu et~al\mbox{.}(2023)]%
        {xu2023multi}
\bibfield{author}{\bibinfo{person}{Jingcao Xu}, \bibinfo{person}{Chaokun Wang}, \bibinfo{person}{Cheng Wu}, \bibinfo{person}{Yang Song}, \bibinfo{person}{Kai Zheng}, \bibinfo{person}{Xiaowei Wang}, \bibinfo{person}{Changping Wang}, \bibinfo{person}{Guorui Zhou}, {and} \bibinfo{person}{Kun Gai}.} \bibinfo{year}{2023}\natexlab{}.
\newblock \showarticletitle{Multi-behavior Self-supervised Learning for Recommendation}.
\newblock \bibinfo{journal}{\emph{arXiv preprint arXiv:2305.18238}} (\bibinfo{year}{2023}).
\newblock


\bibitem[Xu et~al\mbox{.}(2022)]%
        {xu2022mixture}
\bibfield{author}{\bibinfo{person}{Zhenhui Xu}, \bibinfo{person}{Meng Zhao}, \bibinfo{person}{Liqun Liu}, \bibinfo{person}{Lei Xiao}, \bibinfo{person}{Xiaopeng Zhang}, {and} \bibinfo{person}{Bifeng Zhang}.} \bibinfo{year}{2022}\natexlab{}.
\newblock \showarticletitle{Mixture of virtual-kernel experts for multi-objective user profile modeling}. In \bibinfo{booktitle}{\emph{Proceedings of the 28th ACM SIGKDD Conference on Knowledge Discovery and Data Mining}}. \bibinfo{pages}{4257--4267}.
\newblock


\bibitem[Xue et~al\mbox{.}(2026)]%
        {xue2026generative}
\bibfield{author}{\bibinfo{person}{Ben Xue}, \bibinfo{person}{Dan Liu}, \bibinfo{person}{Lixiang Wang}, \bibinfo{person}{Mingjie Sun}, \bibinfo{person}{Peng Wang}, \bibinfo{person}{Pengfei Zhang}, \bibinfo{person}{Shaoyun Shi}, \bibinfo{person}{Tianyu Xu}, \bibinfo{person}{Yunhao Sha}, \bibinfo{person}{Zhiqiang Liu}, {et~al\mbox{.}}} \bibinfo{year}{2026}\natexlab{}.
\newblock \showarticletitle{Generative Recommendation for Large-Scale Advertising}.
\newblock \bibinfo{journal}{\emph{arXiv preprint arXiv:2602.22732}} (\bibinfo{year}{2026}).
\newblock


\bibitem[Xue et~al\mbox{.}(2025)]%
        {xue2025e2e}
\bibfield{author}{\bibinfo{person}{Rui Xue}, \bibinfo{person}{Shichao Zhu}, \bibinfo{person}{Liang Qin}, \bibinfo{person}{Guangmou Pan}, \bibinfo{person}{Yang Song}, {and} \bibinfo{person}{Tianfu Wu}.} \bibinfo{year}{2025}\natexlab{}.
\newblock \showarticletitle{E2E-GRec: An End-to-End Joint Training Framework for Graph Neural Networks and Recommender Systems}.
\newblock \bibinfo{journal}{\emph{arXiv preprint arXiv:2511.20564}} (\bibinfo{year}{2025}).
\newblock


\bibitem[Yang et~al\mbox{.}(2022a)]%
        {yang2022task}
\bibfield{author}{\bibinfo{person}{Jieyu Yang}, \bibinfo{person}{Zhaoxin Huan}, \bibinfo{person}{Yong He}, \bibinfo{person}{Ke Ding}, \bibinfo{person}{Liang Zhang}, \bibinfo{person}{Xiaolu Zhang}, \bibinfo{person}{Jun Zhou}, {and} \bibinfo{person}{Linjian Mo}.} \bibinfo{year}{2022}\natexlab{a}.
\newblock \showarticletitle{Task Similarity Aware Meta Learning for Cold-Start Recommendation}. In \bibinfo{booktitle}{\emph{Proceedings of the 31st ACM International Conference on Information \& Knowledge Management}}. \bibinfo{pages}{4630--4634}.
\newblock


\bibitem[Yang et~al\mbox{.}(2022b)]%
        {yang2022hicf}
\bibfield{author}{\bibinfo{person}{Menglin Yang}, \bibinfo{person}{Zhihao Li}, \bibinfo{person}{Min Zhou}, \bibinfo{person}{Jiahong Liu}, {and} \bibinfo{person}{Irwin King}.} \bibinfo{year}{2022}\natexlab{b}.
\newblock \showarticletitle{Hicf: Hyperbolic informative collaborative filtering}. In \bibinfo{booktitle}{\emph{Proceedings of the 28th ACM SIGKDD conference on knowledge discovery and data mining}}. \bibinfo{pages}{2212--2221}.
\newblock


\bibitem[Yang(2022)]%
        {yang2022efficient}
\bibfield{author}{\bibinfo{person}{Renchi Yang}.} \bibinfo{year}{2022}\natexlab{}.
\newblock \showarticletitle{Efficient and effective similarity search over bipartite graphs}. In \bibinfo{booktitle}{\emph{WWW}}. \bibinfo{pages}{308--318}.
\newblock


\bibitem[Yang and Shi(2024)]%
        {yang2024efficient}
\bibfield{author}{\bibinfo{person}{Renchi Yang} {and} \bibinfo{person}{Jieming Shi}.} \bibinfo{year}{2024}\natexlab{}.
\newblock \showarticletitle{Efficient high-quality clustering for large bipartite graphs}.
\newblock \bibinfo{journal}{\emph{Proceedings of the ACM on Management of Data}} \bibinfo{volume}{2}, \bibinfo{number}{1} (\bibinfo{year}{2024}), \bibinfo{pages}{1--27}.
\newblock


\bibitem[Yang et~al\mbox{.}(2022c)]%
        {yang2022scalable}
\bibfield{author}{\bibinfo{person}{Renchi Yang}, \bibinfo{person}{Jieming Shi}, \bibinfo{person}{Keke Huang}, {and} \bibinfo{person}{Xiaokui Xiao}.} \bibinfo{year}{2022}\natexlab{c}.
\newblock \showarticletitle{Scalable and Effective Bipartite Network Embedding}. In \bibinfo{booktitle}{\emph{Proceedings of the 2022 International Conference on Management of Data}}. \bibinfo{pages}{1977--1991}.
\newblock


\bibitem[Yang et~al\mbox{.}(2020)]%
        {yang2020large}
\bibfield{author}{\bibinfo{person}{Xiaoyong Yang}, \bibinfo{person}{Yadong Zhu}, \bibinfo{person}{Yi Zhang}, \bibinfo{person}{Xiaobo Wang}, {and} \bibinfo{person}{Quan Yuan}.} \bibinfo{year}{2020}\natexlab{}.
\newblock \showarticletitle{Large scale product graph construction for recommendation in e-commerce}.
\newblock \bibinfo{journal}{\emph{arXiv preprint arXiv:2010.05525}} (\bibinfo{year}{2020}).
\newblock


\bibitem[Zhang et~al\mbox{.}(2023)]%
        {zhang2023contrastive}
\bibfield{author}{\bibinfo{person}{Bo-Wen Zhang}, \bibinfo{person}{Yan Yan}, {and} \bibinfo{person}{Jiapei Yu}.} \bibinfo{year}{2023}\natexlab{}.
\newblock \showarticletitle{Contrastive Learning of Sentence Embeddings in Product Search}. In \bibinfo{booktitle}{\emph{ICASSP 2023-2023 IEEE International Conference on Acoustics, Speech and Signal Processing (ICASSP)}}. IEEE, \bibinfo{pages}{1--5}.
\newblock


\bibitem[Zhang et~al\mbox{.}(2024)]%
        {zhang2024bigset}
\bibfield{author}{\bibinfo{person}{Shiding Zhang}, \bibinfo{person}{Jianye Yang}, \bibinfo{person}{Wenjie Zhang}, \bibinfo{person}{Shiyu Yang}, \bibinfo{person}{Ying Zhang}, {and} \bibinfo{person}{Xuemin Lin}.} \bibinfo{year}{2024}\natexlab{}.
\newblock \showarticletitle{BigSet: An Efficient Set Intersection Approach}.
\newblock \bibinfo{journal}{\emph{IEEE Transactions on Knowledge and Data Engineering}} \bibinfo{volume}{36}, \bibinfo{number}{12} (\bibinfo{year}{2024}), \bibinfo{pages}{7677--7691}.
\newblock


\bibitem[Zzh et~al\mbox{.}(2022)]%
        {zzh2022industrial}
\bibfield{author}{\bibinfo{person}{Zzh}, \bibinfo{person}{Wei Zhang}, {and} \bibinfo{person}{Wentao}.} \bibinfo{year}{2022}\natexlab{}.
\newblock \showarticletitle{Industrial Solution in Fashion-domain Recommendation by an Efficient Pipeline using GNN and Lightgbm}.
\newblock In \bibinfo{booktitle}{\emph{Proceedings of the Recommender Systems Challenge 2022}}. \bibinfo{pages}{45--49}.
\newblock


\end{thebibliography}

\appendix
\section{The \texttt{QFilter++} Algorithm}\label{sec:FSI}

In the following, we detail the efficiency optimizations for BSR construction and intersection cardinality computation implemented in \texttt{QFilter++}, extending the design of \texttt{QFilter}~\cite{han2018speeding} to user-item graphs.

\begin{table}[h]
\centering
\caption{BSR Construction Overhead.}
\label{tab:BSRcost}
\vspace{-2ex}
\begin{small}
\addtolength{\tabcolsep}{-0.4em}
\resizebox{\columnwidth}{!}{
\begin{tabular}{c|*{4}{ccc}}
\toprule
\multirow{2}{*}{\makecell{Order\\ Scheme}} & 
\multicolumn{2}{c}{Twitch} & 
\multicolumn{2}{c}{Yambda} &
\multicolumn{2}{c}{Mag} &\\
 \cmidrule(lr){2-3} \cmidrule(lr){4-5} \cmidrule(lr){6-7} 
\multicolumn{1}{c|}{}  & \multicolumn{1}{c}{Prep.(s)} & \multicolumn{1}{c}{Time(ms)} & \multicolumn{1}{c}{Prep.(s)} & \multicolumn{1}{c}{Time(ms)} & \multicolumn{1}{c}{Prep.(s)} & \multicolumn{1}{c}{Time(ms)} \\
\midrule
 \texttt{Ori} 
& 31.06 & 0.813 
& 293.6 & 11.262 
& 126.4 & 11.759\\

 \texttt{Deg}
& 31.55 & 0.825 
& 291.7 & 12.658
& 170.7 & 8.019 \\

\texttt{GRO}
& 99.55 & 0.680 
& 8635 & 7.397 
& 1470	& 6.968\\

\texttt{Ours}
& 44.07 & 0.681 
& 1285 & 7.232 
& 216.0 & 6.816\\

\bottomrule
\end{tabular}
}
\end{small}
\end{table}

\subsection{Optimized BSR Construction}
To address the high preprocessing cost of BSR construction, we propose a simple yet effective sampling-based optimization strategy. As shown in Table~\ref{tab:BSRcost}, the BSR structure is sensitive to node ordering before construction, and using the original ordering can be up to twice as slow as using GRO~\cite{han2018speeding}. However, GRO incurs a high computational cost of  $O(\log|\I| \cdot \sum_{v\in \I}\sum_{u\in\U(v)}d(v)d(u))$ as it greedily selects locally optimal items.
In contrast, we observe that the state chunk size of BSR is typically small (e.g., 32), making it inefficient to spend substantial computation on estimating exact contributions for local optimization. Instead, we propose to approximate the contribution of items via random sampling, namely LightGRO, which achieves comparable effectiveness while significantly reducing computational overhead.

Specifically, we begin with an adaptive degree-based sorting procedure (Lines 1–5). Under the power-law assumption, items are partitioned into a low-degree set $\mathcal{L}$ and a high-degree set $\mathcal{H}$ (Line 2) using an adaptive threshold $\tau_d$ (Line 1). We sort $\mathcal{L}$ with \texttt{BucketSort} and $\mathcal{H}$ with \texttt{QuickSort}, achieving near-linear complexity in practice. 
The subsequent steps follow those of GRO~\cite{han2018speeding}. When filling an empty bucket, we select the remaining item $v$ with the largest degree from $\mathcal{M}_0$ (Line 8), place it into a new bucket, and remove it from $\mathcal{M}_0$ (Line 10). Otherwise, if the current bucket is not empty, we select the item $v$ with the largest weight from the top of $\mathcal{Q}$, where $\mathcal{Q}$ is a priority queue. Instead of examining all neighbors, we sample a set of $S$ neighbors from $\U(v)$. We then accumulate weights from the corresponding items $v_t$ (Lines 13-14). To reduce the overhead of frequent priority queue updates, the updates are batched and applied together (Line 15). After reordering, we construct the BSR with bucket size 32 following~\cite{han2018speeding} (Lines 17–21).

\subsection{Light Intersection Cardinality Computation}
However, computing the intersection cardinality via full set intersection requires materializing and storing the explicit results, which may incur substantially higher costs than direct cardinality computation, often dominating the overall execution time due to extra memory accesses and result write-back overhead. As Figure~\ref{fig:speedup} shows, intersection cardinality computation can achieve up to $3\times$ speedup.

Specifically, we specialize the original \texttt{QFilter} (Algorithm 4 in~\cite{han2018speeding}) for intersection cardinality computation by removing unnecessary result materialization. 
To illustrate the modification, we revisit the structure of \texttt{QFilter}. The original procedure, based on the pseudo-code in Appendix~A.3 of~\cite{han2018speeding}, consists of three phases, namely candidate pruning (Steps S1-S2), candidate verification (Steps S3-S4), and result materialization (Steps S5-S6), followed by pointer advancement (Step S7).
In our cardinality-only variant, the pruning (Steps S1-S2) and verification (Steps S3-S4) stages are retained, since the algorithm must still filter out impossible matches through byte-checking, align candidate pairs, and verify both base equality and non-empty state intersection. This design choice is also consistent with the analysis in Section~4.2 of~\cite{han2018speeding}, which shows that byte-checking is highly effective because “no-match” cases are common whereas “multi-match” cases are rare. 
By contrast, the materialization stage (Steps S5-S6) is removed. In the original algorithm, these steps gather valid result chunks and write them to output arrays. For cardinality computation, this materialization is unnecessary. Instead, after obtaining the base-equality mask and the non-empty intersection mask, we directly count the number of set bits in the resulting bitmask for each block and accumulate them to obtain the final cardinality.
Finally, Step S7 remains unchanged, since the pointer advancement logic depends only on the ordering of base values and is independent of whether the intersection results are materialized.
Overall, this specialization preserves the filtering and alignment mechanisms of \texttt{QFilter} while removing the result construction phase. 

\begin{algorithm}[h]
\caption{Optimized BSR Construction for \texttt{QFilter++}}\label{alg:FastSet}
\KwIn{$\G=(\U\cup\I,\EDG)$, threshold $\tau$}
\KwOut{$\textsf{BSR}$}
$\tau_d \gets \texttt{nth-element}(\{d(v)\mid v\in \I\}, |\I| - \lceil \frac{|\I|}{\log|\I|} \rceil)$\;
$\mathcal{L}\gets \{v_i\in \I\mid d(v_i) < \tau_d\}$, $\mathcal{H}\gets \I\setminus\mathcal{L}$\;
sort $\mathcal{L}$ by descending degree using \texttt{BucketSort}\;
sort $\mathcal{H}$ by descending degree using \texttt{QuickSort}\;
$\mathcal{M}_0 = \mathcal{H} \mathbin\Vert \mathcal{L} $\;
\For{$i \gets 0$ to $n-1$}{
    \If{$i \bmod 32 = 0$}{
        $v \gets \mathcal{M}_0.\texttt{front}()$, $\mathcal{Q}\gets \emptyset$, $\mathcal{W}\gets \emptyset$\;
    }\lElse{
        $v \gets \mathcal{Q}.\texttt{top}()$, $\mathcal{Q}.\texttt{pop}()$
    }    
    $\mathcal{M}(v) \gets i$, $\mathcal{M}_0.\texttt{delete}(v)$\;
    $S \gets \texttt{Sample}(\U(v), \min\{\tau_s,|\U(v)|\})$, $\mathcal{P} \gets \emptyset$\;
    \For{$u \in S$}{
        \For{$v_t \in \I(u) \cap \mathcal{M}_0$}{
            $\mathcal{P} \gets \mathcal{P} \cup \{v_t\}$, 
            $\mathcal{W}(v_t) \gets \mathcal{W}(v_t) + 1$\;
        }
    }
    \lFor{$v_t \in \mathcal{P}$}{
        $\mathcal{Q}(v_t)=\mathcal{W}(v_t)$
    }
}

\For{$u\in \U$}{
    $\textsf{BSR}(u) \gets \emptyset$\;
    \For{$v_i\in \I(u)$}{
        ${b} \gets \mathcal{M}(v_i)\gg 5;\ {s} \gets \mathcal{M}(v_i)\mathbin{\&} 31$\;
        \lIf{$b \notin \textnormal{\textsf{BSR}}(u)$}{
            $\textsf{BSR}(u) \gets \textsf{BSR}(u) \cup (b,2^{s})$
        }
        \lElse{
            $\textsf{BSR}(u,b) \gets \textsf{BSR}(u,b) \mathbin{|}  2^{s}$
        }
    }
}
\end{algorithm}

\subsection{\update{Dynamic Update Cost}}\label{sec:qft-update}
\update{
Table~\ref{tab:update-time-max} reports the cost of maintaining \texttt{QFilter++} under dynamic edge updates. For each dataset, we randomly sample 10000 edges for update evaluation. We first compute the item ordering and build the corresponding BSR representation during preprocessing. For each subsequent single-edge update, we keep the ordering index unchanged and only update the sorted neighbor list and locally reconstruct the BSR representation of the affected user.
Overall, the update overhead is small across the evaluated datasets. Except for {\em Yambda}, the average insertion and deletion latencies are below $9\,\mu$s and $8\,\mu$s, respectively, while the maximum costs remain below $1.3$ ms for insertion and $0.1$ ms for deletion. {\em Yambda} incurs higher costs because it has a substantially larger average user degree. Since each update locally reconstructs the compressed representation of the affected user, the update cost increases with user degree. This also explains the gap between the average and maximum latency, as the latter is dominated by a few extremely high-degree users.
}

\begin{table}[!t]
\renewcommand{\arraystretch}{0.9}
\centering
\updatefig{}
\caption{\update{Update time of \texttt{QFilter++} for edge insertions and deletions.}}
\label{tab:update-time-max}
\vspace{-2ex}
\begin{small}
\begin{tabular}{c|cc|cc}
\hline
\multirow{2}{*}{\bf Dataset}
& \multicolumn{2}{c|}{\bf Average ($\mu$s)}
& \multicolumn{2}{c}{\bf Maximum ($\mu$s)} \\
& \bf Insert & \bf Delete
& \bf Insert & \bf Delete \\
\hline
Movielens    & 8.667  & 7.665  & 598.641   & 97.574    \\
Gowalla      & 3.144  & 2.178  & 1229.203  & 37.060    \\
AmazonBook   & 4.418  & 3.330  & 1192.984  & 59.092    \\
SteamGame    & 1.319  & 0.402  & 1206.379  & 15.560    \\
MIND         & 2.443  & 1.398  & 1227.299  & 24.256    \\
Twitch       & 3.074  & 1.090  & 1272.724  & 8.727     \\
Yambda       & 81.360 & 70.561 & 21233.573 & 16383.458 \\
MAG          & 7.694  & 5.630  & 1236.364  & 31.089    \\
\hline
\end{tabular}
\end{small}
\arrayrulecolor{black}
\end{table}

\update{
\section{Correctness-Guaranteed Extension of \kalgo{}}\label{sec:KASC-cor}
In addition to the heuristic \kalgo{}, which avoids excessive samples through borderline items identification, we further investigate a
correctness-guaranteed variant, \kalgoc{}.
Specifically, for each observed item $v_t$, \kalgoc{} maintains its estimated Swing score $\widetilde{\textsf{sw}}(v_q,v_t)$ and empirical variance $\sigma_t$ in an online manner.
Based on the empirical Bernstein inequality, after $n$ samples, we
construct the lower and upper confidence bounds
$\widetilde{\textsf{sw}}(v_q,v_t)-\Phi(n,\sigma_t,\zeta,\delta)$ and
$\widetilde{\textsf{sw}}(v_q,v_t)+\Phi(n,\sigma_t,\zeta,\delta)$,
respectively.
For the items that have not been observed in the samples, $\Phi(n,0,\zeta,\delta)$ serves as a uniform upper bound on their Swing scores.
Once the minimum lower bound among the current top-$K$ items exceeds the maximum upper bound among all the remaining items, the current top-$K$ result is certified to be exact with the prescribed confidence.
If such separation cannot be established efficiently, \kalgoc{} eventually resorts to exact computation, as correctness certification can require substantially tighter confidence bounds and hence considerably more samples than the range-based identification in \kalgo{}, particularly when the $K$-th and $(K+1)$-th Swing scores are close.
}

\stitle{\update{Correctness Analysis}}
\update{
The correctness of \kalgoc{} follows from the maintained confidence intervals.
Suppose that, with probability at least $1-\delta$, the true Swing score of every observed item $v_t$ lies within
$\left[
\widetilde{\textsf{sw}}(v_q,v_t)-f(n,\sigma_t,\zeta,\delta),
\widetilde{\textsf{sw}}(v_q,v_t)+f(n,\sigma_t,\zeta,\delta)
\right]$,
while the Swing score of every unseen item is upper-bounded by $f(n,0,\zeta,\delta)$.
Accordingly, $\underline{\beta}_K$ lower-bounds the Swing scores of all items in $\mathcal{H}$, whereas $\overline{\beta}_{K+1}$ upper-bounds those of all items outside $\mathcal{H}$.
Once the separation condition
$\underline{\beta}_K>\overline{\beta}_{K+1}$
holds, for any $v_i\in\mathcal{H}$ and $v_j\notin\mathcal{H}$, we have
$\textsf{sw}(v_q,v_i)
\ge \underline{\beta}_K
> \overline{\beta}_{K+1}
\ge \textsf{sw}(v_q,v_j)$.
Therefore, no item outside $\mathcal{H}$ can have a larger Swing score than any item in $\mathcal{H}$, and hence $\mathcal{H}$ constitutes the exact top-$K$ result with probability at least $1-\delta$.
}

\begin{algorithm}[!t]
\caption{\kalgoc{}}\label{alg:topk-sample}
\KwIn{$\G=(\U,\I,\EDG)$, query item $v_q$.}
\KwOut{$\widetilde{\textsf{sw}}(v_q,v_t)\ \forall{v_t\in \I}$}

\For{$\epsilon\gets \frac{1}{2},\frac{1}{4},\frac{1}{8},\ldots$}{
Calculate $n$ according to Eq.~\eqref{eq:comp-n}\;
\lIf{$n>|\U(v_q)|\cdot(|\U(v_q)|-1)/2$}{
    \Return{$\widetilde{\textsf{\textnormal{sw}}}(v_q,v_t)$ by Algorithm~\ref{alg:basic}}
}
Initialize an empty map $S$ and set $\mathcal{A}$\;
\For{$r\gets 1$ to $n$}{
Sample $u_i, u_j\in \U(v_q)$ s.t. $u_i\neq u_j$ and $d(u_i),d(u_j)\ge 2$\;
Increase $S(u_i,u_j)$ by 1\;
}
Initialize $r(v_t),\widetilde{\textsf{sw}}(v_t)$,$\sigma(v_t)$\ $\forall{v_t\in \I}$\;
\For{$(u_i,u_j)\in S$}{
$\I_{i,j}\gets \I(u_i) \cap \I(u_j)$\;
\For{$v_t\in \I_{i,j}$}{
$\mathcal{A} \gets \mathcal{A} \cup\{v_t\}$ \;
$r(v_t) \gets r(v_t) + S(u_i,u_j)$\;
Update $\widetilde{\textsf{sw}}(v_q,v_t)$ according to Eq.~\eqref{eq:update-mu}\;
Update $\sigma(v_t)$ according to Eq.~\eqref{eq:update-sigma}\;
}}
Initialize an empty heap $\mathcal{H}$ of size $K$\;
\For{$v_t \in \I$ s.t. $r(v_t)\neq 0$}{
$\widetilde{\textsf{sw}}(v_q,v_t) \gets \frac{r(v_t)}{n}\cdot \widetilde{\textsf{sw}}(v_q,v_t)$\;
$\sigma(v_t)\gets \sigma(v_t)+\frac{(n-r(v_t))(\frac{r(v_t)}{n}\cdot\widetilde{\textsf{sw}}(v_q,v_t)^2-\sigma(v_t))}{n-1}$\;
$v^\ast\gets \mathcal{H}.\textsf{min}()$\;
\lIf{$\widetilde{\textsf{\textnormal{sw}}}(v_q,v_t)>\widetilde{\textsf{\textnormal{sw}}}(v_q,v^\ast)$}{
$\mathcal{H}.\textsf{push}(v_t)$
}
}

\Comment{Separation Condition}
$\underline{\beta}_k\gets\underset{v_\ell\in\mathcal{H}}{\min}\ \widetilde{\textsf{\textnormal{sw}}}(v_q,v_\ell)-f(n,\sigma_\ell,\zeta,\delta)$\;
{
\footnotesize
$\overline{\beta}_{k+1}\gets{\max}\{{\underset{v_\ell\in\mathcal{A}\setminus\mathcal{H}}{\max}}\ \widetilde{\textsf{sw}}(v_q,v_\ell)+f(n,\sigma_\ell,\zeta,\delta),f(n,0,\zeta,\delta)\}$\;
}
\lIf{$\underline{\beta}_k>\overline{\beta}_{k+1}$}{\Return{$\widetilde{\textsf{\textnormal{sw}}}(v_q,v_t)\ \forall v_t\in\I$}}
}
\end{algorithm}

\section{Additional Experiments}\label{sec:add-exp}

\subsection{Intersection Cardinality Computation}
Figure~\ref{fig:speedup} reports the practical speedup of computing the intersection size directly, compared with explicitly materializing the intersection set. The results show that the size-only computation is consistently faster in nearly all evaluated cases. In particular, the speedup reaches up to $3.68\times$, while most cases fall between $1.3\times$ and $1.6\times$. These gains indicate that a considerable fraction of the cost in full set intersection comes not from identifying common elements itself, but from materializing and storing the resulting set. This improvement is more pronounced on MovieLens due to denser user overlap, which amplifies the cost of materialization. Therefore, when only the intersection cardinality is needed, bypassing result construction provides a clear efficiency advantage.

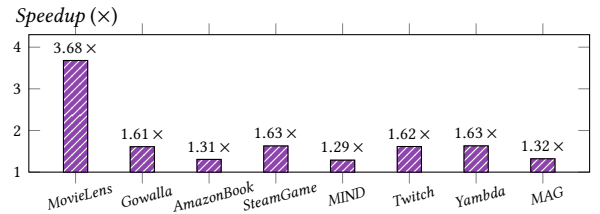
\begin{figure}[h]
\centering
\begin{small}
\begin{tikzpicture}[scale=1]
\begin{axis}[
    height=3.4cm,
    width=9cm,
    ybar=0.6pt,
    nodes near coords={\pgfmathprintnumber{\pgfplotspointmeta}\,$\times$},
    bar width=0.32cm,
    enlarge x limits=0.1,
    ylabel={\em Speedup} ($\times$),
    ymin=1,
    ymax=4.3,
    ytick={1,2,3,4},
    xtick={1,2,3,4,5,6,7,8},
    xticklabels={{\em MovieLens},{\em Gowalla},{\em AmazonBook},{\em SteamGame},{\em MIND},{\em Twitch},{\em Yambda},{\em MAG}},
    xticklabel style={rotate=15, anchor= north, font=\scriptsize},
    yticklabel style = {font=\scriptsize},
    every node near coord/.append style={font=\scriptsize},
    every axis y label/.style={at={(current axis.north west)},right=5mm,above=0mm},
    legend style={draw=none, at={(1.02,1.02)},anchor=north west,cells={anchor=west},font=\scriptsize},
    legend image code/.code={ \draw [#1] (0cm,-0.1cm) rectangle (0.3cm,0.15cm);},
    ]

\addplot [pattern color=white, preaction={fill, myblue}, pattern={north east lines}] coordinates {(1,3.68)
(2,1.61)
(3,1.305)
(4,1.628)
(5,1.287)
(6,1.615)
(7,1.63)
(8,1.32)
}; 

\end{axis}
\end{tikzpicture}
\hspace{1mm}
\vspace{-3mm}
\end{small}
\caption{Practical speedup of intersection-size computation over full set intersection.} \label{fig:speedup}
\end{figure}

\begin{table}[!t]
\renewcommand{\arraystretch}{0.9}
\centering
\updatefig{}
\caption{\update{Preprocessing costs and space overhead.}}
\label{tab:preprocess}
\vspace{-2ex}
\resizebox{\columnwidth}{!}{%
\begin{small}
\begin{tabular}{c|cc|cc|c}
\hline
\multirow{2}{*}{\bf Dataset} 
& \multicolumn{2}{c|}{\bf Time (s)} 
& \multicolumn{2}{c|}{\bf Space (GB)} 
& \multirow{2}{*}{\bf Graph Size} \\
& \bf Hashmap & \bf BSR & \bf Hashmap & \bf BSR \\
\hline
Movielens  & 0.077 & 0.414 & 0.004 & 0.003 & 0.011 GB\\
Gowalla    & 0.229 & 0.704 & 0.013 & 0.005 & 0.015 GB\\
AmazonBook & 0.699 & 4.189 & 0.043 & 0.014 & 0.039 GB\\
SteamGame  & 3.131 & 1.632 & 0.398 & 0.066 & 0.190 GB\\
MIND       & 6.056 & 4.364 & 0.404 & 0.111 & 0.228 GB\\
Twitch     & 60.050 & 44.065 & 4.195 & 0.999 & 2.927 GB\\
Yambda     & 119.759 & 1284.860 & 9.295 & 2.862 & 13.002 GB\\
MAG        & 225.489 & 216.014 & 15.899 & 4.692 & 11.054 GB\\
\hline
\end{tabular}
\end{small}
}
\arrayrulecolor{black}
\end{table}

\begin{table}[!t]
\renewcommand{\arraystretch}{0.9}
\centering
\updatefig{}
\caption{\update{Extra memory overhead during query processing.}}
\label{tab:query-extra-memory}
\vspace{-2ex}
\begin{small}
\begin{tabular}{c|cc|cc}
\hline
\multirow{2}{*}{\bf Dataset}
& \multicolumn{2}{c|}{\bf Average Extra RSS (KB)}
& \multicolumn{2}{c}{\bf Peak Extra RSS (MB)} \\
& \bf \algo{} & \bf \kalgo{}
& \bf \algo{} & \bf \kalgo{} \\
\hline
Movielens  & 77.310  & \textbf{15.041} & 74.359  & \textbf{19.509} \\
Gowalla    & 0.002   & \textbf{0.001}           & 0.002   & \textbf{0.001} \\
AmazonBook & \textbf{0.032} & 0.077    & \textbf{0.031} & 0.076 \\
SteamGame  & 118.773 & \textbf{24.531} & 125.875 & \textbf{23.586} \\
MIND       & 121.407 & \textbf{10.735} & 93.227  & \textbf{11.478} \\
Twitch     & 103.052 & \textbf{3.421}  & 110.703 & \textbf{6.973} \\
Yambda     & \textbf{92.126} & 3630.875 & \textbf{74.824} & 136.496 \\
MAG        & \textbf{0.065} & 17.770   & \textbf{0.125} & 21.352 \\
\hline
\end{tabular}
\end{small}
\arrayrulecolor{black}
\end{table}

\begin{table}[!t]
\renewcommand{\arraystretch}{0.9}
\centering
\updatefig{}
\caption{\update{Complete memory usage during query processing.}}
\label{tab:query-memory-full}
\vspace{-2ex}
\begin{small}
\begin{tabular}{c|cc|cc}
\hline
\multirow{2}{*}{\bf Dataset}
& \multicolumn{2}{c|}{\bf Average RSS (GB)}
& \multicolumn{2}{c}{\bf Peak RSS (GB)} \\
& \bf \algo{} & \bf \kalgo{}
& \bf \algo{} & \bf \kalgo{} \\
\hline
Movielens  & 0.102 & \textbf{0.054} & 0.106 & \textbf{0.056} \\
Gowalla    & 0.053 & \textbf{0.051} & 0.054 & \textbf{0.052} \\
AmazonBook & 0.139 & \textbf{0.136} & 0.139 & \textbf{0.136} \\
SteamGame  & 0.961 & \textbf{0.868} & 0.979 & \textbf{0.878} \\
MIND       & 1.031 & \textbf{0.935} & 1.047 & \textbf{0.939} \\
Twitch     & 10.705 & \textbf{10.618} & 10.719 & \textbf{10.624} \\
MAG        & 40.602 & \textbf{40.455} & 40.602 & \textbf{40.456} \\
Yambda     & \textbf{35.784} & 35.831 & \textbf{35.854} & 35.912 \\
\hline
\end{tabular}
\end{small}
\arrayrulecolor{black}
\end{table}

\subsection{BSR Construction Cost}
Table~\ref{tab:BSRcost} presents the BSR construction time (denoted as Prep.) and the average set intersection time under different ordering schemes in the \texttt{GNS} algorithm with the same parameters. The results show that our LightGRO-based BSR construction is consistently faster in both preprocessing and set intersection costs. Specifically, on the large dataset {\em Yambda}, it achieves up to an $6.7\times$ speedup over GRO in preprocessing time, while maintaining comparable or even better intersection performance. In contrast, simple ordering strategies, such as the original index or degree-ascending order, tend to perform poorly on large datasets.

\update{
\subsection{Preprocessing Cost and Space Overhead}\label{sec:space}
}

\stitle{\update{Preprocess Cost}}
\update{
Table~\ref{tab:preprocess} reports the preprocessing cost and memory overhead for constructing the hashmap used for constant-time membership checking, and the BSR representation in \texttt{QFilter++}.
While BSR incurs a higher preprocessing cost than the lightweight hashmap, our optimized implementation significantly reduces its construction time, achieving up to $6\times$ speedup over the original design~\cite{han2018speeding}, while remaining practical at scale. This one-time cost is quickly amortized under frequent queries, where BSR delivers substantial efficiency gains, while the hashmap supports efficient dynamic updates without full reconstruction.
}

\stitle{\update{Query Memory Usage}}
\update{
Table~\ref{tab:query-extra-memory} reports the additional peak and average memory usage during query processing.
On most datasets, \kalgo{} requires less extra memory than \algo{}, as candidate generation and refinement reduce the number of sampled user pairs that need to be maintained.
However, this reduction is dataset-dependent, the auxiliary storage for variance estimates, hit counts, and a top-$K$ priority queue offsets or outweighs the savings from reduced sampling.
Overall, the additional memory overhead during query processing remains modest.
}

\stitle{\update{Complete Memory Usage}}
\update{
Table~\ref{tab:query-memory-full} reports the total memory footprint during query processing, accounting for the graph, hash map, BSR, and dynamic query memory. The two methods show similar memory usage overall, since the total footprint is dominated by the graph and preprocessing structures. For implementation convenience, we retain both the original and reordered graphs, together with the hashmap and BSR, throughout query processing. As a result, the reduction in dynamic query memory achieved by \kalgo{} only leads to a modest reduction in the total memory footprint on most datasets, while {\em Yambda} shows a slightly higher usage. Overall, \kalgo{} introduces no substantial increase in complete memory usage compared with \algo{}.
}

\subsection{Parameter Settings}\label{sec:para-setting}
In this section, we introduce the parameters that are not discussed in the main text. Some parameters are fixed for each dataset because they have little impact on the experimental results. Specifically, the parameter $\rho$ in Eq.~\eqref{eq:select} is fixed to 0.5 for \algo{} and 0.8 for \kalgo{} across all datasets, and the number of extra candidates $\kappa$ is fixed to 5 for all datasets. We then perform an exhaustive grid search over the parameter spaces of \algo{} and \kalgo{} to obtain the optimal configurations.
Other parameter settings with the best performance are summarized below.

For \algo{}, the optimal $\epsilon$ is 0.025 on {\em Movielens}, 0.1 on {\em Gowalla}, 0.05 on {\em AmazonBook}, 0.02 on {\em SteamGame}, 0.02 on {\em MIND}, 0.02 on {\em Twitch}, and 0.02 on {\em Yambda}. On {\em MAG}, the optimal $\epsilon$ is 0.05 for popular queries and 0.1 for random queries.

\update{For \kalgo{}, the parameter configuration requires slightly finer tuning than that of \algo{}, since the top-$K$ extension introduces additional control over candidate identification and refinement. In particular, the optimal $(\epsilon,\beta)$ may vary with the dataset, query type, and $K$, reflecting different query-time and precision trade-offs under different settings. The resulting optimal configurations are summarized in Table~\ref{tab:kasc_parameters}.}

\begin{table}[!t]
\renewcommand{\arraystretch}{0.9}
\centering
\caption{\update{Optimal parameter settings of \kalgo{}.}}
\label{tab:kasc_parameters}
\vspace{-2ex}
\updatefig{}
\begin{small}
\begin{tabular}{c|c|ccc}
\hline
\multirow{2}{*}{\bf Dataset}
& \multirow{2}{*}{\bf Query}
& \multicolumn{3}{c}{\bf $(\epsilon,\beta)$} \\
& & $\mathbf{K=20}$ & $\mathbf{K=50}$ & $\mathbf{K=100}$ \\
\hline
\multirow{2}{*}{Movielens}
& Popular & (0.06, 0.1) & (0.06, 0.1) & (0.06, 0.1) \\
& Random  & (0.05, 0.1) & (0.04, 0.1) & (0.04, 0.1) \\
\hline
Gowalla
& Both & (0.1, 0.2) & (0.1, 0.2) & (0.1, 0.2) \\
\hline
AmazonBook
& Both & (0.05, 0.2) & (0.05, 0.2) & (0.05, 0.2) \\
\hline
\multirow{2}{*}{SteamGame}
& Popular & (0.03, 0.2) & (0.025, 0.2) & (0.02, 0.2) \\
& Random  & (0.05, 0.2) & (0.03, 0.2)  & (0.025, 0.2) \\
\hline
\multirow{2}{*}{MIND}
& Popular & (0.05, 0.2) & (0.03, 0.2) & (0.025, 0.3) \\
& Random  & (0.05, 0.1) & (0.05, 0.2) & (0.04, 0.2) \\
\hline
\multirow{2}{*}{Twitch}
& Popular & (0.05, 0.2) & (0.03, 0.2) & (0.03, 0.2) \\
& Random  & (0.05, 0.1) & (0.05, 0.1) & (0.05, 0.2) \\
\hline
\multirow{2}{*}{MAG}
& Popular & (0.1, 0.05) & (0.1, 0.1) & (0.08, 0.2) \\
& Random  & (0.3, 0.05) & (0.2, 0.1) & (0.2, 0.3) \\
\hline
\multirow{2}{*}{Yambda}
& Popular & (0.05, 0.1) & (0.05, 0.2) & (0.06, 0.1) \\
& Random  & (0.1, 0.1)  & (0.1, 0.1)  & (0.05, 0.1) \\
\hline
\end{tabular}
\end{small}
\arrayrulecolor{black}
\end{table}

\subsection{\update{Complete Parameter Analysis}}\label{sec:complete-parameter-analysis}
\begin{figure*}[!t]
\centering
\begin{small}
\newcommand{\ASCRhoTimeSubplot}[3]{%
\subfloat[{\em #1}]{%
\begin{tikzpicture}[scale=1,every mark/.append style={mark size=3pt}]
\begin{axis}[
    height=2.7cm,
    width=3.0cm,
    ylabel={\it query time(ms)},
    xmin=1, xmax=7,
    xtick={1,...,7},
    xticklabel style={font=\scriptsize,rotate=40,anchor=north east},
    yticklabel style={font=\footnotesize},
    xticklabels={0.2,0.3,0.4,0.5,0.6,0.7,0.8},
    scaled y ticks=false,
    every axis y label/.style={font=\footnotesize,at={(current axis.north west)},right=5mm,above=0mm},
    #3
]
\addplot[line width=0.7mm,smooth,color=NSCcol1] plot coordinates {#2};
\end{axis}
\end{tikzpicture}\hspace{0mm}%
}%
}
\begin{tikzpicture}
\begin{customlegend}[
    legend columns=3,
    legend entries={$K=20$,$K=50$,$K=100$},
    legend style={at={(0.45,1.35)},anchor=north,draw=none,font=\footnotesize,column sep=0.2cm}
]
\addlegendimage{line width=0.7mm,mark size=4pt,color=NSCcol1}
\addlegendimage{line width=0.7mm,mark size=4pt,color=my_violet}
\addlegendimage{line width=0.7mm,mark size=4pt,color=cyan}
\end{customlegend}
\end{tikzpicture}
\\[-\lineskip]
\vspace{-4mm}

\ASCRhoTimeSubplot{MovieLens}
{(1,102.214) (2,102.332) (3,102.885) (4,109.062) (5,107.256) (6,112.479) (7,115.447)}
{ymin=100,ymax=120,ytick={100,110,120}}
\ASCRhoTimeSubplot{Gowalla}
{(1,0.259) (2,0.259) (3,0.257) (4,0.259) (5,0.259) (6,0.258) (7,0.256)}
{ymin=0.25,ymax=0.27,ytick={0.25,0.26,0.27}}
\ASCRhoTimeSubplot{AmazonBook}
{(1,0.808) (2,0.820) (3,0.815) (4,0.823) (5,0.818) (6,0.818) (7,0.819)}
{ymin=0.80,ymax=0.83,ytick={0.80,0.81,0.82,0.83},yticklabels={0.80,0.81,0.82,0.83}}
\ASCRhoTimeSubplot{SteamGame}
{(1,56.431) (2,63.992) (3,66.320) (4,77.890) (5,99.530) (6,90.075) (7,107.742)}
{ymin=50,ymax=110,ytick={50,70,90,110}}
\ASCRhoTimeSubplot{MIND}
{(1,46.833) (2,45.937) (3,45.938) (4,49.395) (5,51.921) (6,53.689) (7,58.986)}
{ymin=45,ymax=60,ytick={45,50,55,60}}
\ASCRhoTimeSubplot{Twitch}
{(1,41.877) (2,39.044) (3,42.000) (4,43.302) (5,46.039) (6,49.330) (7,49.421)}
{ymin=38,ymax=51,ytick={38,42,46,50}}
\ASCRhoTimeSubplot{Yambda}
{(1,384.795) (2,320.918) (3,301.770) (4,294.355) (5,286.082) (6,281.869) (7,288.409)}
{ymin=280,ymax=400,ytick={280,320,360,400}}
\ASCRhoTimeSubplot{MAG}
{(1,4.228) (2,4.358) (3,4.077) (4,3.831) (5,3.869) (6,4.348) (7,3.399)}
{ymin=3.2,ymax=4.5,ytick={3.2,3.6,4.0,4.4},yticklabels={3.2,3.6,4.0,4.4}}

\end{small}
\vspace{-3mm}
\caption{Query time of \algo{} when varying $\rho$.}
\label{fig:asc-vary-time-rho}
\end{figure*}
\begin{figure*}[!t]
\centering
\begin{small}
\newcommand{\ASCRhoPrecisionSubplot}[4]{%
\subfloat[{\em #1}]{%
\begin{tikzpicture}[scale=1,every mark/.append style={mark size=3pt}]
\begin{axis}[
    height=2.7cm,
    width=3.0cm,
    ylabel={\it precision},
    xmin=1, xmax=7,
    ymin=0.999, ymax=1,
    xtick={1,2,3,4,5,6,7},
    ytick={0.999,1},
    xticklabel style={font=\scriptsize,rotate=40,anchor=north east},
    yticklabel style={font=\scriptsize},
    xticklabels={0.2,0.3,0.4,0.5,0.6,0.7,0.8},
    yticklabels={99.9,100},
    every axis y label/.style={font=\footnotesize,at={(current axis.north west)},right=3mm,above=0mm},
]
\addplot[line width=0.7mm,smooth,color=NSCcol1] plot coordinates {#2};
\addplot[line width=0.7mm,smooth,color=my_violet] plot coordinates {#3};
\addplot[line width=0.7mm,smooth,color=cyan] plot coordinates {#4};
\end{axis}
\end{tikzpicture}\hspace{0mm}%
}%
}

\vspace{-4mm}

\ASCRhoPrecisionSubplot{MovieLens}
{(1,1.000000) (2,1.000000) (3,1.000000) (4,1.000000) (5,1.000000) (6,1.000000) (7,1.000000)}
{(1,0.999979) (2,0.999979) (3,0.999979) (4,0.999979) (5,0.999979) (6,0.999979) (7,0.999979)}
{(1,0.999949) (2,0.999959) (3,0.999959) (4,0.999979) (5,1.000000) (6,1.000000) (7,0.999979)}
\ASCRhoPrecisionSubplot{Gowalla}
{(1,1.000000) (2,1.000000) (3,1.000000) (4,1.000000) (5,1.000000) (6,1.000000) (7,1.000000)}
{(1,1.000000) (2,1.000000) (3,1.000000) (4,1.000000) (5,1.000000) (6,1.000000) (7,1.000000)}
{(1,1.000000) (2,1.000000) (3,1.000000) (4,1.000000) (5,1.000000) (6,1.000000) (7,1.000000)}
\ASCRhoPrecisionSubplot{AmazonBook}
{(1,1.000000) (2,1.000000) (3,1.000000) (4,1.000000) (5,1.000000) (6,1.000000) (7,1.000000)}
{(1,1.000000) (2,1.000000) (3,1.000000) (4,1.000000) (5,1.000000) (6,1.000000) (7,1.000000)}
{(1,1.000000) (2,1.000000) (3,1.000000) (4,1.000000) (5,1.000000) (6,1.000000) (7,1.000000)}
\ASCRhoPrecisionSubplot{SteamGame}
{(1,0.999942) (2,1.000000) (3,0.999885) (4,0.999885) (5,0.999942) (6,0.999942) (7,0.999712)}
{(1,0.999746) (2,0.999723) (3,0.999769) (4,0.999723) (5,0.999700) (6,0.999700) (7,0.999562)}
{(1,0.999758) (2,0.999723) (3,0.999781) (4,0.999712) (5,0.999700) (6,0.999689) (7,0.999619)}
\ASCRhoPrecisionSubplot{MIND}
{(1,0.999930) (2,0.999930) (3,0.999930) (4,0.999930) (5,0.999860) (6,0.999930) (7,0.999930)}
{(1,0.999832) (2,0.999832) (3,0.999832) (4,0.999832) (5,0.999916) (6,0.999860) (7,0.999860)}
{(1,0.999846) (2,0.999832) (3,0.999832) (4,0.999846) (5,0.999804) (6,0.999846) (7,0.999763)}
\ASCRhoPrecisionSubplot{Twitch}
{(1,1.000000) (2,0.999947) (3,0.999895) (4,1.000000) (5,1.000000) (6,1.000000) (7,1.000000)}
{(1,0.999979) (2,0.999937) (3,0.999958) (4,0.999895) (5,0.999916) (6,0.999916) (7,0.999916)}
{(1,0.999884) (2,0.999863) (3,0.999895) (4,0.999884) (5,0.999895) (6,0.999895) (7,0.999895)}
\ASCRhoPrecisionSubplot{Yambda}
{(1,0.999903) (2,0.999903) (3,0.999903) (4,1.000000) (5,1.000000) (6,1.000000) (7,1.000000)}
{(1,0.999884) (2,0.999961) (3,0.999923) (4,0.999923) (5,0.999923) (6,0.999923) (7,0.999923)}
{(1,0.999923) (2,0.999903) (3,0.999884) (4,0.999903) (5,0.999923) (6,0.999923) (7,0.999923)}
\ASCRhoPrecisionSubplot{MAG}
{(1,0.999898) (2,0.999949) (3,0.999949) (4,0.999949) (5,0.999949) (6,0.999949) (7,0.999898)}
{(1,0.999837) (2,0.999817) (3,0.999878) (4,0.999837) (5,0.999837) (6,0.999837) (7,0.999858)}
{(1,0.999736) (2,0.999746) (3,0.999817) (4,0.999837) (5,0.999837) (6,0.999837) (7,0.999807)}

\end{small}
\vspace{-3mm}
\caption{Precision of \algo{} when varying $\rho$.}
\label{fig:asc-vary-precision-rho}
\end{figure*}

\begin{figure*}[!t]
\centering
\begin{small}
\newcommand{\VaryRhoSubplot}[5]{%
\subfloat[{\em #1}]{%
\begin{tikzpicture}[scale=1,every mark/.append style={mark size=3pt}]
\begin{axis}[
    height=2.7cm,
    width=3.0cm,
    ylabel={\it query time(ms)},
    xmin=1, xmax=7,
    xtick={1,...,7},
    xticklabel style={font=\scriptsize,rotate=40,anchor=north east},
    yticklabel style={font=\footnotesize},
    xticklabels={0.3,0.4,0.5,0.6,0.7,0.8,0.9},
    scaled y ticks=false,
    every axis y label/.style={font=\footnotesize,at={(current axis.north west)},right=5mm,above=0mm},
    #5
]
\addplot[line width=0.7mm,smooth,color=NSCcol1] plot coordinates {#2};
\addplot[line width=0.7mm,smooth,color=my_violet] plot coordinates {#3};
\addplot[line width=0.7mm,smooth,color=cyan] plot coordinates {#4};
\end{axis}
\end{tikzpicture}\hspace{0mm}%
}%
}

\begin{tikzpicture}
\begin{customlegend}[
    legend columns=3,
    legend entries={$K=20$,$K=50$,$K=100$},
    legend style={at={(0.45,1.35)},anchor=north,draw=none,font=\footnotesize,column sep=0.2cm}
]
\addlegendimage{line width=0.7mm,mark size=4pt,color=NSCcol1}
\addlegendimage{line width=0.7mm,mark size=4pt,color=my_violet}
\addlegendimage{line width=0.7mm,mark size=4pt,color=cyan}
\end{customlegend}
\end{tikzpicture}
\\[-\lineskip]
\vspace{-4mm}

\VaryRhoSubplot{MovieLens}
{(1,67.316) (2,63.796) (3,62.638) (4,60.401) (5,55.366) (6,53.262) (7,55.545)}
{(1,82.595) (2,80.859) (3,80.442) (4,75.778) (5,72.258) (6,70.315) (7,69.161)}
{(1,81.102) (2,79.464) (3,78.397) (4,74.695) (5,70.718) (6,69.016) (7,70.531)}
{ymin=50,ymax=85,ytick={50,60,70,80}}
\VaryRhoSubplot{Gowalla}
{(1,0.257) (2,0.256) (3,0.253) (4,0.259) (5,0.255) (6,0.257) (7,0.258)}
{(1,0.266) (2,0.265) (3,0.265) (4,0.266) (5,0.268) (6,0.271) (7,0.269)}
{(1,0.256) (2,0.260) (3,0.259) (4,0.257) (5,0.257) (6,0.260) (7,0.263)}
{ymin=0.25,ymax=0.275,ytick={0.25,0.26,0.27}}
\VaryRhoSubplot{AmazonBook}
{(1,0.814) (2,0.816) (3,0.820) (4,0.814) (5,0.809) (6,0.819) (7,0.812)}
{(1,0.837) (2,0.837) (3,0.829) (4,0.836) (5,0.851) (6,0.838) (7,0.835)}
{(1,0.811) (2,0.809) (3,0.815) (4,0.812) (5,0.822) (6,0.827) (7,0.808)}
{ymin=0.80,ymax=0.86,ytick={0.80,0.82,0.84,0.86}, yticklabels={0.80,0.82,0.84,0.86}}
\VaryRhoSubplot{SteamGame}
{(1,24.627) (2,19.8) (3,18.127) (4,16.768) (5,16.285) (6,15.884) (7,7.283)}
{(1,30.305) (2,24.604) (3,19.156) (4,18.380) (5,18.289) (6,17.401) (7,14.421)}
{(1,35.398) (2,31.793) (3,26.933) (4,23.897) (5,22.845) (6,20.169) (7,15.858)}
{ymin=5,ymax=38,ytick={5,15,25,35}}
\VaryRhoSubplot{MIND}
{(1,9.525) (2,7.945) (3,6.687) (4,6.244) (5,5.937) (6,5.463) (7,5.355)}
{(1,9.614) (2,8.008) (3,6.767) (4,6.345) (5,6.091) (6,5.603) (7,5.476)}
{(1,15.792) (2,13.416) (3,8.448) (4,8.549) (5,8.001) (6,7.758) (7,7.479)}
{ymin=5,ymax=17,ytick={5,10,15}}
\VaryRhoSubplot{Twitch}
{(1,5.063) (2,4.503) (3,4.162) (4,3.896) (5,3.901) (6,3.782) (7,4.401)}
{(1,4.964) (2,4.391) (3,4.040) (4,3.566) (5,3.532) (6,3.410) (7,3.844)}
{(1,5.239) (2,4.586) (3,4.250) (4,3.744) (5,3.715) (6,3.626) (7,3.962)}
{ymin=3,ymax=5.5,ytick={3,4,5}}
\VaryRhoSubplot{Yambda}
{(1,31.539) (2,26.900) (3,29.384) (4,24.881) (5,25.978) (6,24.767) (7,26.572)}
{(1,29.045) (2,29.067) (3,26.776) (4,26.657) (5,26.552) (6,27.011) (7,27.181)}
{(1,70.960) (2,68.026) (3,68.248) (4,60.521) (5,60.538) (6,60.761) (7,60.707)}
{ymin=20,ymax=75,ytick={20,40,60}}
\VaryRhoSubplot{MAG}
{(1,1.452) (2,1.363) (3,1.398) (4,1.384) (5,1.594) (6,1.401) (7,1.798)}
{(1,1.5) (2,1.295) (3,1.304) (4,1.245) (5,1.299) (6,1.24) (7,2.137)}
{(1,1.597) (2,1.617) (3,1.708) (4,1.445) (5,1.387) (6,1.89) (7,2.582)}
{ymin=1.2,ymax=2.7,ytick={1.2,1.7,2.2,2.7}}

\end{small}
\vspace{-3mm}
\caption{Query time of \kalgo{} when varying $\rho$.}
\end{figure*}
\begin{figure*}[!t]
\centering
\begin{small}

\newcommand{\VaryRhoPrecisionSubplot}[4]{%
\subfloat[{\em #1}]{%
\begin{tikzpicture}[scale=1,every mark/.append style={mark size=3pt}]
\begin{axis}[
    height=2.7cm,
    width=3.0cm,
    ylabel={\it precision},
    xmin=1, xmax=7,
    ymin=0.999, ymax=1,
    xtick={1,2,3,4,5,6,7},
    ytick={0.999,1},
    xticklabel style={font=\scriptsize,rotate=40,anchor=north east},
    yticklabel style={font=\scriptsize},
    xticklabels={0.3,0.4,0.5,0.6,0.7,0.8,0.9},
    yticklabels={99.9,100},
    every axis y label/.style={
        font=\footnotesize,
        at={(current axis.north west)},
        right=3mm,above=0mm
    },
]
\addplot[line width=0.7mm,smooth,color=NSCcol1]
    plot coordinates {#2};
\addplot[line width=0.7mm,smooth,color=my_violet]
    plot coordinates {#3};
\addplot[line width=0.7mm,smooth,color=cyan]
    plot coordinates {#4};
\end{axis}
\end{tikzpicture}\hspace{0mm}%
}%
}

\vspace{-5mm}

\VaryRhoPrecisionSubplot{MovieLens}
{(1,0.999897) (2,0.999846) (3,0.999794) (4,0.999743) (5,0.999743) (6,0.999794) (7,0.999691)}
{(1,0.999918) (2,0.999897) (3,0.999938) (4,0.999897) (5,0.999938) (6,0.999897) (7,0.999918)}
{(1,0.999887) (2,0.999877) (3,0.999907) (4,0.999866) (5,0.999877) (6,0.999835) (7,0.999784)}
\VaryRhoPrecisionSubplot{Gowalla}
{(1,1.000000) (2,1.000000) (3,1.000000) (4,1.000000) (5,1.000000) (6,1.000000) (7,1.000000)}
{(1,1.000000) (2,1.000000) (3,1.000000) (4,1.000000) (5,1.000000) (6,1.000000) (7,1.000000)}
{(1,1.000000) (2,1.000000) (3,1.000000) (4,1.000000) (5,1.000000) (6,1.000000) (7,1.000000)}
\VaryRhoPrecisionSubplot{AmazonBook}
{(1,1.000000) (2,1.000000) (3,1.000000) (4,1.000000) (5,1.000000) (6,1.000000) (7,1.000000)}
{(1,1.000000) (2,1.000000) (3,1.000000) (4,1.000000) (5,1.000000) (6,1.000000) (7,1.000000)}
{(1,1.000000) (2,1.000000) (3,1.000000) (4,1.000000) (5,1.000000) (6,1.000000) (7,1.000000)}
\VaryRhoPrecisionSubplot{SteamGame}
{(1,0.999942) (2,1) (3,1) (4,0.999942) (5,1) (6,1) (7,0.999539)}
{(1,0.999769) (2,0.999931) (3,0.999862) (4,0.999954) (5,0.999885) (6,0.999792) (7,0.999493)}
{(1,0.999815) (2,0.999815) (3,0.999839) (4,0.999839) (5,0.999815) (6,0.999792) (7,0.999343)}
\VaryRhoPrecisionSubplot{MIND}
{(1,0.999651) (2,0.999651) (3,0.999791) (4,0.999721) (5,0.999441) (6,0.999791) (7,0.999721)}
{(1,0.999385) (2,0.999525) (3,0.999469) (4,0.999358) (5,0.999358) (6,0.999162) (7,0.999274)}
{(1,0.999372) (2,0.999344) (3,0.999330) (4,0.999330) (5,0.999232) (6,0.999190) (7,0.999246)}
\VaryRhoPrecisionSubplot{Twitch}
{(1,1) (2,1) (3,1) (4,0.999947) (5,0.999947) (6,0.999895) (7,1.000000)}
{(1,0.999747) (2,0.999621) (3,0.999579) (4,0.999579) (5,0.999579) (6,0.999642) (7,0.999811)}
{(1,0.999095) (2,0.999126) (3,0.999093) (4,0.999083) (5,0.999073) (6,0.999005) (7,0.999284)}
\VaryRhoPrecisionSubplot{Yambda}
{(1,0.999131) (2,0.999517) (3,0.999807) (4,0.999710) (5,0.999710) (6,0.999710) (7,0.999614)}
{(1,0.999058) (2,0.999151) (3,0.999344) (4,0.999498) (5,0.999228) (6,0.999228) (7,0.999305)}
{(1,0.999788) (2,0.999865) (3,0.999865) (4,0.999710) (5,0.999710) (6,0.999710) (7,0.999710)}
\VaryRhoPrecisionSubplot{MAG}
{(1,0.999695) (2,0.999797) (3,0.999746) (4,0.999543) (5,0.999644) (6,0.999593) (7,0.999746)}
{(1,0.999553) (2,0.999512) (3,0.999512) (4,0.999533) (5,0.999573) (6,0.999573) (7,0.999492)}
{(1,0.999207) (2,0.999329) (3,0.999329) (4,0.999258) (5,0.999228) (6,0.999228) (7,0.999167)}

\end{small}
\vspace{-3mm}
\caption{Precision of \kalgo{} when varying $\rho$.}
\end{figure*}

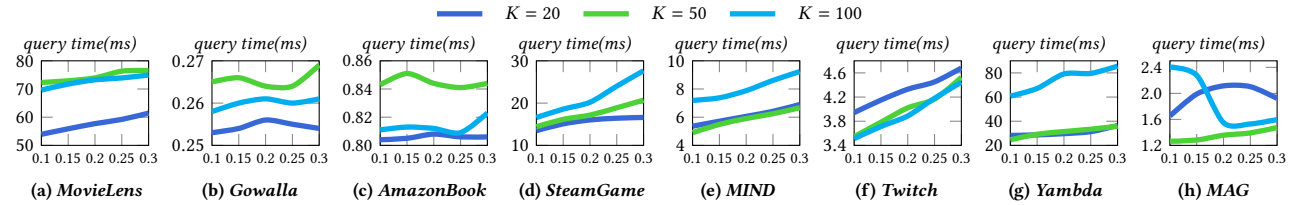
\begin{figure*}[!t]
\centering
\begin{small}
\newcommand{\VaryBetaSubplot}[5]{%
\subfloat[{\em #1}]{%
\begin{tikzpicture}[scale=1,every mark/.append style={mark size=3pt}]
\begin{axis}[
    height=2.7cm,
    width=3.0cm,
    ylabel={\it query time(ms)},
    xmin=1, xmax=5,
    xtick={1,...,5},
    xticklabel style={font=\scriptsize},
    yticklabel style={font=\footnotesize},
    xticklabels={0.1,0.15,0.2,0.25,0.3},
    scaled y ticks=false,
    every axis y label/.style={font=\footnotesize,at={(current axis.north west)},right=5mm,above=0mm},
    #5
]
\addplot[line width=0.7mm,smooth,color=NSCcol1] plot coordinates {#2};
\addplot[line width=0.7mm,smooth,color=my_violet] plot coordinates {#3};
\addplot[line width=0.7mm,smooth,color=cyan] plot coordinates {#4};
\end{axis}
\end{tikzpicture}\hspace{0mm}%
}%
}

\begin{tikzpicture}
\begin{customlegend}[
    legend columns=3,
    legend entries={$K=20$,$K=50$,$K=100$},
    legend style={at={(0.45,1.35)},anchor=north,draw=none,font=\footnotesize,column sep=0.2cm}
]
\addlegendimage{line width=0.7mm,mark size=4pt,color=NSCcol1}
\addlegendimage{line width=0.7mm,mark size=4pt,color=my_violet}
\addlegendimage{line width=0.7mm,mark size=4pt,color=cyan}
\end{customlegend}
\end{tikzpicture}
\\[-\lineskip]
\vspace{-4mm}

\VaryBetaSubplot{MovieLens}
{(1,53.925) (2,55.883) (3,57.710) (4,59.210) (5,61.438)}
{(1,72.321) (2,72.902) (3,73.945) (4,76.333) (5,76.571)}
{(1,69.632) (2,71.749) (3,73.261) (4,73.916) (5,74.962)}
{ymin=50,ymax=80,ytick={50,60,70,80}}
\VaryBetaSubplot{Gowalla}
{(1,0.253) (2,0.254) (3,0.256) (4,0.255) (5,0.254)}
{(1,0.265) (2,0.266) (3,0.264) (4,0.264) (5,0.269)}
{(1,0.258) (2,0.260) (3,0.261) (4,0.260) (5,0.261)}
{ymin=0.25,ymax=0.27,ytick={0.25,0.26,0.27}}
\VaryBetaSubplot{AmazonBook}
{(1,0.804) (2,0.805) (3,0.808) (4,0.806) (5,0.806)}
{(1,0.843) (2,0.851) (3,0.844) (4,0.841) (5,0.844)}
{(1,0.811) (2,0.813) (3,0.812) (4,0.809) (5,0.823)}
{ymin=0.80,ymax=0.86,ytick={0.80,0.82,0.84,0.86},yticklabels={0.80,0.82,0.84,0.86}}
\VaryBetaSubplot{SteamGame}
{(1,13.479) (2,15.089) (3,16.013) (4,16.429) (5,16.601)}
{(1,14.401) (2,16.157) (3,17.159) (4,18.889) (5,20.712)}
{(1,16.580) (2,18.579) (3,20.198) (4,24.009) (5,27.658)}
{ymin=10,ymax=30,ytick={10,20,30}}
\VaryBetaSubplot{MIND}
{(1,5.360) (2,5.715) (3,6.058) (4,6.415) (5,6.898)}
{(1,4.900) (2,5.490) (3,5.904) (4,6.245) (5,6.670)}
{(1,7.189) (2,7.373) (3,7.882) (4,8.611) (5,9.237)}
{ymin=4,ymax=10,ytick={4,6,8,10}}
\VaryBetaSubplot{Twitch}
{(1,3.942) (2,4.150) (3,4.330) (4,4.451) (5,4.681)}
{(1,3.544) (2,3.784) (3,4.022) (4,4.166) (5,4.526)}
{(1,3.511) (2,3.713) (3,3.887) (4,4.180) (5,4.446)}
{ymin=3.4,ymax=4.8,ytick={3.4,3.8,4.2,4.6}}
\VaryBetaSubplot{Yambda}
{(1,28.213) (2,28.685) (3,29.695) (4,31.348) (5,36.420)}
{(1,24.566) (2,29.226) (3,31.414) (4,33.405) (5,35.869)}
{(1,60.780) (2,67.137) (3,79.145) (4,79.576) (5,85.616)}
{ymin=20,ymax=90,ytick={20,40,60,80}}
\VaryBetaSubplot{MAG}
{(1,1.651) (2,1.989) (3,2.115) (4,2.103) (5,1.922)}
{(1,1.263) (2,1.282) (3,1.356) (4,1.391) (5,1.475)}
{(1,2.403) (2,2.278) (3,1.549) (4,1.529) (5,1.597)}
{ymin=1.2,ymax=2.5,ytick={1.2,1.6,2.0,2.4},yticklabels={1.2,1.6,2.0,2.4}}

\end{small}
\vspace{-3mm}
\caption{Query time of \kalgo{} when varying $\beta$.}
\end{figure*}
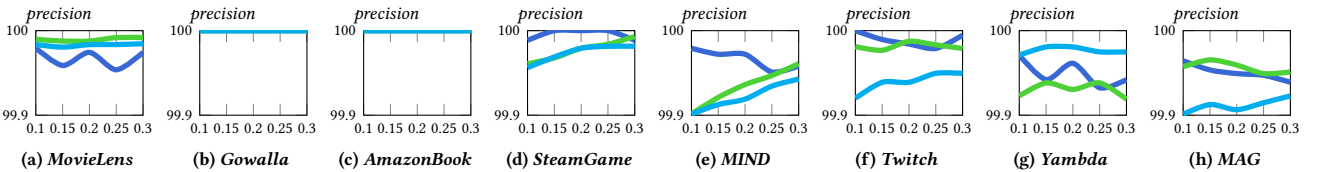
\begin{figure*}[!t]
\centering
\begin{small}

\newcommand{\VaryBetaPrecisionSubplot}[4]{%
\subfloat[{\em #1}]{%
\begin{tikzpicture}[scale=1,every mark/.append style={mark size=3pt}]
\begin{axis}[
    height=2.7cm,
    width=3.0cm,
    ylabel={\it precision},
    xmin=1, xmax=5,
    ymin=0.999, ymax=1,
    xtick={1,2,3,4,5},
    ytick={0.999,1},
    xticklabel style={font=\scriptsize},
    yticklabel style={font=\scriptsize},
    xticklabels={0.1,0.15,0.2,0.25,0.3},
    yticklabels={99.9,100},
    every axis y label/.style={
        font=\footnotesize,
        at={(current axis.north west)},
        right=3mm,above=0mm
    },
]
\addplot[line width=0.7mm,smooth,color=NSCcol1]
    plot coordinates {#2};
\addplot[line width=0.7mm,smooth,color=my_violet]
    plot coordinates {#3};
\addplot[line width=0.7mm,smooth,color=cyan]
    plot coordinates {#4};
\end{axis}
\end{tikzpicture}\hspace{0mm}%
}%
}

\vspace{-5mm}

\VaryBetaPrecisionSubplot{MovieLens}
{(1,0.999794) (2,0.999588) (3,0.999743) (4,0.999537) (5,0.999743)}
{(1,0.999897) (2,0.999877) (3,0.999877) (4,0.999918) (5,0.999918)}
{(1,0.999835) (2,0.999805) (3,0.999835) (4,0.999835) (5,0.999846)}
\VaryBetaPrecisionSubplot{Gowalla}
{(1,1.000000) (2,1.000000) (3,1.000000) (4,1.000000) (5,1.000000)}
{(1,1.000000) (2,1.000000) (3,1.000000) (4,1.000000) (5,1.000000)}
{(1,1.000000) (2,1.000000) (3,1.000000) (4,1.000000) (5,1.000000)}
\VaryBetaPrecisionSubplot{AmazonBook}
{(1,1.000000) (2,1.000000) (3,1.000000) (4,1.000000) (5,1.000000)}
{(1,1.000000) (2,1.000000) (3,1.000000) (4,1.000000) (5,1.000000)}
{(1,1.000000) (2,1.000000) (3,1.000000) (4,1.000000) (5,1.000000)}
\VaryBetaPrecisionSubplot{SteamGame}
{(1,0.999885) (2,1) (3,1) (4,1) (5,0.999885)}
{(1,0.999608) (2,0.999677) (3,0.999792) (4,0.999839) (5,0.999931)}
{(1,0.999562) (2,0.999689) (3,0.999792) (4,0.999815) (5,0.999815)}
\VaryBetaPrecisionSubplot{MIND}
{(1,0.999791) (2,0.999721) (3,0.999721) (4,0.999511) (5,0.999581)}
{(1,0.999015) (2,0.999211) (3,0.999362) (4,0.999469) (5,0.999609)}
{(1,0.999013) (2,0.999126) (3,0.999190) (4,0.999344) (5,0.999427)}
\VaryBetaPrecisionSubplot{Twitch}
{(1,1) (2,0.999895) (3,0.999842) (4,0.999789) (5,0.999947)}
{(1,0.999811) (2,0.999768) (3,0.999874) (4,0.999832) (5,0.999789)}
{(1,0.999200) (2,0.999389) (3,0.999389) (4,0.999495) (5,0.999495)}
\VaryBetaPrecisionSubplot{Yambda}
{(1,0.999710) (2,0.999421) (3,0.999614) (4,0.999324) (5,0.999421)}
{(1,0.999228) (2,0.999382) (3,0.999305) (4,0.999382) (5,0.999189)}
{(1,0.999710) (2,0.999807) (3,0.999807) (4,0.999749) (5,0.999749)}
\VaryBetaPrecisionSubplot{MAG}
{(1,0.999644) (2,0.999535) (3,0.999492) (4,0.999472) (5,0.999390)}
{(1,0.999573) (2,0.999654) (3,0.999593) (4,0.999492) (5,0.999512)}
{(1,0.999012) (2,0.999126) (3,0.999065) (4,0.999147) (5,0.999228)}

\end{small}
\vspace{-3mm}
\caption{Precision of \kalgo{} when varying $\beta$.}
\end{figure*}

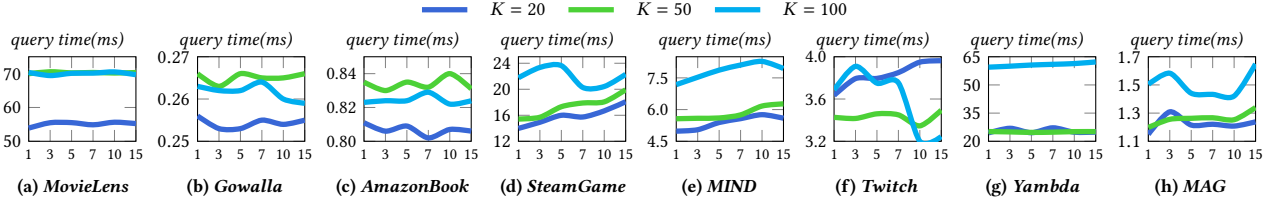
\begin{figure*}[!t]
\centering
\begin{small}
\newcommand{\VaryKappaSubplot}[5]{%
\subfloat[{\em #1}]{%
\begin{tikzpicture}[scale=1,every mark/.append style={mark size=3pt}]
\begin{axis}[
    height=2.7cm,
    width=3.0cm,
    ylabel={\it query time(ms)},
    xmin=1, xmax=6,
    xtick={1,...,6},
    xticklabel style={font=\scriptsize},
    yticklabel style={font=\footnotesize},
    xticklabels={1,3,5,7,10,15},
    scaled y ticks=false,
    every axis y label/.style={font=\footnotesize,at={(current axis.north west)},right=5mm,above=0mm},
    #5
]
\addplot[line width=0.7mm,smooth,color=NSCcol1] plot coordinates {#2};
\addplot[line width=0.7mm,smooth,color=my_violet] plot coordinates {#3};
\addplot[line width=0.7mm,smooth,color=cyan] plot coordinates {#4};
\end{axis}
\end{tikzpicture}\hspace{0mm}%
}%
}

\begin{tikzpicture}
\begin{customlegend}[
    legend columns=3,
    legend entries={$K=20$,$K=50$,$K=100$},
    legend style={at={(0.45,1.35)},anchor=north,draw=none,font=\footnotesize,column sep=0.2cm}
]
\addlegendimage{line width=0.7mm,mark size=4pt,color=NSCcol1}
\addlegendimage{line width=0.7mm,mark size=4pt,color=my_violet}
\addlegendimage{line width=0.7mm,mark size=4pt,color=cyan}
\end{customlegend}
\end{tikzpicture}
\\[-\lineskip]
\vspace{-4mm}
\VaryKappaSubplot{MovieLens}
{(1,53.897) (2,55.492) (3,55.517) (4,54.876) (5,55.640) (6,55.228)}
{(1,69.980) (2,70.568) (3,70.232) (4,70.350) (5,70.188) (6,70.333)}
{(1,70.362) (2,69.434) (3,70.109) (4,70.162) (5,70.516) (6,69.679)}
{ymin=50,ymax=75,ytick={50,60,70}}
\VaryKappaSubplot{Gowalla}
{(1,0.256) (2,0.253) (3,0.253) (4,0.255) (5,0.254) (6,0.255)}
{(1,0.266) (2,0.263) (3,0.266) (4,0.265) (5,0.265) (6,0.266)}
{(1,0.263) (2,0.262) (3,0.262) (4,0.264) (5,0.260) (6,0.259)}
{ymin=0.25,ymax=0.27,ytick={0.25,0.26,0.27}}
\VaryKappaSubplot{AmazonBook}
{(1,0.811) (2,0.806) (3,0.809) (4,0.802) (5,0.807) (6,0.806)}
{(1,0.835) (2,0.830) (3,0.835) (4,0.832) (5,0.840) (6,0.831)}
{(1,0.823) (2,0.824) (3,0.824) (4,0.829) (5,0.822) (6,0.824)}
{ymin=0.80,ymax=0.85,ytick={0.80,0.82,0.84},yticklabels={0.80,0.82,0.84}}
\VaryKappaSubplot{SteamGame}
{(1,13.985) (2,14.925) (3,15.996) (4,15.736) (5,16.674) (6,18.086)}
{(1,15.426) (2,15.695) (3,17.316) (4,17.891) (5,18.062) (6,19.928)}
{(1,21.770) (2,23.355) (3,23.631) (4,20.280) (5,20.423) (6,22.305)}
{ymin=12,ymax=25,ytick={12,16,20,24}}
\VaryKappaSubplot{MIND}
{(1,4.980) (2,5.047) (3,5.389) (4,5.570) (5,5.759) (6,5.591)}
{(1,5.567) (2,5.588) (3,5.603) (4,5.749) (5,6.168) (6,6.270)}
{(1,7.181) (2,7.530) (3,7.867) (4,8.108) (5,8.280) (6,7.943)}
{ymin=4.5,ymax=8.5,ytick={4.5,6.0,7.5},yticklabels={4.5,6.0,7.5}}
\VaryKappaSubplot{Twitch}
{(1,3.634) (2,3.793) (3,3.795) (4,3.850) (5,3.949) (6,3.962)}
{(1,3.427) (2,3.416) (3,3.460) (4,3.450) (5,3.346) (6,3.495)}
{(1,3.685) (2,3.909) (3,3.750) (4,3.754) (5,3.200) (6,3.245)}
{ymin=3.2,ymax=4.0,ytick={3.2,3.6,4.0},yticklabels={3.2,3.6,4.0}}
\VaryKappaSubplot{Yambda}
{(1,24.479) (2,26.925) (3,24.566) (4,27.176) (5,24.843) (6,24.888)}
{(1,25.178) (2,24.995) (3,24.722) (4,24.852) (5,25.142) (6,25.154)}
{(1,59.458) (2,59.978) (3,60.594) (4,60.969) (5,61.365) (6,62.234)}
{ymin=20,ymax=65,ytick={20,35,50,65}}
\VaryKappaSubplot{MAG}
{(1,1.149) (2,1.309) (3,1.215) (4,1.221) (5,1.208) (6,1.238)}
{(1,1.198) (2,1.256) (3,1.263) (4,1.267) (5,1.253) (6,1.340)}
{(1,1.504) (2,1.583) (3,1.441) (4,1.433) (5,1.421) (6,1.649)}
{ymin=1.1,ymax=1.7,ytick={1.1,1.3,1.5,1.7}}
\end{small}
\vspace{-3mm}
\caption{Query time of \kalgo{} when varying $\kappa$.}
\label{fig:vary-time-kappa}
\end{figure*}
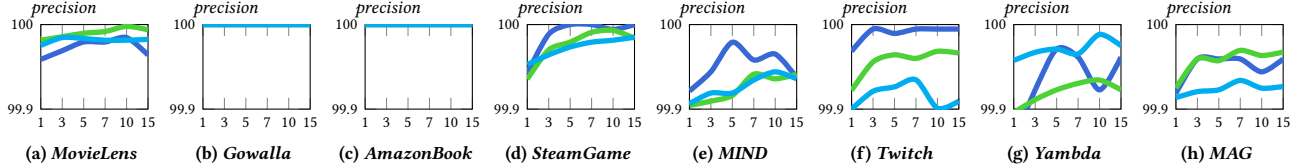
\begin{figure*}[!t]
\centering
\begin{small}

\newcommand{\VaryKappaPrecisionSubplot}[4]{%
\subfloat[{\em #1}]{%
\begin{tikzpicture}[scale=1,every mark/.append style={mark size=3pt}]
\begin{axis}[
    height=2.7cm,
    width=3.0cm,
    ylabel={\it precision},
    xmin=1, xmax=6,
    ymin=0.999, ymax=1,
    xtick={1,2,3,4,5,6},
    ytick={0.999,1},
    xticklabel style={font=\scriptsize},
    yticklabel style={font=\scriptsize},
    xticklabels={1,3,5,7,10,15},
    yticklabels={99.9,100},
    every axis y label/.style={
        font=\footnotesize,
        at={(current axis.north west)},
        right=3mm,above=0mm
    },
]
\addplot[line width=0.7mm,smooth,color=NSCcol1]
    plot coordinates {#2};
\addplot[line width=0.7mm,smooth,color=my_violet]
    plot coordinates {#3};
\addplot[line width=0.7mm,smooth,color=cyan]
    plot coordinates {#4};
\end{axis}
\end{tikzpicture}\hspace{0mm}%
}%
}

\vspace{-5mm}

\VaryKappaPrecisionSubplot{MovieLens}
{(1,0.999588) (2,0.999691) (3,0.999794) (4,0.999794) (5,0.999846) (6,0.999640)}
{(1,0.999815) (2,0.999856) (3,0.999897) (4,0.999918) (5,0.999979) (6,0.999938)}
{(1,0.999753) (2,0.999846) (3,0.999835) (4,0.999815) (5,0.999815) (6,0.999825)}
\VaryKappaPrecisionSubplot{Gowalla}
{(1,1.000000) (2,1.000000) (3,1.000000) (4,1.000000) (5,1.000000) (6,1.000000)}
{(1,1.000000) (2,1.000000) (3,1.000000) (4,1.000000) (5,1.000000) (6,1.000000)}
{(1,1.000000) (2,1.000000) (3,1.000000) (4,1.000000) (5,1.000000) (6,1.000000)}
\VaryKappaPrecisionSubplot{AmazonBook}
{(1,1.000000) (2,1.000000) (3,1.000000) (4,1.000000) (5,1.000000) (6,1.000000)}
{(1,1.000000) (2,1.000000) (3,1.000000) (4,1.000000) (5,1.000000) (6,1.000000)}
{(1,1.000000) (2,1.000000) (3,1.000000) (4,1.000000) (5,1.000000) (6,1.000000)}
\VaryKappaPrecisionSubplot{SteamGame}
{(1,0.999423) (2,0.999885) (3,1) (4,1) (5,0.999942) (6,1)}
{(1,0.999354) (2,0.999700) (3,0.999792) (4,0.999908) (5,0.999931) (6,0.999839)}
{(1,0.999527) (2,0.999642) (3,0.999735) (4,0.999792) (5,0.999815) (6,0.999850)}
\VaryKappaPrecisionSubplot{MIND}
{(1,0.999215) (2,0.999441) (3,0.999791) (4,0.999581) (5,0.999651) (6,0.999372)}
{(1,0.999043) (2,0.999094) (3,0.999162) (4,0.999413) (5,0.999358) (6,0.999413)}
{(1,0.999064) (2,0.999190) (3,0.999190) (4,0.999344) (5,0.999441) (6,0.999358)}
\VaryKappaPrecisionSubplot{Twitch}
{(1,0.999684) (2,0.999947) (3,0.999895) (4,0.999947) (5,0.999947) (6,0.999947)}
{(1,0.999221) (2,0.999558) (3,0.999642) (4,0.999600) (5,0.999684) (6,0.999663)}
{(1,0.999000) (2,0.999211) (3,0.999263) (4,0.999347) (5,0.999011) (6,0.999095)}
\VaryKappaPrecisionSubplot{Yambda}
{(1,0.998649) (2,0.999228) (3,0.999710) (4,0.999614) (5,0.999228) (6,0.999614)}
{(1,0.998958) (2,0.999112) (3,0.999228) (4,0.999305) (5,0.999344) (6,0.999228)}
{(1,0.999575) (2,0.999672) (3,0.999710) (4,0.999653) (5,0.999884) (6,0.999749)}
\VaryKappaPrecisionSubplot{MAG}
{(1,0.999187) (2,0.999593) (3,0.999593) (4,0.999593) (5,0.999441) (6,0.999593)}
{(1,0.999248) (2,0.999593) (3,0.999573) (4,0.999695) (5,0.999634) (6,0.999675)}
{(1,0.999136) (2,0.999207) (3,0.999228) (4,0.999339) (5,0.999248) (6,0.999268)}

\end{small}
\vspace{-3mm}
\caption{Precision of \kalgo{} when varying $\kappa$.}
\label{fig:vary-precision-kappa}
\end{figure*}

\update{
This section further studies the sensitivity and transferability of the parameters in \algo{} and \kalgo{} across all eight datasets, as shown in Figures~\ref{fig:asc-vary-time-rho}--\ref{fig:vary-precision-kappa}. Specifically, we vary $\rho\in\{0.2,0.3,\ldots,0.8\}$ for \algo{}, and $\rho\in\{0.3,0.4,\ldots,0.9\}$, $\beta\in\{0.1,0.15,0.2,0.25,0.3\}$, and $\kappa\in\{1,3,5,7,10,15\}$ for \kalgo{}, while fixing the other parameters and calibrating $\epsilon$ to maintain top-$K$ precision above $99.9\%$. The remaining parameters follow the calibrated settings reported in Appendix~\ref{sec:para-setting}.
}

\stitle{\update{Varying $\rho$ in \algo{}}}
\update{We first examine the sensitivity of $\rho$ in \algo{} across all eight datasets. The precision remains consistently above $99.9\%$ over the tested range, while the query time varies only moderately on most datasets. Although some medium-scale graphs exhibit larger fluctuations, no single value is uniformly optimal across all datasets. We therefore use $\rho=0.5$ as a transferable default for \algo{}, which provides stable accuracy and competitive efficiency without dataset-specific tuning.}

\stitle{\update{Varying $\rho$ in \kalgo{}}}
\update{The effect of $\rho$ is consistent across different values of $K$. On {\em MovieLens}, {\em SteamGame}, {\em MIND}, and {\em Twitch}, increasing $\rho$ generally reduces query time until the improvement saturates around $0.7$--$0.8$. {\em Gowalla} and {\em AmazonBook} are considerably less sensitive because their query times are already small, while {\em Yambda} and {\em MAG} exhibit a less monotonic trend. These results indicate that $\rho$ mainly controls the efficiency of adaptive processing rather than the final accuracy once it is in a reasonable range. We therefore use $\rho=0.8$ as a transferable default for \kalgo{}, without separately tuning it for each dataset or $K$.}

\stitle{\update{Varying $\beta$}}
\update{Increasing $\beta$ generally increases query time on {\em MovieLens}, {\em SteamGame}, {\em MIND}, {\em Twitch}, and {\em Yambda}, whereas the effect is much smaller on {\em Gowalla} and {\em AmazonBook}. Thus, $\beta$ primarily determines the amount of computation devoted to difficult or borderline candidates and provides an efficiency and refinement trade-off rather than a sharp accuracy transition. We use $\beta=0.2$ as a robust initialization for a new dataset and tune it only when optimizing the best query time under the precision constraint.}

\stitle{\update{Varying $\kappa$}}
\update{The effect of $\kappa$ on query time is substantially weaker than that of $\rho$ or $\beta$ for most datasets. Once enough extra candidates are retained, increasing $\kappa$ gives little additional benefit and only introduces modest extra refinement cost. In contrast, an overly small value such as $\kappa=1$ can reduce precision, most visibly on {\em Yambda}. The results are stable across $K=20,50,100$, suggesting that $\kappa$ does not need to scale with $K$ in the tested range. We therefore fix $\kappa=5$, which lies near the knee of the accuracy--efficiency trade-off and transfers well across datasets and different $K$ values.}

\update{
\subsection{Adaptation of Similarity Search Methods}
\label{app:baseline_adaptation}
Since \texttt{All-Pairs}~\cite{bayardo2007scaling} and \texttt{WHIMP}~\cite{sharma2017hashes} were originally designed to retrieve all pairs whose cosine similarity exceeds a predefined threshold $\tau$, their optimizations based on inner products cannot be directly applied to the more complex computation of Swing scores. 
Therefore, we adapt them as candidate generation methods for the Swing retrieval setting. 
Each item is represented by its historical interaction set, and the query interfaces of both methods are modified to support queries from an arbitrary source item.
}

\update{
Specifically, we select $\tau$ from a set of values within $[0.01, 0.5]$ according to the retrieval performance on the validation set for each baseline. 
For \texttt{All-Pairs}, instead of constructing and updating the inverted index during pairwise retrieval, we construct the inverted index over the entire graph in advance for the selected value of $\tau$.
For \texttt{WHIMP}, we precompute the SimHash signatures of all items. 
These preprocessing modifications allow any item to be used as the source node at query time. Apart from these preprocessing and query interface modifications, we retain the original candidate filtering and pruning mechanisms of both methods.
Using the selected threshold, each baseline first generates candidates based on cosine similarity.
We vary the maximum candidate pool size over $\{150, 200, 300, 400, 500\}$ to achieve the best trade-off between retrieval effectiveness and computational efficiency.
The Swing scores of the retained candidates are then computed for final ranking and evaluation. 
On large-scale graphs, where exact Swing computation is expensive, we accelerate this step using our adaptive pairwise Algorithm~8, treating the candidate set as the borderline set $\mathcal{B}$ and setting the sample size $n$ according to Eq.~\eqref{eq:comp-n}.
}

\update{
\subsection{Item-to-item Retrieval Setting}\label{sec:i2i-setting}
We compare against representative baselines from three categories:
\begin{itemize}[leftmargin=*,noitemsep]
\item neighborhood-based methods: \texttt{ItemCF}~\cite{sarwar2001item}, \texttt{Adamic-Adar}~\cite{adamic2003friends}, \texttt{Jaccard} and \texttt{Partial Correlation}~\cite{hayashi2022rethinking};
\item high-order proximity methods: \texttt{PPR}~\cite{wei2018topppr} and \texttt{SimRank}~\cite{kusumoto2014scalable};
\item embedding-based methods: \texttt{Item2vec}~\cite{barkan2016item2vec}.
\end{itemize}
The same evaluation protocol and candidate filtering strategy are applied to all methods.
Each dataset is randomly split into $80\%$ training and $20\%$ test interactions, and 200 test users are randomly sampled for evaluation.
For each user, the training items are used as source items, and each i2i method computes relevance scores from every source item to candidate items.
The resulting score vectors are L1-normalized and aggregated using max, mean, or sum pooling, and the best-performing aggregation strategy is adopted.
After filtering items observed in the training history, we rank the remaining candidates by their aggregated scores and retrieve the top-$K$ items for $K\in\{20,50,100\}$.
We report Recall@$K$ averaged over the evaluated users and the average online retrieval time per user, measured as the cumulative i2i query time over all source items.
}

\subsection{\update{Discussion on Scalability and Sharded Deployment}}\label{sec:discuss-shard}

\update{
Although our experiments are conducted on a single node, this setting is intended to isolate the computational efficiency and memory footprint of the proposed single-source top-$K$ Swing query processing, rather than to impose a single-machine deployment assumption. After the graph and auxiliary data structures are constructed, the online processing of a query node $v_q$ primarily accesses its neighboring users and the neighborhood information needed to evaluate the corresponding user pairs, without scanning the entire graph. Therefore, in a large-scale deployment, adjacency lists and auxiliary graph data can be partitioned across multiple shards and fetched from the corresponding workers on demand.
}

\update{
The query computation is also naturally decomposable. The Swing score of each candidate item is obtained by accumulating the contributions of individual user pairs associated with $v_q$. Thus, these user pairs can be partitioned among multiple workers, each of which independently computes partial candidate scores using the required neighborhood information. The partial scores can then be summed across workers, followed by a global top-$K$ selection. Since this partitioning only distributes the underlying pairwise computations and does not alter the estimator itself, it does not affect the approximation guarantee of our method.}

\section{Theoretical Proofs}\label{sec:proof}

\begin{proof}[\bf Proof of Lemma~\ref{lem:swing-range}]
Since each term $\frac{1}{\alpha + |\I(u_i) \cap \I(u_j)|} > 0$, by the definition of $\textsf{sw}(v_q,v_t)$ in Eq.~\eqref{eq:swing}, $\textsf{sw}(v_q,v_t)$ attains the minimum value of $0$ when $\U(v_q)\cap \U(v_t)$ is an empty set.
On the other hand, since $v_q \neq v_t$, $u_i, u_j\in \U(v_q)\cap \U(v_t)$ indicates that $v_q,v_t\in \I(u_i)\cap \I(u_j)$, i.e., $|\I(u_i)\cap \I(u_j)|\ge 2$. Therefore, we can derive that
\begin{align*}
\textsf{sw}(v_q,v_t) &= \sum_{u_i, u_j\in \U(v_q)\cap \U(v_t)}{\frac{1}{\alpha+|\I(u_i)\cap \I(u_j)|}} \\
& \le \sum_{u_i, u_j\in \U(v_q)}{\frac{1}{\alpha+|\I(u_i)\cap \I(u_j)|}} \\
& \le d(v_q)\cdot (d(v_q)-1)\cdot \frac{1}{\alpha+2} = \frac{d(v_q)^2-d(v_q )}{\alpha+2}.
\end{align*}
The lemma is proved.
\end{proof}

\begin{proof}[\bf Proof of Lemma~\ref{lem:exact-time}]
First, computing the set intersection in Line 4 takes $O(\min(d(u_i), d(u_j)))$ time. For all valid user pairs, we have
\begin{small}
\begin{align*}
\sum_{u_i\neq u_j\in \U(v_q)}{\frac{\min(d(u_i), d(u_j))}{2}} &\le  \frac{1}{4}\sum_{u_i\neq u_j\in \U(v_q)}{d(u_i)+d(u_j)}\\
& = \frac{1}{2} \sum_{u_i\neq u_j\in \U(v_q)}{d(u_i)}\\
& = \frac{1}{2} \left(\sum_{u_i, u_j\in \U(v_q)}{d(u_i)}-\sum_{u_i\in \U(v_q)}{d(u_i)}\right)\\
&=\frac{d(v_q)-1}{2}\sum_{u\in \U(v_q)}{d(u)}.
\end{align*}
\end{small}
Consequently, the worst-case time complexity of Algorithm~\ref{alg:basic} is then $O\left(\frac{d(v_q)-1}{2}\cdot D(v_q)\right)$.
\end{proof}

\begin{proof}[\bf Proof of Theorem~\ref{lem:naive-mc}]

\begin{lemma}[Chernoff Bound~\cite{chung2006concentration}]\label{lem:chernoff}
For a set of $\{Z_i\}_{i=1}^n$ i.i.d. random variables with $Z_i\in [0,\zeta]$, let $Z=\sum_{i=1}^n{Z_i}$ and $\mathbb{E}[Z]$ be its mean. Then,
\begin{small}
\begin{equation}
\mathbb{P}\left[\left|Z-\mathbb{E}[Z]\right|\ge \varepsilon\right] \le \exp{\left(-\frac{\varepsilon^2}{2\zeta(\varepsilon/3+\mathbb{E}[Z])}\right)}.
\end{equation}
\end{small}
\end{lemma}

Define the random variable for the $r$-th trial (with user pairs $u_i$ and $u_j$) is
\[Z_r = \frac{d(v_q)\cdot (d(v_q)-1)}{\alpha+|\I_{i,j}|} \cdot \mathbb{1}[v_t\in \I_{i,j}].\]
According to our analysis in \S\ref{sec:baseline}, 
\[\widetilde{\textsf{\textnormal{sw}}}(v_q, v_t) = \frac{1}{n}\sum_{r=1}^{n}{Z_r}\]
is an unbiased estimator.
Notice that $v_q,v_t\in \I_{i,j}$, and then, $2\le |\I_{i,j}|\le \min\{d(u_i), d(u_j)\}$. The random variable $Z_r$ can be bounded by
\begin{align*}
Z_r & \le \frac{d(v_q)\cdot (d(v_q)-1)}{\alpha+|\I_{i,j}|} \le \frac{d(v_q)\cdot (d(v_q)-1)}{\alpha+2} = \zeta.
\end{align*}
By letting $\varepsilon=n\cdot \epsilon\cdot {\textsf{\textnormal{sw}}}(v_q, v_t)$ and applying Lemma~\ref{lem:chernoff}, we can derive
\begin{align*}
& \mathbb{P}\left[\left|\widetilde{\textsf{\textnormal{sw}}}(v_q, v_t)-{\textsf{\textnormal{sw}}}(v_q, v_t)\right| \ge \epsilon\cdot {\textsf{\textnormal{sw}}}(v_q, v_t)\right] \\
&= \mathbb{P}\left[\left|n\cdot\widetilde{\textsf{\textnormal{sw}}}(v_q, v_t)-{n\cdot\textsf{\textnormal{sw}}}(v_q, v_t)\right| \ge n\cdot\epsilon\cdot {\textsf{\textnormal{sw}}}(v_q, v_t)\right]\\
& \le \exp{\left(-\frac{(n\cdot \epsilon\cdot {\textsf{\textnormal{sw}}}(v_q, v_t))^2}{2\zeta(\frac{n \cdot \epsilon\cdot {\textsf{\textnormal{sw}}}(v_q, v_t)}{3}+ {n\cdot \textsf{\textnormal{sw}}}(v_q, v_t))}\right)}\\
& = \exp{\left(-\frac{n\cdot \epsilon^2\cdot {\textsf{\textnormal{sw}}}(v_q, v_t)}{2\zeta(\frac{\epsilon}{3}+ 1)}\right)}.
\end{align*}
Note that $n\ge \frac{2\zeta(1+\epsilon/3)}{\epsilon^2\lambda}\log\frac{1}{\delta}$. Then we have 
\begin{align*}
& \mathbb{P}\left[\left|\widetilde{\textsf{\textnormal{sw}}}(v_q, v_t)-{\textsf{\textnormal{sw}}}(v_q, v_t)\right| \ge \epsilon\cdot {\textsf{\textnormal{sw}}}(v_q, v_t)\right] \\
& \le \exp{\left(-\frac{{\textsf{\textnormal{sw}}}(v_q, v_t)}{\lambda}\log\frac{1}{\delta}\right)}.
\end{align*}
When ${\textsf{\textnormal{sw}}}(v_q, v_t)\ge \lambda$, we have
\begin{equation*}
\mathbb{P}\left[\left|\widetilde{\textsf{\textnormal{sw}}}(v_q, v_t)-{\textsf{\textnormal{sw}}}(v_q, v_t)\right| \ge \epsilon\cdot {\textsf{\textnormal{sw}}}(v_q, v_t)\right] \le \delta.
\end{equation*}
Applying Lemma~\ref{lem:chernoff} with $\varepsilon=n \cdot \epsilon\cdot \lambda$, we can derive 
\begin{align*}
& \mathbb{P}\left[\left|\widetilde{\textsf{\textnormal{sw}}}(v_q, v_t)-{\textsf{\textnormal{sw}}}(v_q, v_t)\right| \ge \epsilon\cdot \lambda \right] \\
& \le \exp{\left(-\frac{n\cdot (\epsilon\cdot\lambda)^2}{2\zeta(\frac{\epsilon\cdot\lambda}{3}+ {\textsf{\textnormal{sw}}}(v_q,v_t))}\right)}.
\end{align*}
When ${\textsf{\textnormal{sw}}}(v_q, v_t)< \lambda$, we have
\begin{equation*}
\mathbb{P}\left[\left|\widetilde{\textsf{\textnormal{sw}}}(v_q, v_t)-{\textsf{\textnormal{sw}}}(v_q, v_t)\right| \ge \epsilon\cdot \lambda\right] \le \delta,
\end{equation*}
which finishes the proof.

\end{proof}

\begin{proof}[\bf Proof of Theorem~\ref{lem:USS}]
\update{
Consider the $r$-th random event that catches target item $v_t$ in the subset $\I^{\prime}_{i,j}$ and yields an update $Z_r$ to $\widetilde{\textsf{\textnormal{sw}}}(v_q, v_t)$.
Let $\mathbb{1}_{v_t\in \I^{\prime}_{i,j}}$ be the indicator variable where $\mathbb{1}_{v_t\in \I^{\prime}_{i,j}}=1$ if $v_t\in \I^{\prime}_{i,j}$ and $0$ otherwise. Since $\I^{\prime}_{i,j}$ is obtained as Line 4 in Algorithm~\ref{alg:subset_sample} and $\I^{\prime}(u_i)$ is drawn uniformly at random from $\I(u_i)$, we have
\begin{align*}
\mathbb{E}[\mathbb{1}_{v_t\in\I^\prime_{i,j}}]
&=
\frac{|\I^\prime(u_i)|}{d(u_i)}
=
\frac{\lceil\gamma\cdot d(u_i)\rceil}{d(u_i)}.
\end{align*}
By Line 7 in Algorithm~\ref{alg:subset_sample}, the expected value of the update score target item $v_t$ is
\begin{equation}
\mathbb{E}[Z_r] = \frac{1}{\frac{|\I^\prime(u_i)|}{d(u_i)}\cdot (\alpha+|\I_{i,j}|)} \cdot \mathbb{E}[\mathbb{1}_{v_t\in \I^{\prime}_{i,j}}] = \frac{1}{\alpha+|\I_{i,j}|}.
\end{equation}
By the linearity of expectation over all user pairs sharing $v_t$, we have
\begin{align*}
\mathbb{E}[
\widetilde{\textsf{\textnormal{sw}}}(v_q,v_t)]
=
{\textsf{\textnormal{sw}}}(v_q,v_t),
\end{align*}
indicating that
$\widetilde{\textsf{\textnormal{sw}}}(v_q,v_t)$ is an unbiased estimator.
}
The random variable $Z_r$ can be bounded by
\begin{align*}
Z_r &\le \max_{r}{\frac{1}{\frac{|\I^\prime(u_i)|}{d(u_i)}\cdot (\alpha+|\I_{i,j}|)}}\\
& = \max_{r}{\frac{1}{\frac{\lceil \gamma\cdot d(u_i)\rceil}{d(u_i)}\cdot (\alpha+|\I_{i,j}|) }}.
\end{align*}
Note that $\gamma\le \min\left\{{\frac{2(\epsilon/3+1)}{(\alpha+2)\cdot\lambda\epsilon^2 }\cdot\log\left(\frac{1}{\delta}\right)}, 1\right\}$. Then, we have
\begin{equation*}
Z_r \le \frac{\lambda\epsilon^2}{2(\epsilon/3+1)}\cdot \log\left(\frac{1}{\delta}\right) =\zeta.
\end{equation*}

By letting $\varepsilon=\epsilon\cdot {\textsf{\textnormal{sw}}}(v_q, v_t)$ and applying Lemma~\ref{lem:chernoff}, we can derive
\begin{align*}
& \mathbb{P}\left[\left|\widetilde{\textsf{\textnormal{sw}}}(v_q, v_t)-{\textsf{\textnormal{sw}}}(v_q, v_t)\right| \ge \epsilon\cdot {\textsf{\textnormal{sw}}}(v_q, v_t)\right] \\
& \le \exp{\left(-\frac{(\epsilon\cdot {\textsf{\textnormal{sw}}}(v_q, v_t))^2}{2\zeta(\frac{\epsilon\cdot {\textsf{\textnormal{sw}}}(v_q, v_t)}{3}+ {\textsf{\textnormal{sw}}}(v_q, v_t))}\right)}\\
& = \exp{\left(-\frac{\epsilon^2\cdot {\textsf{\textnormal{sw}}}(v_q, v_t)}{2\zeta(\frac{\epsilon}{3}+ 1)}\right)}.
\end{align*}
When ${\textsf{\textnormal{sw}}}(v_q, v_t)\ge \lambda$, we have
\begin{equation*}
\mathbb{P}\left[\left|\widetilde{\textsf{\textnormal{sw}}}(v_q, v_t)-{\textsf{\textnormal{sw}}}(v_q, v_t)\right| \ge \epsilon\cdot {\textsf{\textnormal{sw}}}(v_q, v_t)\right] \le \delta.
\end{equation*}

Applying Lemma~\ref{lem:chernoff} with $\varepsilon=\epsilon\cdot \lambda$, we can derive 
\begin{align*}
& \mathbb{P}\left[\left|\widetilde{\textsf{\textnormal{sw}}}(v_q, v_t)-{\textsf{\textnormal{sw}}}(v_q, v_t)\right| \ge \epsilon\cdot \lambda\right] \\
& \le \exp{\left(-\frac{(\epsilon\cdot\lambda)^2}{2\zeta(\frac{\epsilon\cdot \lambda}{3}+ {\textsf{\textnormal{sw}}}(v_q, v_t))}\right)}.
\end{align*}
When ${\textsf{\textnormal{sw}}}(v_q, v_t)< \lambda$, we have
\begin{equation*}
\mathbb{P}\left[\left|\widetilde{\textsf{\textnormal{sw}}}(v_q, v_t)-{\textsf{\textnormal{sw}}}(v_q, v_t)\right| \ge \epsilon\cdot \lambda\right] \le \delta,
\end{equation*}
which finishes the proof.
\end{proof}

\begin{proof}[\bf Proof of Theorem~\ref{lem:cost-ASC}]
When $\frac{d(v_q)(d(v_q)-1)}{2}\le n$, the runtime for \texttt{GNS} can be simplified as
\begin{small}
\begin{equation*}
T_{\texttt{GNS}}=\rho_{\texttt{GNS}} \cdot \left(n+\frac{d(v_q)(d(v_q)-1)}{2}\cdot\frac{D(v_q)}{d(v_q)}\right) \ge \rho_{\texttt{GNS}}\cdot n +\frac{1}{\rho}\cdot T_{\texttt{Exact}},
\end{equation*}
\end{small}
which is definitely more costly than \texttt{Exact}.
In the meanwhile, notice that in this case, $\frac{d(v_q)(d(v_q)-1)}{2}\cdot \rho\le n$, indicating that Eq.~\eqref{eq:select} holds and \texttt{Exact} will be selected for execution as expected.

On the other hand, when $n<\frac{d(v_q)(d(v_q)-1)}{2}$, we have $T_{\texttt{GNS}} = \rho_{\texttt{GNS}} \cdot n\cdot\left(1+\frac{D(v_q)}{d(v_q)}\right)$.
If Eq.~\eqref{eq:select} is satisfied, namely \texttt{USS} is executed by \algo{}), we can derive
\begin{small}
\begin{align*}
&\frac{d(v_q)\cdot (d(v_q)-1)}{2}\cdot \rho \le n\left(1+\frac{d(v_q)}{D(v_q)}\right)\\
\Leftrightarrow & \frac{d(v_q)\cdot (d(v_q)-1)}{2}\cdot \rho_{\texttt{USS}} \le n\left(1+\frac{d(v_q)}{D(v_q)}\right)\cdot \rho_{\texttt{GNS}} \\
\Leftrightarrow & \left(\frac{d(v_q)\cdot (d(v_q)-1)}{2}\cdot \frac{D(v_q)}{d(v_q)}\right)\cdot \rho_{\texttt{USS}} \le n\left(\frac{D(v_q)}{d(v_q)}+1\right)\cdot \rho_{\texttt{GNS}} \\
\Leftrightarrow & T_{\texttt{USS}} \le T_{\texttt{GNS}}.
\end{align*}
\end{small}
If Eq.~\eqref{eq:select} does not hold, \algo{} invokes \texttt{GNS}. We can also derive that $T_{\texttt{USS}} > T_{\texttt{GNS}}$, which finishes the proof.
\end{proof}

\begin{proof}[\bf Proof of Lemma~\ref{lem:welford}]
\begin{lemma}[Chan's Algorithm~\cite{chan1982updating}]\label{lem:chan_pairwise}
Let $A = \{x_1, \dots, x_m\}$ and $B = \{x_{m+1}, \dots, x_{m+n}\}$ be two samples. We have 
\begin{align*}
    T_{1,m} &= \sum_{i=1}^{m} x_{i}, & S_{1,m} &= \sum_{i=1}^{m} (x_{i} - \frac{1}{m}T_{1,m})^2, \\
    T_{m+1,m+n} &= \sum_{i=m+1}^{m+n} x_{i}, & S_{m+1,m+n} &= \sum_{i=m+1}^{m+n} (x_{i} - \frac{1}{n}T_{m+1,m+n})^2.
\end{align*}
Let $C = A \cup B$ be the combined sample. It can be shown that
\begin{align*}
    T_{1,m+n} &= T_{1,m} + T_{m+1,m+n}, \\
    S_{1,m+n} &= S_{1,m} + S_{m+1,m+n} + \frac{m}{n(m+n)}(\frac{n}{m}T_{1,m}-T_{m+1,m+n})^2.
\end{align*}
\end{lemma}
We now apply Lemma~\ref{lem:chan_pairwise}, setting $A = \{Z_1, \dots, Z_{n-c}\}$ to be the set of the first $n-c$ samples and $B = \{Z_{n-c+1}, \dots, Z_{n}\}$ to be the new batch of $c$ samples whose values are all $w$,
\begin{align*}
    T_{1,n-c} &= (n-c)\mu_{n-c}, & S_{1,n-c} &= (n-c-1)\sigma_{n-c}^2, \\
    T_{n-c+1,n} & = cw, & S_{n-c+1,n} &=0.
\end{align*}
Then, we derive that
\begin{equation*}
    \mu_n = \frac{1}{n}T_{1,n} = \frac{(n-c)\mu_{n-c} +cw}{n} = \mu_{n-c} + \frac{c(w-\mu_{n-c})}{n},
\end{equation*}
\begin{align*}
    \sigma^2_{n}&=\frac{1}{n-1}S_{1,n}\\
    &= \frac{1}{n-1}((n-c-1)\sigma_{n-c}^2+\frac{n-c}{cn}(\frac{c(n-c)\mu_{n-c}}{n-c}-cw)^2 \\
    &= \sigma^2_{n-c} + \frac{1}{n-1}(\frac{n-c}{cn}(c\mu_{n-c}-cw)^2 -c\sigma_{n-c}^2)\\
    &= \sigma^2_{n-c} + \frac{c}{n-1}(\frac{n-c}{n}(\mu_{n-c}-w)^2 -\sigma_{n-c}^2).
\end{align*}
\end{proof}

\begin{proof}[\bf Proof of Lemma~\ref{lem:mean-var}]
First, we need the following lemma.
\begin{lemma}\label{lem:welford}
Let $\{Z_i\}_{i=1}^n$ be a sequence of $n$ random variables. Denote the sample mean and variance of the first $(n-c)$ observations as $\mu_{n-c}=\frac{1}{n-c}{\sum_{i=1}^{n-c}Z_i}$ and $\sigma_{n-c}=\frac{1}{n-c-1}\sum_{i=1}^{n-c}(Z_i-\mu_{n-c})^2$, respectively. Suppose that the next $c$ observations are identical and equal to $x$, i.e., $Z_{n-c+1}=\cdots=Z_{n}=x$. Then, the updated sample mean $\mu_{n}$ and variance $\sigma_{n}$ over all $n$ observations are
{
\small
\setlength{\arraycolsep}{2pt}
\begin{equation*}
\begin{array}{cc}
\mu_n = \mu_{n-c} + \dfrac{c(x-\mu_{n-c})}{n},
&
\sigma_n = \sigma_{n-c} + \dfrac{c\left(\dfrac{n-c}{n}(x-\mu_{n-c})^2-\sigma_{n-c}\right)}{n-1}.
\end{array}
\end{equation*}
}
\end{lemma}
Algorithm~\ref{alg:candidate} computes $\widetilde{\textsf{\textnormal{sw}}}(v_q,v_t)$ and $\sigma(v_t)$ in two parts.
Firstly, the flow-processing section (Lines 6-11) can be interpreted as 
adding $|S(u_i,u_j)|$ identical observations at each iteration, each with an increment 
$\frac{d(v_q)\cdot(d(v_q)-1)}{\alpha + |\I_{i,j}|}$. 
For each target item $v_t$, $r(v_t)$ records the number of accumulated observations. 
Applying Lemma~\ref{lem:welford} to this grouped update yields the update formulas 
for the sample mean $\widetilde{\textsf{sw}}(v_q,v_t)$ in Eq.~\eqref{eq:update-mu} 
and the sample variance $\sigma(v_t)$ in Eq.~\eqref{eq:update-sigma}, 
which update the statistics from $r(v_t)$ to $r(v_t)+|S(u_i,u_j)|$ observations.
Afterwards, Lines 12-14 perform a final correction from $r(v_t)$ to $n_f$ sampling rounds, by viewing the missing $n_f-r(v_t)$ rounds as additional observations with value zero and applying Lemma~\ref{lem:welford}, the sample mean and variance are updated as
$\widetilde{\textsf{sw}}(v_q,v_t) \gets \frac{r(v_t)}{n_f}\cdot \widetilde{\textsf{sw}}(v_q,v_t)$ and 
$\sigma(v_t)\gets \sigma(v_t)+\frac{(n_f-r(v_t))(\frac{r(v_t)}{n_f}\cdot\widetilde{\textsf{sw}}(v_q,v_t)^2-\sigma(v_t))}{n_f-1}$.
Therefore, for any $v_t \in \I$ $\widetilde{\textsf{\textnormal{sw}}}(v_q,v_t)$ and $\sigma(v_t)$ are equal to sample mean and variance over all $n_f$ observations.
This establishes that for any $v_t \in \I$, 
$\widetilde{\textsf{sw}}(v_q,v_t)$ and $\sigma(v_t)$ exactly match the sample mean and sample variance over all $n_f$ observations.
\end{proof}

\begin{proof}[\bf Proof of Theorem~\ref{lem:kASC-approx}]
If Eq.~\eqref{eq:select} holds, Algorithm~\ref{alg:KASC} invokes Algorithm~\ref{alg:subset_sample}, thereby achieving an $(\epsilon,\lambda)$-approximation with a probability of at least $1-\delta$.
Otherwise, Algorithm~\ref{alg:KASC} invokes Algorithm~\ref{alg:candidate} and Algorithm~\ref{alg:pair}, resulting in a mixed approximation guarantee.

First, Algorithm~\ref{alg:candidate} maintains the sample mean $\widetilde{\textsf{\textnormal{sw}}}(v_q, v_t)$, which is equivalent to that computed by \texttt{GNS} over $n_f$ sampled user pairs from $\U(v_q)$, by applying Lemma~\ref{lem:mean-var}. 
Based on this equivalence, we can derive a similar approximation guarantee as in Theorem~\ref{lem:naive-mc}. 
In particular, by setting $\varepsilon = n_f\cdot \epsilon \cdot \textsf{\textnormal{sw}}(v_q,v_t)$ and applying Lemma~\ref{lem:chernoff}, we obtain the following result.

\begin{align*}
& \mathbb{P}\left[\left|\widetilde{\textsf{\textnormal{sw}}}(v_q, v_t)-{\textsf{\textnormal{sw}}}(v_q, v_t)\right| \ge \epsilon\cdot {\textsf{\textnormal{sw}}}(v_q, v_t)\right] \\
& \le \exp{\left(-\frac{n_f\cdot \epsilon^2\cdot {\textsf{\textnormal{sw}}}(v_q, v_t)}{2\zeta(\frac{\epsilon}{3}+ 1)}\right)}.
\end{align*}
Note $n_f=\beta \cdot n \ge \cdot\frac{2\zeta(1+\epsilon/3)}{\epsilon^2(\lambda/\beta)}\log\frac{1}{\delta}$. Then we have
\begin{align*}
& \mathbb{P}\left[\left|\widetilde{\textsf{\textnormal{sw}}}(v_q, v_t)-{\textsf{\textnormal{sw}}}(v_q, v_t)\right| \ge \epsilon\cdot {\textsf{\textnormal{sw}}}(v_q, v_t)\right] \\
& \le \exp{\left(-\frac{{\textsf{\textnormal{sw}}}(v_q, v_t)}{\lambda/\beta}\log\frac{1}{\delta}\right)}.
\end{align*}
When ${\textsf{\textnormal{sw}}}(v_q, v_t)\ge \lambda/\beta$, we have
\begin{equation*}
\mathbb{P}\left[\left|\widetilde{\textsf{\textnormal{sw}}}(v_q, v_t)-{\textsf{\textnormal{sw}}}(v_q, v_t)\right| \ge \epsilon\cdot {\textsf{\textnormal{sw}}}(v_q, v_t)\right] \le \delta.
\end{equation*}
Applying Lemma~\ref{lem:chernoff} with $\varepsilon=n_f \cdot \epsilon\cdot \lambda/\beta$, we can derive 
\begin{align*}
& \mathbb{P}\left[\left|\widetilde{\textsf{\textnormal{sw}}}(v_q, v_t)-{\textsf{\textnormal{sw}}}(v_q, v_t)\right| \ge \epsilon\cdot \lambda/\beta \right] \\
& \le \exp{\left(-\frac{n_f\cdot (\epsilon\cdot\lambda/\beta)^2}{2\zeta(\frac{\epsilon\cdot\lambda/\beta}{3}+ {\textsf{\textnormal{sw}}}(v_q,v_t))}\right)}.
\end{align*}
When ${\textsf{\textnormal{sw}}}(v_q, v_t)< \lambda/\beta$, we have
\begin{equation*}
\mathbb{P}\left[\left|\widetilde{\textsf{\textnormal{sw}}}(v_q, v_t)-{\textsf{\textnormal{sw}}}(v_q, v_t)\right| \ge \epsilon\cdot \lambda/\beta\right] \le \delta.
\end{equation*}
Combining the two cases, for any target item $v_t\in \I$, $\widetilde{\textsf{\textnormal{sw}}}(v_q,v_t)$ computed by Algorithm~\ref{alg:candidate} is the $(\epsilon,\lambda/\beta)$-approximate Swing with a probability of at least $1-\delta$.

Afterwards, Algorithm~\ref{alg:pair} refines the borderline items $\mathcal{B}$ generated by Algorithm~\ref{alg:candidate}.
If $n_r \ge |\U(v_q)|/2$, Algorithm~\ref{alg:pair} adaptively switches to a brute-force method (Lines 5--7) to compute the exact Swing scores. Otherwise, it estimates the scores using $n_r$ user-pair samples drawn uniformly from $\U_{q,\mathcal{B}}(v_q)$ (Lines 9-12).
Since $|\U_{q,\mathcal{B}}(v_q)| \le |\U(v_q)|$, the random variable
$Z_r = \frac{|\U_{q,\mathcal{B}}(v_q)|\cdot(|\U_{q,\mathcal{B}}(v_q)|-1)}{\alpha+|\I_{i,j}|}$
is still bounded by
$\zeta = \frac{|\U(v_q)|\cdot(|\U(v_q)|-1)}{\alpha+2}.$
Therefore, the same argument used for Algorithm~\ref{alg:candidate} applies here, with the sample size replaced by $n_r=(1-\beta)n.$
By Lemma~\ref{lem:chernoff}, for any $v_t\in\mathcal{B}$, the estimator produced in Lines 9-12 is an $(\epsilon,\lambda/(1-\beta))$-approximation of ${\textsf{\textnormal{sw}}}(v_q,v_t)$ with probability at least $1-\delta$.
When $n_r < |\U(v_q)|/2$, Algorithm~\ref{alg:pair} performs $n_r$ random sampling steps (Lines 14-17) from $\U(v_q)$, with a more efficient membership check. The same analysis applies, implying that $\widetilde{\textsf{\textnormal{sw}}}(v_q,v_t)$ is an $(\epsilon,\lambda/(1-\beta))$-approximation of ${\textsf{\textnormal{sw}}}(v_q,v_t)$ with probability at least $1-\delta$.
Therefore, for any target item $v_t\in \mathcal{B}$, $\widetilde{\textsf{\textnormal{sw}}}(v_q,v_t)$ output by Algorithm~\ref{alg:pair} is the $(\epsilon,\lambda/(1-\beta))$-approximate Swing with a probability of at least $1-\delta$.

Putting the above together, Algorithm~\ref{alg:KASC} provides a mixed approximation guarantee. Specifically, 
for any target item $v_t\in \I \setminus \mathcal{B}$, $\widetilde{\textsf{\textnormal{sw}}}(v_q,v_t)$ is the $(\epsilon,\lambda/\beta)$-approximate Swing with a probability of at least $1-\delta$, and for any target item $v_t\in \mathcal{B}$, $\widetilde{\textsf{\textnormal{sw}}}(v_q,v_t)$ is the $(\epsilon,\lambda/(1-\beta))$-approximate Swing with a probability of at least $1-\delta$.
\end{proof}

\end{document}